\documentclass[acmsmall,nonacm]{acmart}

\acmJournal{PACMPL}
\acmVolume{1}
\acmNumber{CONF} 
\acmArticle{1}
\acmYear{2018}
\acmMonth{1}
\acmDOI{} 
\startPage{1}

\setcopyright{none}

\usepackage{booktabs}   
\usepackage{subcaption} 
\usepackage{algorithm}
\usepackage{algorithmicx}
\usepackage[noend]{algpseudocode}
\usepackage{amsmath}
\usepackage{amsthm}
\usepackage{amssymb}
\usepackage[]{apxproof}		
\usepackage{colortbl}
\usepackage{galois}
\usepackage{multirow}
\usepackage{pgfplots}
\usepackage{stmaryrd}
\usepackage{symbols}
\usepackage{tikz}
\usepackage{xcolor}
\definecolor{mypurple}{RGB}{96,0,32}
\definecolor{myred}{RGB}{255,0,0}
\definecolor{myorange}{RGB}{255,128,0}
\definecolor{myyellow}{RGB}{252,210,5}
\definecolor{mygreen}{RGB}{0,178,0}
\usepackage{xspace}
\usepackage{wasysym}
\usepackage{flushend}

\MakeRobust{\Call}
\algnewcommand\algorithmicswitch{\textbf{switch}}
\algnewcommand\algorithmiccase{\textbf{case}}
\algdef{SE}[SWITCH]{Switch}{EndSwitch}[1]{\algorithmicswitch\
	#1\ \algorithmicdo}{\algorithmicend\
	\algorithmicswitch}%
\algdef{SE}[CASE]{Case}{EndCase}[1]{\algorithmiccase\
	#1\ :}{\algorithmicend\ \algorithmiccase}%
\algtext*{EndSwitch}%
\algtext*{EndCase}%
\def\showmark#1{\tikz[baseline=-2.75pt]\node[black,
 scale=1.5]{\pgfuseplotmark{#1}};}
\def\testmark#1#2{\tikz[baseline=-2.75pt]\node[black,mark
 color=#1,scale=1.5]{\pgfuseplotmark{#2}};}
\def\usemark#1{\tikz[baseline=-2.25pt]\node[black,
	scale=1.5]{\pgfuseplotmark{#1}};}

\theoremstyle{remark} 
\newtheorem*{remark}{Remark} 
\theoremstyle{definition}
\newtheorem{theorem}{Theorem}[section]
\newtheorem{lemma}[theorem]{Lemma}
\newtheorem{corollary}[theorem]{Corollary}
\newtheorem{example}[theorem]{Example}
\renewcommand{\appendixsectionformat}[2]{Supplementary
	Material for Section~#1 (#2)}

\begin{document}

\title[Perfectly Parallel Fairness 
Certification of Neural Networks]{Perfectly Parallel 
Fairness 
Certification\\of Neural Networks}


\author{Caterina Urban}
\affiliation{
	\institution{INRIA}
}
\affiliation{
	\institution{DIENS, \'Ecole Normale Sup\'erieure, 
	CNRS, PSL University}
	\city{Paris}
	\country{France}
}
\email{caterina.urban@inria.fr}

\author{Maria Christakis}
\affiliation{
	\institution{MPI-SWS}
	\country{Germany}
}
\email{maria@mpi-sws.org}

\author{Valentin W\"ustholz}
\affiliation{
	\institution{ConsenSys Diligence}
	\country{Germany}
}
\email{valentin.wustholz@consensys.net}

\author{Fuyuan Zhang}
\affiliation{
	\institution{MPI-SWS}
	\country{Germany}
}
\email{fuyuan@mpi-sws.org}

\begin{abstract}
Recently, there is growing concern that 
machine-learning models, which
currently assist or even automate decision 
making, reproduce, and in
the worst case reinforce, bias of the training data. 
The development
of tools and techniques for certifying fairness of 
these models or
describing their biased behavior is, therefore, 
critical.
%
In this paper, we propose a \emph{perfectly parallel} 
static analysis for
certifying \emph{causal fairness} of feed-forward 
neural
networks used for classification of tabular data. When 
certification succeeds, our approach provides 
definite
guarantees, otherwise, it describes and quantifies the biased
behavior. We design the analysis to be 
\emph{sound}, in 
practice also \emph{exact},
and configurable in terms of scalability and precision, thereby
enabling \emph{pay-as-you-go certification}. We 
implement our approach in an
open-source tool and demonstrate its effectiveness on models trained
with popular datasets.
\end{abstract}



\maketitle

\section{Introduction}

Due to the tremendous advances in machine learning and the vast
amounts of available data, software systems, and neural networks in
particular, are of ever-increasing importance in our everyday
decisions, whether by assisting them or by autonomously making them.
We are already witnessing the wide adoption and societal impact of
such software in criminal justice, health care, and social welfare, to
name a few examples. It is, therefore, not far-fetched to imagine a
future where most of the decision making is automated.

However, several studies have recently raised concerns about the
fairness of such systems. For instance, consider a commercial
recidivism-risk assessment algorithm that was found racially
biased~\cite{COMPAS}. Similarly, a commercial algorithm that is widely
used in the U.S. health care system falsely determined that Black
patients were healthier than other equally sick patients by using
health costs to represent health needs~\cite{ObermeyerPowers2019}.
There is also empirical evidence of gender bias in image searches, for
instance, there are fewer results depicting women when searching for
certain occupations, such as CEO~\cite{KayMatuszek2015}. Commercial
facial recognition algorithms, which are increasingly used in law
enforcement, are less effective for women and darker skin
types~\cite{BuolamwiniGebru2018}.

In other words, machine-learning software may reproduce, or even
reinforce, bias that is directly or indirectly present in the training
data. This awareness 
will certainly lead to regulations
and strict audits in the
future. It is, therefore, critical to develop tools and techniques for
certifying fairness of neural networks and understanding the
circumstances of their potentially biased behavior.

\paragraph{\textbf{Causal Fairness.}} We make a 
step forward in meeting
these needs by designing a static analysis framework for certifying
\emph{causal fairness}~\cite{Galhotra17} of feed-forward neural
networks used for classification tasks. Specifically,
given a choice (e.g., driven by a causal model) of 
input features that are 
considered (directly or indirectly) sensitive to bias, 
\emph{a neural
  network is causally fair if the output classification is not
  affected by different values of the chosen 
  features}. Note that,
unlike local robustness of neural networks, causal fairness is a
\emph{global} property, which is evaluated with respect to all inputs,
instead of only those within a particular distance metric.

Of course, the most obvious approach to avoid such bias is to remove
any sensitive feature from the training data, called fairness through
unawareness~\cite{GrgicHlaca16}.  However, this 
does not work for three main
reasons. First, neural networks learn from latent variables (e.g.,
\cite{LumIsaac2016,UdeshiAC18}). For instance, a credit-screening
algorithm might not use race (or gender) as an explicit input but
still be biased with respect to it, say, by using the ZIP code of
applicants as proxy for race (or their first name as proxy for
gender). Therefore, simply removing a sensitive feature does not
necessarily free the training data or the corresponding neural network
from bias.  Second, the training data is only a relatively small
sample of the entire input space, on portions of which the neural
network might end up being inaccurate. For example, if women are
underrepresented in the training data, a credit-screening algorithm is
less likely to be accurate for them. Third, the information provided
by a sensitive feature might be necessary, for instance, to introduce
intended bias in a certain input region. Assume a credit-screening
algorithm that should not discriminate with respect to age unless it
is above a particular threshold. Above this age threshold, the higher
the requested credit amount, the lower the chances of receiving it. In
such cases, removing the sensitive feature is not even possible.

\paragraph{\textbf{Our Approach.}}
Verification of global neural-network properties, such as causal
fairness, is still a long way from being practical (see
Section~\ref{sect:relatedWork}). In this paper, we 
propose an approach
that brings us closer to this aspiration.
Our approach certifies causal fairness
of neural networks used for \emph{classification of tabular data} by employing a
combination of a forward and a backward static analysis. On a high
level, the forward pass aims to reduce the overall analysis effort. At
its core, it divides the input space of the network into independent
partitions.  The backward analysis then attempts to certify fairness
of the classification within each partition (in a \emph{perfectly
  parallel} fashion) with respect to a chosen (set of)
feature(s), which may be directly or indirectly
sensitive, for instance, race or ZIP code. In the end, our approach
reports for which regions of the input space the neural network is
proved fair and for which there is bias. Note that we do not
necessarily need to analyze the entire input space; our technique is
also able to answer specific bias queries about a fraction of the
input space, e.g., are Hispanics over 45 years old
discriminated against with respect to gender?

The scalability-vs-precision tradeoff of our 
approach is configurable.
Partitions that
do not satisfy the given configuration are excluded 
from the analysis and
may be resumed later, with a more flexible configuration. This enables
usage scenarios in which our approach adapts to the available
resources, e.g., time or CPUs, and is run 
incrementally. In other
words, we designed a \emph{pay-as-you-go certification} approach that
the more resources it is given, the larger the region of the input
space it is able to analyze.

\paragraph{\textbf{Related Work.}} In the 
literature, related work on
determining fairness of machine-learning models has focused on
providing probabilistic guarantees~\cite{Bastani18}. In contrast, our
approach gives definite guarantees for those input partitions that
satisfy the analysis configuration. Similarly to our approach, there
is work that also aims to provide definite
guarantees~\cite{Albarghouthi17b} (although for different fairness
criteria). However, it has been shown to scale only up to neural
networks with two hidden neurons. Our approach is significantly more
scalable since its design enables perfectly parallel fairness
certification of each input partition.

\paragraph{\textbf{Contributions.}} We make the 
following contributions:
\begin{enumerate}
\item We propose a perfectly parallel static 
analysis approach
  for certifying causal fairness of feed-forward 
  neural networks used for classification of tabular data. 
  If certification fails, our approach
  can describe and quantify the biased input space 
  region(s). 

\item We show that our approach is sound and, in 
practice, exact for the analyzed
  regions of the input space.

\item We discuss the configurable scalability-vs-precision tradeoff of
  our approach that enables pay-as-you-go certification.
 
\item We implement our approach in an open-source tool called \tool
  and evaluate it on neural networks trained with popular datasets.
  We show the effectiveness of our approach in detecting injected bias
  and answering bias queries. We also experiment with the precision
  and scalability of the analysis and discuss the 
  tradeoffs.
\end{enumerate}


%
%

\section{Overview}
\label{sec:Overview}

In this section, we give an overview of our approach using a small
constructed example, which is shown in Figure~\ref{fig:toy}.

\paragraph{\textbf{Example.}} The figure depicts a 
feed-forward 
neural
network for credit approval. There are two inputs $\node_{0,1}$ and
$\node_{0,2}$ (shown in purple). Input $\node_{0,1}$ denotes the
requested credit amount and $\node_{0,2}$ denotes age. Both inputs
have continuous values in the range $[0,1]$. Output $\node_{3,2}$
(shown in green) denotes that the credit request is 
approved,
whereas $\node_{3,1}$ (in red) denotes that it is denied. The neural
network also consists of two hidden layers with two nodes each
(in gray).

Now, let us assume that this neural network is 
trained to deny requests for large credit amounts 
from older people. Otherwise, the network does not 
discriminate with respect to age
for small credit amounts. There is also no bias for 
younger people with respect to the requested credit.
When choosing age as the sensitive input, our 
approach can certify fairness with 
respect to different age groups for small credit amounts. Our approach
is also able to find (as well as quantify) bias with 
respect to age for large credit amounts. 
Note that this bias may be intended or accidental --- our analysis
does not aim to address this question.

Our approach does not require age to be an explicit input of the
neural network. For example, $\node_{0,2}$ could denote the ZIP code
of credit applicants, and the network could still use it as proxy for
age. That is, requests for large credit amounts are denied for a
certain range of ZIP codes (where older people tend to live), yet
there is no discrimination between ZIP codes for small credit
amounts. When choosing the ZIP code as the 
sensitive input, our
approach would again be able to detect bias with respect to it for
large credit amounts.

Below, we present on a high level how our approach 
achieves these
results.

\begin{figure}[t]
	\includegraphics[width=0.5\columnwidth]{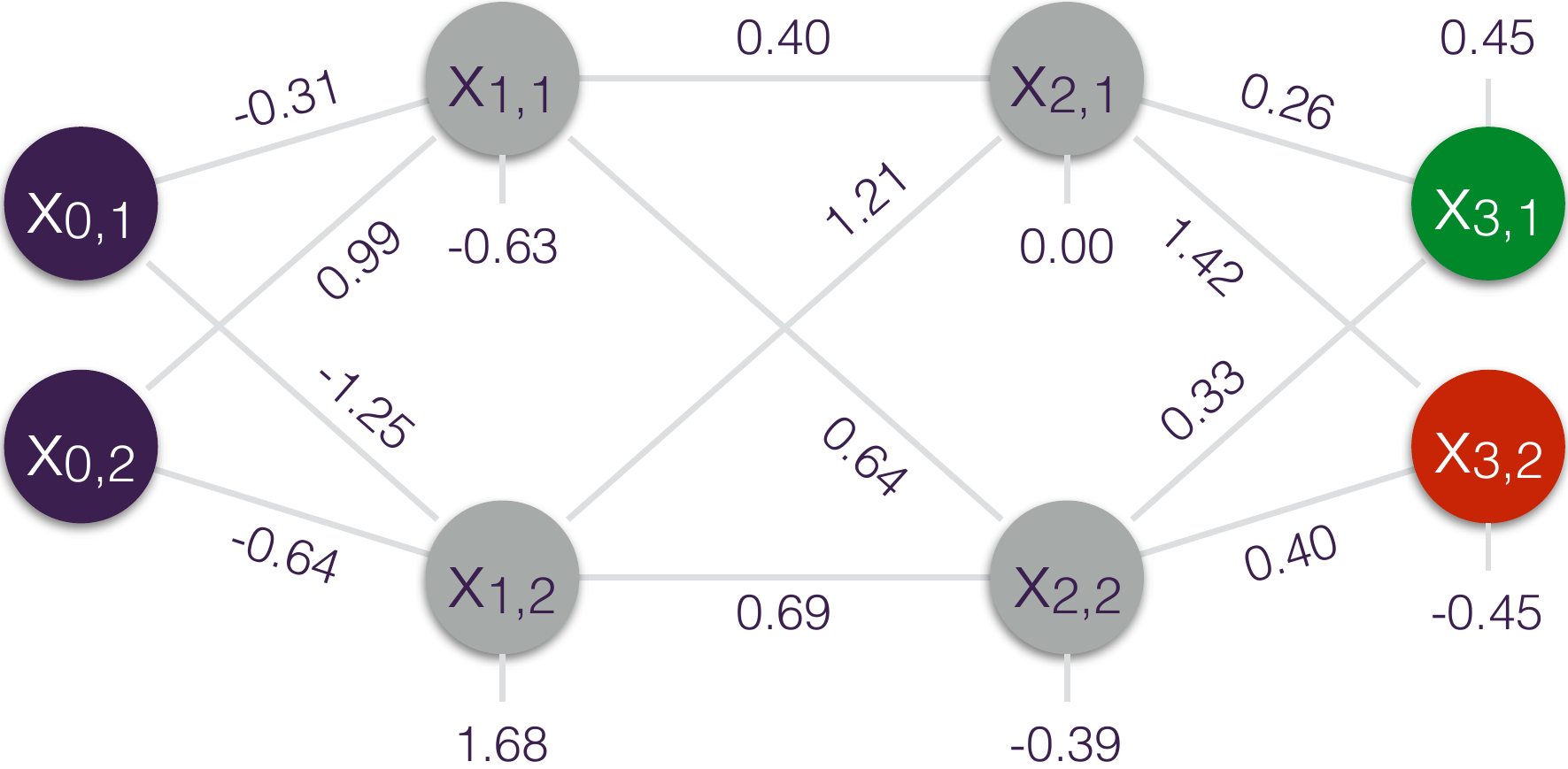}
	\caption{Small, constructed example of trained 
	feed-forward neural
          network for credit approval. 
      }
        \label{fig:toy}
        \vspace{-1em}
\end{figure}

\paragraph{\textbf{Na\"ive Approach.}} In 
theory, the 
simplest way to certify causal fairness  
is to first analyze the neural network 
backwards 
starting from each
output node, in our case $\node_{3,1}$ and 
$\node_{3,2}$. This 
allows us to determine the regions of the input 
space (i.e., age and
requested credit amount) for which credit is 
approved and denied.  For
example, assume that we find that requests are 
denied for credit amounts larger than $10~000$ 
(i.e., $10~000 < \node_{0,1}$) and age greater than 
$60$ 
(i.e., $60 < \node_{0,2}$), while they are approved 
for 
$\node_{0,1} \leq 10~000$ and $60 < \node_{0,2}$
or for $\node_{0,2} \leq 60$.

The second step is to forget the value of the 
sensitive input (i.e.,
age) or, in other words, to project these regions 
over the credit amount.
In our example, after projection 
we have that credit requests are denied for $10~000 
< \node_{0,1}$ and approved for any value of 
$\node_{0,1}$.
A non-empty intersection between the 
projected input regions indicates bias with respect 
to the sensitive input. In our example, the 
intersection is non-empty for $10~000 
< \node_{0,1}$: there exist people 
that differ in age but request the same credit 
amount (greater than $10~000$), some of
whom receive the credit while others do not.

This approach, however, is not practical. 
Specifically, neural
networks with \relu{} activation 
functions (see
Section~\ref{sec:neuralnets} for more details, other 
activation functions are discussed in 
Section~\ref{sec:analysis}), each 
hidden node effectively represents a disjunction 
between two activation statuses (active and 
inactive). In our example, there are $2^4$ possible 
activation patterns for the $4$ hidden nodes. To 
retain 
maximum precision, the analysis
would have to explore all of them, which does not 
scale in practice.

\paragraph{\textbf{Our Approach.}}
Our analysis is based on the observation that 
\emph{there might exist many activation
  patterns that do not correspond to a region of the input
  space} \cite{Hanin19}. Such patterns can, 
  therefore, be ignored 
  during the analysis. 
 We push this idea further by defining
\emph{abstract activation 
patterns}, which fix the activation status of only 
certain nodes and thus represent sets of (concrete) 
activation patterns.
Typically, \emph{a relatively
small number of abstract activation patterns is 
sufficient for
covering the entire input space}, without 
necessarily 
representing and exploring all
possible concrete patterns.

Identifying those patterns that definitely correspond to a region of
the input space is only possible with a forward analysis. Hence, we
combine a forward pre-analysis with a backward 
analysis. The pre-analysis
partitions the input space into independent 
partitions corresponding to
abstract activation patterns. Then, the backward analysis tries to
prove fairness of the neural network for each such partition.

More specifically, we set an upper bound 
$\upperbound$ on the
number of tolerated disjunctions (i.e., on the number of nodes with
an unknown activation status) per abstract 
activation pattern. Our forward pre-analysis uses a
cheap abstract
domain (e.g., the boxes domain~\cite{Cousot76}) to 
\emph{iteratively}
partition the input space along the 
\emph{non-sensitive} input dimensions to obtain 
\emph{fair} input partitions (i.e., boxes).
Each partition 
satisfies one of
the following conditions: (a)~its classification 
is already 
fair because only
one network output is reachable for all inputs in the 
region, (b)~it
has an abstract activation pattern with at most $\upperbound$ unknown
nodes, or (c)~it needs to be partitioned further. 
We call
partitions that satisfy condition (b) \emph{feasible}.

In our example, let $\upperbound = 2$. At first, the 
analysis considers the
entire input space, that is, $\node_{0,1}: [0,1]$ (credit amount) and
$\node_{0,2}: [0,1]$ (age). (Note that we could also specify
a part of the input space for analysis.)
The abstract activation pattern corresponding to 
this initial partition 
$\partition$ is 
$\epsilon$ (i.e., no hidden nodes have fixed 
activation status) and, thus, the number of 
disjunctions would be 4, which is greater than 
$\upperbound$.
Therefore, $\partition$ needs to be divided 
into $\partition_1$ 
($\node_{0,1}:
[0,0.5].  \node_{0,2}: [0,1]$) and $\partition_2$ ($\node_{0,1}:
[0.5,1].  \node_{0,2}: [0,1]$). Observe that the input space is not split
with respect to $\node_{0,2}$, which is the sensitive
input. Now, $\partition_1$ is feasible since its 
abstract
activation pattern is 
$\node_{1,2}\node_{2,1}\node_{2,2}$ (i.e., 3
nodes are always active), 
while
$\partition_2$ must be divided further since its 
abstract
activation pattern is $\epsilon$.

To control the number of partitions, we impose a 
lower bound
$\lowerbound$ on the size of each of their dimensions. Partitions that
require a dimension of a smaller size are \emph{excluded}. In other
words, they are not considered until more analysis \emph{budget}
becomes available, that is, a larger $\upperbound$ or a smaller
$\lowerbound$.

In our example, let
$\lowerbound = 0.25$. 
The forward pre-analysis further divides
$\partition_2$ into $\partition_{2,1}$ ($\node_{0,1}: [0.5,0.75].
\node_{0,2}: [0,1]$) and $\partition_{2,2}$ ($\node_{0,1}: [0.75,1].
\node_{0,2}: [0,1]$). Now, $\partition_{2,1}$ is 
feasible, with
abstract pattern $\node_{1,2}\node_{2,1}$, while 
$\partition_{2,2}$ is
not. 
However, $\partition_{2,2}$ may not be split further 
because the
size of the only non-sensitive dimension $\node_{0,1}$ has already
reached the lower bound $\lowerbound$. As a result, $\partition_{2,2}$
is excluded, and only the remaining $75\%$ of the 
input
space is considered for analysis. 

Next, feasible input 
partitions (within bounds $\lowerbound$ and 
$\upperbound$) are grouped by abstract 
activation patterns. In our example, the
pattern corresponding to $\partition_1$, namely
$\node_{1,2}\node_{2,1}\node_{2,2}$, is subsumed 
by the (more abstract) pattern of
$\partition_{2,1}$, namely $\node_{1,2}\node_{2,1}$. Consequently, we
group $\partition_1$ and $\partition_{2,1}$ under
pattern $\node_{1,2}\node_{2,1}$.

The backward analysis is then run \emph{in 
parallel} for each representative abstract
activation pattern, 
in our
example $\node_{1,2}\node_{2,1}$. This analysis 
determines the
region of the input space (within a given partition group) for which
each output of the neural network is returned, e.g., credit is
approved for $c_1 \leq \node_{0,1} \leq c_2$ and $a_1 \leq \node_{0,2}
\leq a_2$. To achieve this, the analysis uses an expensive abstract
domain, for instance, disjunctive or powerset
polyhedra~\cite{Cousot78,Cousot79}, and leverages abstract activation
patterns to avoid disjunctions. For instance, pattern
$\node_{1,2}\node_{2,1}$ only requires reasoning 
about two disjunctions
from the remaining hidden nodes $\node_{1,1}$ and $\node_{2,2}$.

Finally, fairness is checked for each partition in the same way that
it is done by the na\"ive approach for the entire input space. In our
example, we prove that the classification within 
$\partition_1$ is fair and determine that within 
$\partition_{2,1}$ the classification is biased. 
Concretely, our approach determines that
bias occurs for $0.54 \leq \node_{0,1} \leq 0.75$, which corresponds
to $21\%$ of the entire input space (assuming a uniform probability distribution).
In other words, the network
returns different outputs for people that request the same credit in
the above range but differ in age. Recall that partition
$\partition_{2,2}$, where $0.75 \leq \node_{0,1} \leq 1$, was excluded
from analysis, and therefore, we cannot draw any conclusions about
whether there is any bias for people requesting credit in this range.

Note that bias may also be quantified according to a 
probability
distribution of the input space. In particular, it might be that
credit requests in the range $0.54 \leq \node_{0,1} \leq 0.75$ are
more (resp. less) common in practice. 
Given their probability distribution,
our analysis computes a tailored percentage of bias, 
which in this
case would be greater (resp. less) than 21\%.


\section{Feed-Forward Deep Neural
Networks}\label{sec:neuralnets}

%

Formally, a \emph{feed-forward deep neural
network}
consists of an input layer ($\layer_0$), an output
layer ($\layer_\layers$), and a number of hidden
layers ($\layer_1, \dots, \layer_{\layers-1}$) in
between.
Each layer $\layer_i$ contains $\size{\layer_i}$
nodes
and, with the exception of the input layer,
is associated to a $\size{\layer_i} \times
\size{\layer_{i-1}}$-matrix $\weights_i$ of weight
coefficients and
a vector $\biases_i$ of $\size{\layer_i}$ bias
coefficients.
In the
following, we use $\nodes$ to denote the set of all
nodes, $\nodes_i$ to denote the set of nodes of
the $i$th layer, and $\node_{i,j}$ to
denote the $j$th node of the $i$th
layer of a neural network. We focus here on neural
networks used for \emph{classification} tasks.
Thus, $\size{\layer_\layers}$ is the number of
target classes (e.g., 2 classes in
Figure~\ref{fig:toy}).

The value of the input nodes is given by the input
data: continuous data is represented by one input
node (e.g., $x_{0,1}$ or $x_{0,2}$ in
Figure~\ref{fig:toy}), while categorical data is
represented by multiple input nodes via one-hot
encoding.
In the following, we use $\sensitive$ to denote the
subset of input nodes considered (directly or indirectly) \emph{sensitive} 
to
bias (e.g., $x_{0,2}$ in
Figure~\ref{fig:toy}) and $\overline{\sensitive} \defined 
\nodes_0 \setminus \sensitive$ to denote the input 
nodes not deemed sensitive to bias.

The value of each hidden
and output node $\node_{i,j}$ is computed by an
\emph{activation function} $f$ applied to
a linear
combination of the values of all nodes in
the preceding layer \cite{Goodfellow16}, i.e., 
$\node_{i,j} =
f\left(\sum^{\size{\layer_{i-1}}}_k
\weight^{i}_{j,k}
\cdot
\node_{i-1,k} + \bias_{i,j}\right)$,
where $\weight^{i}_{j,k}$ and $\bias_{i,j}$ are
weight and
bias coefficients in $\weights_i$ and $\biases_i$,
respectively.
In a \emph{fully-connected neural
	network}, all $\weight^{i}_{j,k}$ are
non-zero. Weights and biases are
adjusted during the \emph{training phase} of the
neural
network. In what follows, we focus on already
trained neural networks, which we call
\emph{neural-network models}.

Nowadays, the most commonly used activation
for hidden nodes is
the Rectified Linear Unit (\relu) \cite{NairH10}:
$\relu(x) = \max(x, 0)$.
In this case, the activation used for output
nodes is the identity function. The output
values are then normalized into a probability
distribution on the target classes
\cite{Goodfellow16}. We discuss other activation
functions in Section~\ref{sec:analysis}.

\section{Trace Semantics}\label{sec:traces}

Our approach expresses neural-network models as programs. These
programs consist of assignments for computing the activation value of
each node (e.g., $x_{1,1} = -0.31*x_{0,1} + 0.99*x_{0,2} - 0.63$ in
Figure~\ref{fig:toy}) and implementations of activation functions
(e.g., if-statements for $\relu$s). As is standard practice in static
program analysis, we define a semantics for these programs and use it
to prove soundness of our approach.

The \emph{semantics} of a neural-network model is
a mathematical characterization
of its behavior when executed for all
possible input data. We model the operational
semantics of a feed-forward neural-network
model $\model$ as a transition
system $\operational$,
where $\states$ is a (potentially infinite)
set of states and the \emph{acyclic} transition
relation
$\transitions \subseteq
\states \times \states$ describes the possible
transitions between states \cite{Cousot02,Cousot77}.

More specifically, a state $\s \in \states$ maps
neural-network nodes to their values. Here, for
simplicity, we assume that nodes
have real values, i.e.,
$\function{\s}{\nodes}{\reals}$.
(We discuss 
floating-point values in
Section~\ref{sec:analysis}.)
In the following, we
often only care about the values of a subset of the
neural-network nodes in certain states.
Thus, let
$\restrict{\states}{Y} \defined \set{
	\restrict{\s}{Y} \mathrel{\big|} \s \in
	\states}$ be the restriction of
$\states$ to a domain of interest $Y$.
Sets $\restrict{\states}{\nodes_0}$ and
$\restrict{\states}{\nodes_\layers}$ denote
restrictions
of $\states$ to the network nodes in the
input and output layer, respectively.
With a slight abuse of notation, let $\nodes_{i,j}$
denote $\restrict{\states}{\set{x_{i,j}}}$,
i.e., the restriction of $\states$ to
the singleton set containing $\node_{i,j}$.
Transitions happen between states
with
different
values for consecutive nodes in the same layer, i.e.,
$\transitions \subseteq \nodes_{i,j} \times
\nodes_{i,j+1}$,
or
between states with different values for the last and
first node of consecutive layers of
the network, i.e.,
$\transitions \subseteq \nodes_{i,\size{\layer_i}}
\times
\nodes_{i+1,0}$.
The set $\final \defined \set{\s \in \states \mid
\forall \s'
	\in \states\colon \tuple{s}{s'} \not\in \tau}$ is
	the set of final states of the neural network.
These are partitioned in a set of outcomes
$\classes
\defined \set{
\set{
	\s \in \final
	\mid
	\max \nodes_\layers = \node_{\layers, i}
%
}
\mid 0 \leq i \leq \size{\layer_{\layers}}
}$, depending on the output node
with the highest value (i.e., the target class with
highest probability).

Let 
$\states^n \defined \set{ \s_0 \cdots \s_{n-1}
\mid
\forall i < n\colon \s_i \in \states }$ 
be the set of 
all
sequences of exactly $n$ 
states 
in 
$\states$. 
Let $\states^+ \defined \bigcup_{n \in \nat^+}
\states^n$ be the set of
all non-empty finite
sequences of 
states. 
A 
\emph{trace} 
is a 
sequence of states that respects the 
transition relation $\transitions$, that is, 
$\langle \s,\s' \rangle \in \transitions$
for each pair of consecutive states
$\s, \s'$ in the sequence. We write 
$\overline{\states}^n$ for the set of all traces of 
$n$ states: 
$\overline{\states}^n \defined \set{ \s_0 \cdots 
	\s_{n-1} \in \states^n
	\mid
	\forall i < n - 1\colon \tuple{\s_i}{s_{i+1}} \in 
	\tau  }$. 
%
The
\emph{trace semantics}
$\Upsilon \in \semantics$
generated by a transition system $\operational$ is
the set of all non-empty traces terminating in 
$\final$ 
\cite{Cousot02}:
\begin{equation}\label{eq:maximal}
\Upsilon \defined \bigcup_{n \in \nat^+}
\set{	
	\s_0 \dots \s_{n-1} \in \overline{\states}^n \mid 
	\s_{n-1} \in 
	\final}
\end{equation}
In the rest of the paper, we write
$\maximal{\model}$ to denote the trace semantics
of a neural-network model $\model$.

The trace semantics fully describes the behavior of 
$\model$. However, reasoning about a particular 
property of $\model$ does not need all this 
information and, in fact, is facilitated by the design 
of a semantics that abstracts away from irrelevant 
details about $\model$'s behavior. In the following 
sections, we formally define our property of 
interest, causal fairness, and systematically derive, 
using \emph{abstract interpretation}~\cite{Cousot77}, a 
semantics tailored to reasoning about this 
property.

\section{Causal Fairness}\label{sec:usage}

A \emph{property} is specified by its extension, that is, by the set
of elements having such a property
\cite{Cousot77,Cousot79}. Properties of neural-network models are
properties of their semantics.  Thus, properties of network
models with trace semantics in $\semantics$ are sets of sets of traces
in $\properties$.  In particular, the set of neural-network properties
forms a complete boolean lattice $\langle \properties, \subseteq,
\cup, \cap, \emptyset, \semantics \rangle$ for subset inclusion, that
is, logical implication. The strongest property is the
standard \emph{collecting semantics} $\Lambda \in \properties$:
\begin{equation}\label{eq:collecting}
\Lambda \defined \set{\Upsilon}
\end{equation}
Let $\collecting{\model}$ denote the collecting
semantics of a particular neural-network model
$\model$. Then, model $\model$ satisfies a
given property $\property$ if and only if its
collecting semantics is a subset of
$\property$:
\begin{equation}\label{eq:validation}
\model \models \property
\Leftrightarrow
\collecting{\model} \subseteq \property
\end{equation}

Here, we consider the property of \emph{causal fairness}, which
expresses that the classification determined by a network model does
not depend on sensitive input data.  In particular, the property might
interest the classification of all or just a fraction of the input
space.

More formally, let $\choices$ be the set of all
possible value choices for all sensitive input nodes
in $\sensitive$, e.g., for $\sensitive = 
\set{\node_{0,i},
\node_{0,j}}$
one-hot encoding, say,
gender information,
$\choices = \set{ \set{1, 0}, \set{0, 1} }$; for 
$\sensitive = \set{x_{0,k}}$
encoding continuous data, say, in the range $[0,
1]$, a possibility is $\choices = \set{[0, 0.25],
[0.25, 0.75], [0.75, 1]}$.
In the following, given a trace $\trace \in
\semantics$, we write
$\trace_0$ and $\trace_\omega$ to denote its
initial and final state, respectively.
We also write $\trace_0 =_{\overline{\sensitive}} 
\trace'_0$
to
indicate that the states $\trace_0$ and $\trace'_0$
agree on all values of all non-sensitive input nodes, and 
$\trace_\omega \equiv
\trace'_\omega$ to indicate
that $\trace$ and $\trace'$ have the same outcome
$\class \in \classes$. We
can now formally define when the sensitive input
nodes in $\sensitive$ are \emph{unused} with
respect to a set of traces $T \in \semantics$
\cite{Urban18}. For one-hot encoded 
sensitive inputs\footnote{For continuous sensitive 
inputs, 
we can replace 
	$\trace_0(\sensitive) \neq \choice$ (resp. 
	$\trace_0(\sensitive) = \choice$) with 
	$\trace_0(\sensitive) \not\subseteq \choice$ 
	(resp. 
	$\trace_0(\sensitive) \subseteq \choice$).} 
we have
\begin{equation}\label{eq:unused}
\begin{aligned}
&\unused_\sensitive(T) \defined
\forall \trace \in T, \choice \in
\choices\colon
\trace_0(\sensitive) \neq \choice
\mathrel{\Rightarrow} 
\exists \trace' \in T\colon
\trace_0 =_{\overline{\sensitive}} \trace'_0 
\mathrel{\land}
\trace'_0(\sensitive) = \choice \mathrel{\land}
\trace_\omega \equiv \trace'_\omega,
\end{aligned}
\end{equation}
where $\trace_0(\sensitive) \defined 
\set{\trace_0(\node) \mid \node \in \sensitive}$ 
is
the image of $\sensitive$ under $\trace_0$.
Intuitively, the sensitive input nodes in $\sensitive$
are unused if any possible outcome in $T$ (i.e., any
outcome $\trace_\omega$ of any trace
$\trace$ in $T$) is possible from all possible value
choices for $\sensitive$ (i.e., there exists a trace
$\trace'$ in $T$ for each value choice for
$\sensitive$ with the same outcome as $\trace$).
That is, each outcome is independent of the
value choice for $\sensitive$.

\begin{example}\label{ex:unused}
Let us consider again our example in 
Figure~\ref{fig:toy}. We write $\tuple{c}{a} 
\rightsquigarrow o$ for a trace starting in a 
state with $\node_{0,1} = c$ 
and 
$\node_{0,2} = a$ 
and ending in a state where 
$o$ is the node with the highest value (i.e., the 
output class). The sensitive input $\node_{0,2}$ 
(age) is 
\emph{unused} in $T = \set{
\tuple{0.5}{a} \rightsquigarrow \node_{3,2} \mid 0 
\leq a \leq 1
}$. It is instead \emph{used} in $T' = 
\set{
\tuple{0.75}{a} \rightsquigarrow \node_{3,2} \mid 0 
\leq a < 0.51
} \cup \set{ \tuple{0.75}{a} \rightsquigarrow 
\node_{3,1} \mid 0.51 
\leq a \leq 1
}$.
\end{example}


The causal-fairness property $\fairness$ can now be defined
as $\fairness \defined \set{ \maximal{M}
	\mid
	\unused_\sensitive(\maximal{M}) }$,
that is, as the set of all neural-network
models (or rather, their semantics) that do not use
the values of the sensitive input nodes for
classification. In practice, the
property might interest just a fraction of the input
space, i.e., we define
\begin{equation}\label{eq:fairness}
\fairness[Y]
\defined \set{
	\maximal{\model}^Y
	\mid
	\unused_\sensitive(\maximal{\model}^Y) },
\end{equation}
	where $Y \in \powerset{\states}$ is a set of initial
	states of interest and the restriction
	$T^Y \defined \set{ \trace
	\in T \mid \trace_0 \in Y}$ only
	contains traces of $T \in \semantics$
	that start
	with a
	state in $Y$.
Similarly, in the rest of the paper, we write $S^Y 
\defined 
\set{ T^Y 
\mid T \in S}$ for the set of sets of traces
restricted to initial states in $Y$.
Thus, from Equation ~\ref{eq:validation}, we have 
the following:
\begin{theorem}\label{thm:collecting}
	$\model \models \fairness[Y] \Leftrightarrow
	\collecting{\model}^Y \subseteq
	\fairness[Y]$
\end{theorem}

\begin{proof}
	The proof follows trivially from 
	Equation~\ref{eq:validation} and the definition of 
	$\fairness[Y]$ (cf. Equation~\ref{eq:fairness}) and 
	$\collecting{\model}^Y$.
\end{proof}

\section{Dependency
Semantics}\label{sec:dependency}

We now
use
abstract
interpretation to
systematically derive, by successive abstractions of
the collecting semantics $\Lambda$,
a
\emph{sound and complete}
semantics $\Lambda_\dependency$ that contains
only and exactly the
information needed to reason about $\fairness[Y]$.

\subsection{Outcome
Semantics}\label{subsec:outcome}

Let $T_Z \defined \set{ \trace
	\in T \mid \trace_\omega \in
	Z}$
be the set of traces of $T \in \semantics$
that end with a state in $Z \in \powerset{\states}$.
As before, we write $S_Z
\defined \set{
T_Z 
\mid T \in S}$ for the set of sets of traces restricted to final states in $Z$.
From the definition of
$\fairness[Y]$ (and in particular, from the
definition of $\unused_\sensitive$, cf.
Equation~\ref{eq:unused}), we have: 
\begin{lemma}\label{lem:outcome}
	$\collecting{\model}^Y \subseteq \fairness[Y] 
	\Leftrightarrow \forall \class \in \classes\colon 
	\collecting{\model}^Y_\class \subseteq
	\fairness[Y]$
\end{lemma}

\begin{proof}
	Let $\collecting{\model}^Y \subseteq 
	\fairness[Y]$. From the definition of 
	$\collecting{\model}^Y$ (cf. 
	Equation~\ref{eq:collecting}), we have that 
	$\maximal{\model}^Y \in \fairness[Y]$. 
	Thus, 
	from the definition of $\fairness[Y]$ (cf. 
	Equation~\ref{eq:fairness}), we have
	$\unused_\sensitive(\maximal{\model}^Y)$.
	Now, from the definition of 
	$\unused_\sensitive$ (cf. 
	Equation~\ref{eq:unused}), we equivalently have 
	$\forall \class \in \classes\colon 
	\unused_\sensitive(\maximal{\model}^Y_\class)$.
	Thus, we can conclude that $\forall \class \in 
	\classes\colon 
	\collecting{\model}^Y_\class \subseteq
	\fairness[Y]$.
\end{proof}

In particular, this means that in order to determine whether
a neural-network model $\model$ satisfies causal
fairness, we can independently verify, for each of its
possible target classes $\class \in \classes$, that the
values of its sensitive input nodes are unused.

We use this insight to abstract the collecting
semantics $\Lambda$ by
\emph{partitioning}.
More specifically,
let $\outcome \defined \set{
	\states^+_\class \mid
	\class \in \classes }$ be a trace partition
with respect to outcome.
We have the following Galois
connection
\begin{equation}\label{eq:outcome0}
	\tuple{\properties}{\subseteq}
	\galois{\alpha_\outcome}{\gamma_\outcome}
	\tuple{\properties}{\bulleq},
\end{equation}
where 
$\alpha_\outcome(S) \defined \set{
	T_\class \mid T \in S
	\land \class \in \classes}$. The order $\bulleq$
	is the
pointwise
ordering between sets of traces with the same
outcome, i.e.,
$A \bulleq B
\defined \bigwedge_{\class \in \classes}
\dot{A}_\class
\subseteq \dot{B}_\class$, where $\dot{S}_Z$
denotes the only non-empty set of traces in $S_Z$.
We can now 
define
the \emph{outcome semantics}
$\Lambda_\outcome \in
\powerset{\powerset{\states^+}}$ by abstraction of
$\Lambda$:
\begin{equation}\label{eq:outcome}
\Lambda_\outcome \defined
\alpha_\outcome(\Lambda)
= \set{ \Upsilon_\class  \mid \class \in \classes
}
\end{equation}
In the rest of the paper, we write
$\collecting{\model}_\outcome$ to denote the
outcome semantics
of a particular neural-network model $\model$.
\subsection{Dependency
Semantics}\label{subsec:dependency}

We observe that, to reason about causal fairness,
we do not need to consider all intermediate
computations between the initial and final states of
a trace. Thus, we can further abstract the outcome
semantics into a set of dependencies between
initial states and outcomes of traces.

To this end, we define the following Galois
connection\footnote{Note that here and in the
following, for
convenience, we
abuse
	notation and reuse the order symbol $\bulleq$
	defined over sets of sets of traces, instead of its
	abstraction,
	defined over sets of sets of pairs of states.}
\begin{equation}\label{eq:dependency0}
\tuple{\properties}{\bulleq}
\galois{\alpha_\dependency}{\gamma_\dependency}
\tuple{\powerset{\powerset{\states\times\states}}}{\bulleq},
\end{equation}
where 
$\alpha_\dependency(S)
\defined \set{
	\set{ \langle \trace_0, \trace_\omega \rangle
		\mid \trace \in T } \mid T \in S }$ \cite{Urban18}
abstracts away all intermediate states of any
trace.
We finally derive the \emph{dependency
semantics} $\Lambda_\dependency \in
\powerset{\powerset{\states\times\states}}$:
\begin{equation}\label{eq:dependency}
\Lambda_\dependency \defined
\alpha_\dependency(\Lambda_\outcome) =
\set{
	\set{ \tuple{\trace_0}{\trace_\omega} 
	\mid \trace \in
	\Upsilon_\class }
\mid \class \in \classes}
\end{equation}
In the following,
let $\collecting{\model}_\dependency$ denote the
dependency semantics
of a particular network model $\model$.

Let $R^Y \defined \set{ \tuple{\s}{\_}
	\in R \mid \s \in Y}$ restrict a set of pairs of states to 
	pairs whose first element is in $Y$ and, similarly, let 
	$S^Y \defined \set{ R^Y \mid R \in S}$ restrict a set 
	of sets of pairs of states to first elements in $Y$.
The next result shows that 
$\Lambda_\dependency$ is sound and
complete
for proving causal fairness:
\begin{theorem}\label{thm:dependency}
	$\model \models \fairness[Y] \Leftrightarrow
	\collecting{\model}^Y_\dependency
	 \bulleq
	\alpha_\dependency(\alpha_\outcome(\fairness[Y]))$
\end{theorem}

\begin{proof}
	Let $\model \models \fairness[Y]$. From 
	Theorem~\ref{thm:collecting}, we have that
	$\collecting{\model}^Y \subseteq
	\fairness[Y]$. Thus, from the Galois connections 
	in Equation~\ref{eq:outcome0} 
	and~\ref{eq:dependency0}, we have 
	$\alpha_\dependency(\alpha_\outcome(\collecting{\model}^Y))
	 \bulleq
	\alpha_\dependency(\alpha_\outcome(\fairness[Y]))$.
	 From the definition of 
	 $\collecting{\model}^Y_\dependency$ (cf. 
	 Equation~\ref{eq:dependency}), we can then 
	 conclude 
	 that 
	 $\collecting{\model}^Y_\dependency
	 \bulleq
	 \alpha_\dependency(\alpha_\outcome(\fairness[Y]))$.
\end{proof}

\begin{corollary}\label{cor:dependency}
	$\model \models \fairness[Y] \Leftrightarrow
	\collecting{\model}^Y_\dependency
	\subseteq
	\alpha_\dependency(\fairness[Y])$
\end{corollary}

\begin{proof}
	The proofs follows trivially from the definition of 
	$\bulleq$ (cf. Equation~\ref{eq:outcome0} 
	and~\ref{eq:dependency0}) and 
	Lemma~\ref{lem:outcome}.
\end{proof}

Furthermore, we observe that partitioning with 
respect to outcome induces a partition of the space 
of values of the input nodes \emph{used} for 
classification. 
For instance, partitioning $T'$ in 
Example~\ref{ex:unused} induces a partition on the 
values of (the indeed used node) $\node_{0,2}$.
Thus, 
we can equivalently verify whether 
$\collecting{\model}^Y_\dependency
\subseteq
\alpha_\dependency(\fairness[Y])$
by checking if the dependency semantics 
$\collecting{\model}^Y_\dependency$ 
induces a 
partition of 
$\restrict{Y}{\overline{\sensitive}}$.
%
Let $R_0 \defined \set{ \s \mid 
\tuple{\s}{\_} \in R}$ (resp. $R_
\omega \defined \set{ \s \mid 
\tuple{\_}{\s} \in R}$) be the selection of the first (resp. 
last) element from each pair in a set of pairs of states. We 
formalize this observation below.

\begin{lemma}\label{lem:dependency}
	$\model \models \fairness[Y] \Leftrightarrow 
	\forall A, B \in 
	\collecting{\model}^Y_\dependency\colon (A_\omega 
	\neq B_\omega 
	\Rightarrow \restrict{{A_0}}{\overline{\sensitive}} 
	\cap \restrict{{B_0}}{\overline{\sensitive}} = 
	\emptyset)$ 
\end{lemma}

\begin{proof}
	Let $\model \models \fairness[Y]$. From 
	Corollary~\ref{cor:dependency}, we have that 
	$\collecting{\model}^Y_\dependency
	\subseteq
	\alpha_\dependency(\fairness[Y])$. Thus, from 
	the definition of 
	$\collecting{\model}^Y_\dependency$ (cf. 
	Equation~\ref{eq:dependency}), we have
	$\forall \class \in \classes\colon 
	\alpha_\dependency(\maximal{\model}^Y_\class)
	\in
	\alpha_\dependency(\fairness[Y])$. 
	In particular, from the definition of 
	$\alpha_\dependency$ and $\fairness[Y]$ 
	(cf. Equation~\ref{eq:fairness}),
	we have that 
	$\unused_\sensitive(\maximal{\model}^Y_\class)$
	 for each 
	$\class \in \classes$.
	From the definition of 
	$\unused_\sensitive$ (cf. 
	Equation~\ref{eq:unused}),
	for each pair of \emph{non-empty} 
	$\maximal{\model}^Y_{\class_1}$ and 
	$\maximal{\model}^Y_{\class_2}$ for different 
	$\class_1, \class_2 \in \classes$ (the  
	case in which one or both are 
	empty is trivial), it must 
	necessarily be the 
	value of the non-sensitive input nodes in 
	$\overline{\sensitive}$ that causes the different 
	outcome $\class_1$ or $\class_2$. 
	We can thus conclude that $\forall A, B \in 
	\collecting{\model}^Y_\dependency\colon 
	(A_\omega 
	\neq B_\omega 
	\Rightarrow \restrict{{A_0}}{\overline{\sensitive}} 
	\cap \restrict{{B_0}}{\overline{\sensitive}} = 
	\emptyset)$. 
\end{proof}

\section{Na\"ive Causal-Fairness
Analysis}\label{sec:naive}

In this section, we present a first 
static analysis for 
causal fairness that computes a \emph{sound} 
over-approximation 
$\Lambda\abstraction_\dependency$ of the 
dependency semantics
$\Lambda_\dependency$, i.e., 
$\Lambda_\dependency \bulleq 
\Lambda\abstraction_\dependency$. This analysis corresponds to
the na\"ive approach we discussed in Section~\ref{sec:Overview}.
While it is 
too na\"ive to be practical, it is still useful
for building upon later in the paper.

\begin{algorithm}[t]
	\begin{algorithmic}[1]
		\Function{backward}{$\model$,
			$\abstractdomain$,
			$\node$}
		\State $\text{a} \gets
		\classification{\abstractdomain}{\node}(\textsc{new}_{\abstractdomain})$\label{naive:outcome}
		\For{$i \gets \layers-1 \text{ \textbf{down to} }
			0$}
		\For{$j \gets \size{\layer_i} \text{ \textbf{down
					to} }
			0$}
		\State $\text{a} \gets
		\bwdassign{\abstractdomain}{\node_{i,j}}(\bwdactivation[]{\abstractdomain}{\node_{i,j}}\text{a})$\label{naive:relu}
		\EndFor
		\EndFor
		\State\Return $\text{a}$
		\EndFunction
		
		\Function{check}{$\text{O}$}\label{naive:checkA}
		\State $\text{B} \gets
		\emptyset$\Comment{{\footnotesize
		 B: 
				biased}}
		\ForAll{$\text{o}_1, \text{a}_1 \in
			\text{O}$}
		\ForAll{$\text{o}_2 \neq \text{o}_1,
			\text{a}_2 \in \text{O}$}
		\If{$\text{a}_1
			\sqcap_{\abstractdomain_2}
			\text{a}_2 \neq
			\bot_{\abstractdomain_2}$}
		\State $\text{B} \gets \text{B} \cup
		\set{\text{a}_1
			\sqcap_{\abstractdomain_2}
			\text{a}_2}$
		\EndIf
		\EndFor
		\EndFor
		\State\Return $\text{B}$
		\EndFunction\label{naive:checkZ}
		
		\Function{analyze}{$\model$, 
			$\sensitive$, 
			$Y$, 
			$\abstractdomain$}\label{naive:input} %
		\State $\text{O} \gets
		\dot{\emptyset}$
		\For{$j \gets 0 \text{ 
		\textbf{ up to} }
			\size{\layer_\layers}$}\label{naive:parallelizable}
		\Comment{{\footnotesize perfectly parallelizable}}
		\State $\text{a} \gets \Call{backward}{\model,
			\abstractdomain,
			\node_{\layers,j}}$\label{naive:backward}
		\State $\text{O} \gets \text{O} \cup
		\set{\node_{\layers,j}
			\mapsto
			\restrict{(\assume{\abstractdomain}{Y}\text{a})}{\overline{\sensitive}}}$\label{naive:initial}
		\EndFor
		\State $\text{B} \gets
		\Call{check}{\text{O}}$\label{naive:check}
		\State\Return $\text{B} =
		\emptyset,
		\text{B}$\Comment{{\footnotesize fair:
				$\text{B} =
				\emptyset$, maybe biased: $\text{B} \neq 
				\emptyset$}}
		\EndFunction
	\end{algorithmic}
	\caption{: A Na\"ive Backward
		Analysis}\label{alg:naive}
\end{algorithm}

For simplicity, we consider $\relu$
activation functions. (We discuss extensions to
other activation functions in
Section~\ref{sec:analysis}.) The na\"ive static
analysis
is described in Algorithm~\ref{alg:naive}. It takes as
input (cf.
Line~\ref{naive:input}) a neural-network model $\model$, a set of
sensitive input nodes $\sensitive$ of $\model$, a
(representation of a) set of initial states of interest
$Y$, and an abstract domain $\abstractdomain$
to be used for the analysis. The analysis proceeds
backwards for each outcome (i.e., each target class
$\node_{\layers,j}$) of $\model$ (cf.
Line~\ref{naive:backward}) in order to determine an
over-approximation of the initial states that satisfy
$Y$ and lead to $\node_{\layers,j}$ (cf.
Line~\ref{naive:initial}).

More specifically, the transfer
function $\classification{\abstractdomain}{\node}$
(cf. Line~\ref{naive:outcome}) modifies a given
abstract-domain element to assume the given
outcome $\node$, that is, to assume that $\max
\nodes_\layers = \node$. The transfer
functions
$\bwdactivation{\abstractdomain}{\node_{i,j}}$ and
$\bwdassign{\abstractdomain}{\node_{i,j}}$
(cf.
Line~\ref{naive:relu}) respectively
consider a \textsc{ReLU} operation and
replace $\node_{i,j}$ with the corresponding linear
combination of nodes in the preceding layer (see
Section~\ref{sec:neuralnets}).

\looseness=-1
Finally, the analysis checks whether the computed
over-approximations satisfy causal fairness with respect 
to $\sensitive$ (cf. 
Line~\ref{naive:check}). In particular, it checks 
whether they induce a partition of 
$\restrict{Y}{\overline{\sensitive}}$ as observed 
for Lemma~\ref{lem:dependency} (cf. 
Lines~\ref{naive:checkA}-\ref{naive:checkZ}).
If so, we have proved that $\model$ satisfies 
causal fairness. If not, the analysis returns a set $B$ 
of abstract-domain elements over-approximating 
the input regions in which bias might occur.

\begin{theorem}
If $\Call{analyze}{\model, \sensitive, Y, 
\abstractdomain}$ of Algorithm~\ref{alg:naive} returns 
$\textsc{true}, \emptyset$ 
then $\model$ satisfies $\fairness[Y]$.
\end{theorem}

	\begin{proof}[Proof (Sketch)]
		$\Call{analyze}{\model, \sensitive, Y, 
			\abstractdomain}$ in 
			Algorithm~\ref{alg:naive} computes an 
			\emph{over-approximation} $a$ of the 
			regions of 
			the input space that yield each target class 
			$\node_{\layers,j}$ (cf. 
			Line~\ref{naive:backward}). 
			Thus, it actually computes an 
			over-approximation 
			${\collecting{\model}_\dependency^{Y\abstraction}}$
			 of the 
			dependency semantics
			$\collecting{\model}^Y_\dependency$, i.e., 
			$\collecting{\model}^Y_\dependency 
			\bulleq 
			\collecting{\model}^{Y\abstraction}_\dependency$.
			Thus, if 
			$\collecting{\model}^{Y\abstraction}_\dependency$
satisfies $\fairness[Y]$, i.e., 
$\forall A, B \in 
\collecting{\model}^{Y\abstraction}_\dependency\colon
(A_\omega 
\neq B_\omega 
\Rightarrow \restrict{{A_0}}{\overline{\sensitive}} 
\cap \restrict{{B_0}}{\overline{\sensitive}} = 
\emptyset)$
(according to Lemma~\ref{lem:dependency}, cf. 
Line~\ref{naive:check}), then by transitivity we can 
conclude that also 
$\collecting{\model}^{Y\abstraction}_\dependency$
necessarily satisfies $\fairness[Y]$.
	\end{proof}

In the analysis implementation, there is a tradeoff
between performance and precision, which is reflected in the choice of
abstract domain $\abstractdomain$ and its transfer functions.
Unfortunately, existing numerical abstract domains that are less
expressive than polyhedra~\cite{Cousot78}
would make for a rather fast but too imprecise analysis.
This is because they are not able to precisely handle 
constraints 
like $\max \nodes_\layers = \node$, which are introduced by 
$\classification{\abstractdomain}{\node}$ to partition 
with respect to outcome.

Furthermore, even polyhedra 
would not be precise enough in general. Indeed, each 
$\bwdactivation{\abstractdomain}{\node_{i,j}}$ would 
over-approximate what effectively is a conditional 
branch. 
Let $\size{\model} \defined \size{\layer_1} + \dots + 
\size{\layer_{\layers-1}}$ denote 
the number of hidden nodes (i.e., the number of 
\textsc{ReLU}s) in a model 
$\model$. 
On the other side of the spectrum, one could 
use a disjunctive completion~\cite{Cousot79} of 
polyhedra, thus keeping a separate polyhedron for each 
branch of a \textsc{ReLU}. This would yield a 
precise (in fact, exact) but extremely slow 
analysis: even with parallelization  
(cf. Lines~\ref{naive:parallelizable}), each of the 
\size{\layer_\layers} processes would have to effectively 
explore 
$2^{\size{\model}}$ 
paths!

In the rest of the paper, we improve on this na\"ive 
analysis and show how far we can go 
all the while remaining exact by using disjunctive 
polyhedra.

\section{Parallel Semantics}\label{sec:parallel}

We first have to take a step back and return
to reasoning at 
the concrete-semantics level. At the 
end of Section~\ref{sec:dependency}, we observed that the 
dependency semantics of a neural-network model 
$\model$ satisfying $\fairness[Y]$ effectively 
induces a partition of 
$\restrict{Y}{\overline{\sensitive}}$. We call this 
input partition 
\emph{fair}. 

More formally, given a set $Y$ of 
initial states of interest, we say that an input 
partition 
$\partitions$ of $Y$ is fair if all value choices 
$\choices$ for the sensitive input nodes $\sensitive$ of 
$\model$ are 
possible in all elements of the partitions: $\forall 
\partition \in \partitions, \choice \in \choices\colon 
\exists \s \in \partitions\colon \s(\sensitive) = \choice$.
For instance, $\partitions = \set{T_0, T'_0}$, with   
$T$ and $T'$ 
in Example~\ref{ex:unused} is a fair input partition 
of $Y 
= \set{ \s \mid \s(\node_{0,1}) = 0.5 \vee 
\s(\node_{0,1}) = 0.75}$. 

Given a fair input partition $\partitions$ of $Y$, the 
following 
result shows that we can verify whether a 
model $\model$ 
satisfies $\fairness[Y]$ for each element 
$\partition$ of $\partitions$, \emph{independently}.

\begin{lemma}\label{lem:parallelization}
	$\model \models \fairness[Y] \Leftrightarrow 
	\forall \partition \in \partitions\colon \forall A, B \in 
	\collecting{\model}^\partition_\dependency\colon 
	(A_\omega 
	\neq B_\omega 
	\Rightarrow \restrict{{A_0}}{\overline{\sensitive}} 
	\cap \restrict{{B_0}}{\overline{\sensitive}} = 
	\emptyset)$ 
\end{lemma}

\begin{proof}
	The proof follows trivially from 
	Lemma~\ref{lem:dependency} and the fact that 
	$\partitions$ is a fair partition.
\end{proof}

We use this new insight to further abstract the 
dependency semantics $\Lambda_\dependency$.
We have the following Galois connection
\begin{equation}\label{eq:parallel0}
\tuple{\powerset{\powerset{\states\times\states}}}{\bulleq}
\galois{\alpha_\partitions}{\gamma_\partitions}
\tuple{\powerset{\powerset{\states\times\states}}}{\bulleq_\partitions},
\end{equation}
where 
$\alpha_\partitions(S) \defined \set{
	R^\partition \mid R \in S \land \partition \in 
	\partitions}$. 
Here the order $\bulleq_\partitions$
is the
pointwise
ordering between sets of pairs of states restricted to 
first 
elements in the same $\partition \in \partitions$, i.e.,
$A \bulleq_\partitions B
\defined \bigwedge_{\partition \in \partitions}
\dot{A}^\partition
\subseteq \dot{B}^\partition$, where 
$\dot{S}^\partition$
denotes the only non-empty set of pairs in 
$S^\partition$.
We can now derive the 
\emph{parallel semantics} 
$\Pi^\partitions_\dependency \in
\powerset{\powerset{\states\times\states}}$: 

\begin{equation}\label{eq:parallel}
\begin{aligned}
\Pi^\partitions_\dependency &\defined
\alpha_\partitions(\Lambda_\dependency) = \set{
	\set{ \tuple{\trace_0}{\trace_\omega} 
		\mid \trace \in
		\Upsilon^\partition_\class }
	\mid \partition \in \partitions \land \class \in
	\classes} 
\end{aligned}
\end{equation}
In fact, we derive a hierarchy of semantics,  
as depicted in 
Figure~\ref{fig:hierarchy}. 
We write 
$\parallelizing{\model}^\partitions_\dependency$ to 
denote the 
parallel 
semantics of a particular neural-network model 
$\model$.
%
\begin{figure}[t]
\includegraphics[width=0.15\textwidth]{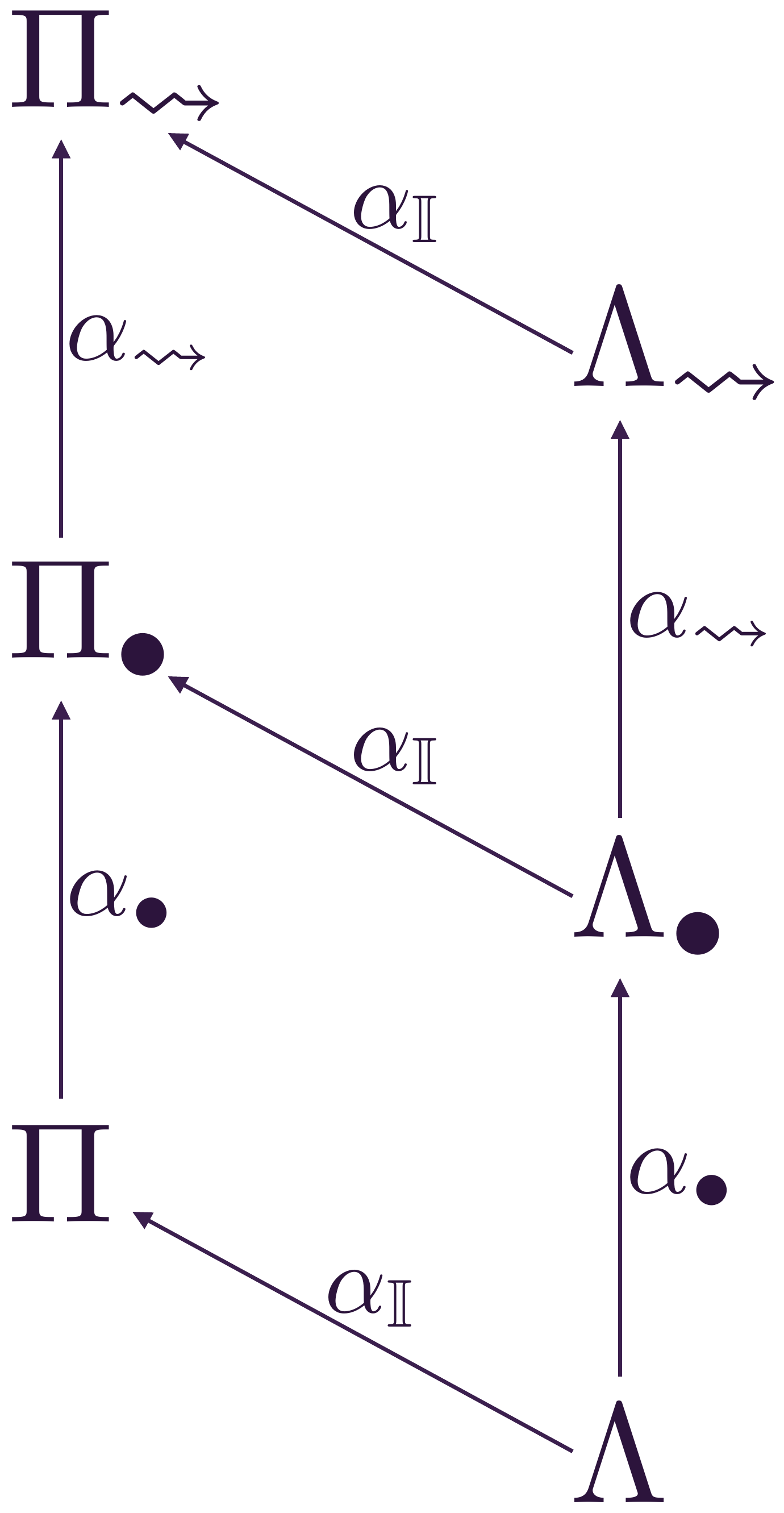}
	\caption{Hierarchy of
		semantics.}\label{fig:hierarchy}
\end{figure}
%
It remains to show
soundness and 
completeness for $\Pi^\partitions_\dependency$.

\begin{theorem}\label{thm:parallel}
	$\model \models \fairness[Y] \Leftrightarrow
	\parallelizing{\model}^\partitions_\dependency
	\bulleq_\partitions
	\alpha_\partitions(\alpha_\dependency(\alpha_\outcome(\fairness[Y])))$
\end{theorem}

\begin{proof}
	Let $\model \models \fairness[Y]$. From 
	Theorem~\ref{thm:dependency}, we have that
	$\collecting{\model}^Y_\dependency
	\bulleq
	\alpha_\dependency(\alpha_\outcome(\fairness[Y]))$.
	Thus, from the Galois connections 
	in Equation~\ref{eq:parallel0}, we have 
	$\alpha_\partitions(\collecting{\model}^Y_\dependency)
	\bulleq
	\alpha_\partitions(\alpha_\dependency(\alpha_\outcome(\fairness[Y])))$.
	From the definition of 
	$\parallelizing{\model}^\partitions_\dependency$
	(cf. Equation~\ref{eq:parallel}), we can then 
	conclude that 
	$\parallelizing{\model}^\partitions_\dependency
	\bulleq_\partitions
	\alpha_\partitions(\alpha_\dependency(\alpha_\outcome(\fairness[Y])))$.
\end{proof}

\begin{corollary}\label{cor:parallel}
	%
	$\model \models \fairness[Y] \Leftrightarrow
	\parallelizing{\model}^\partitions_\dependency
	\subseteq
	\alpha_\partitions(\alpha_\dependency(\fairness[Y]))$
\end{corollary}

\begin{proof}
	The proofs follows trivially from the definition of 
	$\bulleq_\partitions$ (cf. 
	Equation~\ref{eq:outcome0} 
	and~\ref{eq:dependency0} and~\ref{eq:parallel0}) 
	and 
	Lemma~\ref{lem:outcome} 
	and~\ref{lem:parallelization}.
\end{proof}

Finally, from Lemma~\ref{lem:parallelization}, we have 
that we 
can equivalently verify whether 
$\parallelizing{\model}^\partitions_\dependency
\subseteq
\alpha_\partitions(\alpha_\dependency(\fairness[Y]))$ 
by checking if the parallel semantics
$\parallelizing{\model}^\partitions_\dependency$ 
induces a 
partition of each
$\restrict{\partition}{\overline{\sensitive}}$.

\begin{lemma}\label{lem:parallel}
	$\model \models \fairness[Y] \Leftrightarrow 
	\forall \partition \in \partitions\colon\forall A, B \in 
	\parallelizing{\model}^\partitions_\dependency\colon 
	(A^\partition_\omega 
	\neq B^\partition_\omega 
	\Rightarrow 
	\restrict{{A^\partition_0}}{\overline{\sensitive}} 
	\cap 
	\restrict{{B^\partition_0}}{\overline{\sensitive}} = 
	\emptyset)$ 
\end{lemma}

\begin{proof}
	The proof follows trivially from 
	Lemma~\ref{lem:parallelization}.
\end{proof}

\section{Parallel Causal-Fairness 
Analysis}\label{sec:analysis}

In this section, we build on the parallel semantics to design our novel
\emph{perfectly parallel} static analysis for 
causal 
fairness, which automatically finds a fair partition 
$\partitions$ and computes a sound 
over-approximation 
$\Pi^{\partitions\abstraction}_\dependency$
 of ${\Pi^\partitions}_\dependency$, i.e., 
 ${\Pi^\partitions}_\dependency \bulleq_\partitions 
 \Pi^{\partitions\abstraction}_\dependency$. 

\begin{algorithm}[t]
	\begin{algorithmic}[1]
		\Function{forward}{$\model$, 
			$\abstractdomain$, 
			$\partition$} 
		\State $\text{a}, \text{p} \gets
		\assume{\abstractdomain}{\partition}(
		\init{\abstractdomain}),
		\epsilon$\label{patterns:empty}
		\For{$i \gets 1 \text{ \textbf{up to} } \layers$}
		\For{$j \gets 0 \text{ \textbf{up to} }
			\size{\layer_i}$}
		\State $\text{a}, \text{p} \gets
		\fwdactivation[\text{p}]{\abstractdomain}{\node_{i,j}}(\fwdassign{\abstractdomain}{\node_{i,j}}\text{a})$\label{patterns:pattern}
		\EndFor
		\EndFor
		\State\Return $\text{a}, \text{p}$
		\EndFunction
		
		\Function{backward}{$\model$,
			$\abstractdomain$,
			$\mathrm{O}$, $\text{p}$}
		\State $\text{a} \gets
		\classification{\abstractdomain}{\mathrm{O}}(\textsc{new}_{\abstractdomain})$
		\For{$i \gets \layers-1 \text{ \textbf{down to} }
			0$}
		\For{$j \gets \size{\layer_i} \text{ \textbf{down
					to} }
			0$}
		\State $\text{a} \gets
		\bwdassign{\abstractdomain}{\node_{i,j}}(\bwdactivation[\text{p}]{\abstractdomain}{\node_{i,j}}\text{a})$\label{patterns:relu}
		\EndFor
		\EndFor
		\State\Return $\text{a}$
		\EndFunction
		
		\Function{analyze}{$\model$, 
			$\sensitive$, 
			$Y$, 
			$\abstractdomain_1$, 
			$\abstractdomain_2$, 
			$\lowerbound$, 
			$\upperbound$}\label{patterns:input} %
		\State $\text{F}, \text{E}, \text{C} \gets
		\dot{\emptyset}, \dot{\emptyset},
		\emptyset$\label{patterns:sets}\Comment{{\footnotesize
		 F: feasible,
				E: excluded,
				C: completed}}
		\State $\partitions \gets 
		\set{Y}$\label{patterns:start}
		\While{$\partitions \neq 
			\emptyset$}\label{patterns:parallelizableA}
		\Comment{{\footnotesize perfectly parallelizable}}
		\State $\partition \gets 
		\partitions.\Call{get()}{}$\label{patterns:get}
		\State $\text{a}, \text{p} \gets
		\Call{forward}{\model, \abstractdomain_1,
			\partition}$\label{patterns:forward}
		\If{\Call{uniquely-classified}{a}}\label{patterns:fair}\Comment{{\footnotesize
				$\partition$ is already
				fair}}
		\State $\text{C} \gets \text{C} \cup
		\set{\partition}$\label{patterns:completed}
		\ElsIf{$\size{\model} - \size{\text{p}} \leq
			\upperbound$}\label{patterns:upperbound}\Comment{{\footnotesize
				$\partition$ is feasible}}
		\State $\text{F} \gets \text{F} \uplus
		\set{\text{p} \mapsto 
		\partition}$\label{patterns:feasible}
		\ElsIf{$\size{\partition} \leq
			\lowerbound$}\label{patterns:lowerbound}\Comment{{\footnotesize
				$\partition$ is
				excluded}}
		\State $\text{E} \gets \text{E} \uplus
		\set{\text{p} \mapsto 
		\partition}$\label{patterns:excluded}
		\Else\label{patterns:continue}\Comment{{\footnotesize
		 $\partition$ must
				be partitioned further}}
		\State $\partitions \gets \partitions \cup
		\textsc{partition}_{\overline{\sensitive}}(\partition)$\label{patterns:partitioned}
		\EndIf
		\EndWhile
		
		\State $\text{B} \gets
		\emptyset$\Comment{{\footnotesize B: biased}}
		\ForAll{$\text{p}, \partitions \in
			\text{F}$}\label{patterns:parallelizableB}
		\Comment{{\footnotesize perfectly parallelizable}}
		\State $\text{O} \gets \dot{\emptyset}$
		\For{$j \gets 0 \text{ \textbf{up to} }
			\size{\layer_\layers}$}
		\State $\text{a} \gets \Call{backward}{\model,
			\abstractdomain_2, \node_{\layers,j},
			\text{p}}$\label{patterns:backward}
		\State $\text{O} \gets \text{O} \cup
		\set{\node_{\layers,j}
			\mapsto \text{a}}$
		\EndFor
		\ForAll{$\partition \in
			\partitions$}\label{patterns:parallelizableC}
		\State $\text{O'} \gets \dot{\emptyset}$
		%
		%
		\ForAll{$\text{o}, \text{a} \in \text{O}$}
		\State $\text{O'} \gets \text{O'} \cup
		\set{\text{o} \mapsto
			\restrict{(\assume{\abstractdomain_2}{\partition}\text{a})}{\overline{\sensitive}}}$\label{patterns:initial}
		\EndFor
		\State $\text{B} \gets \text{B} \cup
		\Call{check}{\text{O'}}$\label{patterns:check}
		\State $\text{C} \gets \text{C} \cup
		\set{\partition}$\label{patterns:end}
		\EndFor
		\EndFor
		\State\Return $\text{C}, \text{B} = \emptyset,
		\text{B},
		\text{E}$\label{patterns:return}\Comment{{\footnotesize
		 fair: $\text{B} =
				\emptyset$, maybe biased: $\text{B} \neq
				\emptyset$}}
		\EndFunction
	\end{algorithmic}
	\caption{: Our Analysis Based on Activation Patterns}\label{alg:patterns}
\end{algorithm}

\paragraph{\textbf{\relu{} Activation Functions.}}
We again only consider 
\textsc{ReLU} activation functions for now and postpone 
the discussion of other activation functions to the 
end of the section. 
The analysis is 
described in Algorithm~\ref{alg:patterns}. It combines
a forward pre-analysis 
(Lines~\ref{patterns:start}-\ref{patterns:excluded}) 
with a 
backward analysis 
(Lines~\ref{patterns:parallelizableB}-\ref{patterns:end}).
The forward pre-analysis uses an abstract domain 
$\abstractdomain_1$ and builds 
partition $\partitions$, while the backward analysis 
uses an abstract domain $\abstractdomain_2$ and 
performs the actual causal-fairness analysis of a 
neural-network model $\model$ with respect to its
sensitive input nodes $\sensitive$ and a 
(representation of a) set of initial states 
$Y$ (cf. Line~\ref{patterns:input}).

More specifically, the forward pre-analysis bounds 
the number of paths that the backward analysis has 
to explore. Indeed, not all of the 
$2^{\size{\model}}$ paths of a
model $\model$ are necessarily viable starting from 
its input space.

In the rest of this section, we represent each path by an 
\emph{activation pattern}, which determines 
the 
activation status of every $\relu$ operation in 
$\model$. More precisely, an activation pattern is a 
sequence of flags. Each flag $\pattern_{i,j}$ 
represents the 
activation status of the $\relu$ 
operation used to compute the value of hidden 
node $\node_{i,j}$. If $\pattern_{i,j}$ is 
$\node_{i,j}$, the 
$\relu$ is always active, otherwise the $\relu$ is 
always inactive and $\pattern_{i,j}$ is 
$\overline{\node_{i,j}}$. 

An \emph{abstract 
activation pattern} gives the activation status of 
only a subset of the $\relu$s of 
$\model$, and thus, represents a set of activation 
patterns. $\relu$s whose corresponding flag does 
not appear in an abstract activation pattern have an 
unknown (i.e., not fixed) activation status. 
Typically, \emph{only a relatively
small number of abstract activation patterns is 
sufficient for covering the entire input space of a
neural-network model}. The 
design of our analysis builds on this key 
observation.

We set an analysis \emph{budget} by 
providing an upper bound $\upperbound$ (cf. 
Line~\ref{patterns:input}) on the number 
of tolerated $\relu$s with an unknown activation 
status for each element $\partition$ of 
$\partitions$, i.e., on the number of paths 
that are to be explored by the backward analysis in each 
$\partition$. The forward pre-analysis starts with 
the trivial partition $\partitions = \set{Y}$ (cf. 
Line~\ref{patterns:start}). It proceeds forward for 
each element $\partition$ in $\partitions$ (cf. 
Lines~\ref{patterns:get}-\ref{patterns:forward}). The 
transfer function  
$\fwdactivation[\text{p}]{\abstractdomain}{\node_{i,j}}$
considers a $\relu$ operation and 
additionally builds an abstract activation pattern 
$p$ for $\partition$ (cf. 
Line~\ref{patterns:pattern}) starting from the empty 
pattern $\epsilon$ (cf. Line~\ref{patterns:empty}).

If $\partition$ leads to a unique outcome (cf. 
Line~\ref{patterns:fair}), then causal fairness is 
already proved for $\partition$, and 
there is no need for a backward analysis;
$\partition$ is added to the set of 
\emph{completed} 
partitions (cf. Line~\ref{patterns:completed}). 
Instead, if abstract activation pattern $p$ 
fixes the activation status of enough $\relu$s (cf. 
Line~\ref{patterns:upperbound}), we 
say that the backward analysis for $\partition$ is 
\emph{feasible}. In this case, the pair of $p$ and 
$\partition$ is inserted into a map $F$ from 
abstract 
activation patterns to feasible partitions (cf. 
Line~\ref{patterns:feasible}). The insertion takes 
care of merging abstract activation 
patterns that are subsumed by other (more) abstract
patterns. In other words, it groups partitions 
whose abstract activation 
patterns fix 
more $\relu$s with partitions whose patterns fix fewer 
$\relu$s, and therefore, represent a superset of 
(concrete) patterns.

Otherwise, $\partition$ needs to be partitioned 
further, with respect to 
$\overline{\sensitive}$ (cf. 
Line~\ref{patterns:continue}). Partitioning may 
continue until the size of $\partition$ is smaller than the 
given lower bound $\lowerbound$ (cf. 
Lines~\ref{patterns:input} and 
\ref{patterns:lowerbound}). At this point, 
$\partition$ is set aside and excluded from the 
analysis until more resources (a larger upper bound
$\upperbound$ or a smaller lower bound 
$\lowerbound$) 
become available (cf. Line~\ref{patterns:excluded}).

Note that the forward pre-analysis lends itself 
to choosing a relatively cheap abstract domain 
$\abstractdomain_1$ since it does not need to 
precisely handle polyhedral constraints (like $\max 
\nodes_\layers = \node$, needed to partition with 
respect to outcome, cf. Section~\ref{sec:naive}).

The analysis then proceeds backwards, 
independently 
for each abstract activation path $p$ and 
associated group of partitions $\partitions$ (cf. 
Lines~\ref{patterns:parallelizableB} and 
\ref{patterns:backward}). The transfer function 
$\bwdactivation[\text{p}]{\abstractdomain}{\node_{i,j}}$
uses $p$ to choose
which path(s) to explore at each $\relu$ operation, 
i.e., only the active (resp. inactive) path if 
$\node_{i,j}$ (resp. $\overline{\node_{i,j}}$) appears 
in $p$, or both if 
the activation status of the $\relu$ corresponding 
to hidden node $\node_{i,j}$ is unknown. The 
(as we have seen, necessarily) expensive backward 
analysis only 
needs to run for 
each abstract activation pattern in the feasible map 
$F$. This is also why it is advantageous to merge 
subsumed abstract activation paths as described 
above.

Finally, the analysis checks causal fairness of 
each element $\partition$ associated to $p$ (cf. 
Line~\ref{patterns:check}). The analysis returns the 
set of input-space regions $C$ that have been completed
and a set $B$ of abstract-domain elements over-approximating the regions 
in which bias might occur (cf. 
Line~\ref{patterns:return}). If $B$ is empty, then the 
given model $\model$ satisfies 
causal fairness with respect to $\sensitive$ and 
$Y$ over $C$.

\begin{theorem}\label{thm:soundness}
	If function $\Call{analyze}{\model, 
		\sensitive, 
		Y, 
		\abstractdomain_1, 
		\abstractdomain_2, 
		\lowerbound, 
		\upperbound}$ 
	in Algorithm~\ref{alg:patterns} 
		returns 
	$\text{C}, \textsc{true}, \emptyset$, 
	then $\model$ satisfies 
	$\fairness[Y]$ 
	over the input-space fraction $C$.
\end{theorem}

	\begin{proof}[Proof (Sketch)]
		$\Call{analyze}{\model, 
			\sensitive, 
			Y, 
			\abstractdomain_1, 
			\abstractdomain_2, 
			\lowerbound, 
			\upperbound}$ in 
		Algorithm~\ref{alg:patterns} first computes the 
		abstract activation patterns that cover a 
		fraction $C$ of the input space in which the 
		analysis is feasible 
		(Lines~\ref{patterns:start}-\ref{patterns:excluded}).
		 Then, it computes an 
		\emph{over-approximation} $a$ of the 
		regions of $C$
		that yield each target class 
		$\node_{\layers,j}$ (cf. 
		Line~\ref{patterns:backward}). 
		Thus, it actually computes an 
		over-approximation 
		${\parallelizing{\model}_\dependency^{\partitions\abstraction}}$
		of the 
		parallel semantics
		$\parallelizing{\model}^\partitions_\dependency$,
		 i.e., 
		$\parallelizing{\model}^\partitions_\dependency
		\bulleq 
		\parallelizing{\model}^{\partitions\abstraction}_\dependency$.
		Thus, if 
		$\parallelizing{\model}^{\partitions\abstraction}_\dependency$
		satisfies $\fairness[Y]$, i.e., 
		$\forall \partition \in \partitions\colon\forall 
		A, B \in 
		\parallelizing{\model}^{\partitions\abstraction}_\dependency\colon
		(A^\partition_\omega 
		\neq B^\partition_\omega 
		\Rightarrow 
		\restrict{{A^\partition_0}}{\overline{\sensitive}} 
		\cap 
		\restrict{{B^\partition_0}}{\overline{\sensitive}} 
		= 
		\emptyset)$ 
		(according to Lemma~\ref{lem:parallel}, cf. 
		Lines~\ref{patterns:parallelizableC}-\ref{patterns:check}),
		 then by transitivity 
		we 
		can 
		conclude that also 
		$\parallelizing{\model}^{\partitions\abstraction}_\dependency$
		necessarily satisfies $\fairness[Y]$.
	\end{proof}

\begin{remark}
Recall that we assumed neural-network nodes to have real values (cf. 
Section~\ref{sec:traces}). Thus, 
Theorem~\ref{thm:soundness} is true for all choices 
of classical numerical abstract domains 
\cite[etc.]{Cousot76,Cousot78,GhorbalGP09,Mine06a}
for $\abstractdomain_1$ and $\abstractdomain_2$.
If we were to consider floating-point values instead, 
the only sound choices would be floating-point 
abstract domains \cite{Chen08,Mine04,Singh19}.
\end{remark}

\paragraph{\textbf{Other Activation Functions.}}
Let us discuss how activation functions other than 
$\relu$s would be handled. The only difference in 
Algorithm~\ref{alg:patterns} would be the transfer 
functions 
$\fwdactivation[\text{p}]{\abstractdomain}{\node_{i,j}}$
 (cf. Line~\ref{patterns:pattern}) and 
$\bwdactivation[\text{p}]{\abstractdomain}{\node_{i,j}}$
 (cf. Line~\ref{patterns:relu}),
which would have to be replaced with the transfer 
functions corresponding to the considered activation 
function. 

Piecewise-linear activation functions, like
$\textsc{Leaky ReLU}(x) = \max(x, k \cdot x)$ or 
$\textsc{Hard TanH}(x) = \max(-1, \min(x, 1))$, 
can be treated analogously to $\relu$s. 
The case of \textsc{Leaky ReLU}s is trivial.
For \textsc{Hard TanH}s, the patterns $p$ 
used in Algorithm~\ref{alg:patterns} will consist of 
flags $\pattern_{i,j}$ with three possible values, 
depending on whether the corresponding hidden 
node $\node_{i,j}$ has value less than or equal to 
$-1$, greater than or equal to $1$, or between 
$-1$ and $1$.
For these activation functions, our approach remains 
sound and, in practice, exact when using 
disjunctive polyhedra for the backward analysis. 

Other activation
functions, e.g., $\textsc{Sigmoid}(x) = 
\frac{1}{1+e^{-x}}$, can be 
soundly over-approximated 
\cite{Singh19} and similarly treated in a piecewise manner. In 
this case, however, we necessarily lose the 
exactness of the 
analysis, even when using disjunctive polyhedra.

\section{Implementation}\label{sec:implementation}

We implemented our causal-fairness analysis described in the previous
section in a tool called \tool. The implementation is written in
\python and is open 
source\footnote{\url{https://github.com/caterinaurban/Libra}}.

\paragraph{\textbf{Tool Inputs.}} \tool takes as 
input a 
neural-network model 
$\model$ expressed as a \python program (cf.
Section~\ref{sec:neuralnets}), a specification of the input layer
$\layer_0$ of $\model$, an abstract domain for the forward
pre-analysis, and budget constraints $\lowerbound$ and
$\upperbound$. The specification for $\layer_0$ determines which
input nodes correspond to continuous and (one-hot encoded) categorical
data and, among them, which should be considered bias sensitive. We
assume that continuous data is in the range $[0,1]$. A set $Y$ of
initial states of interest is specified using an assumption at the
beginning of the program representation of $\model$.

\paragraph{\textbf{Abstract Domains.}} For the 
forward 
pre-analysis, choices of the abstract domain are 
either 
boxes \cite{Cousot76} (i.e., 
\textsc{boxes} in the following), or a combination of 
boxes and symbolic 
constant
propagation~\cite{Li19,Mine06b} (i.e., 
\textsc{symbolic} in the following), or the 
\textsc{deeppoly} domain~\cite{Singh19}, which is designed for
proving local robustness of neural networks. As 
previously 
mentioned, we use
disjunctive polyhedra for the backward analysis. All abstract domains
are built on top of the \apron abstract-domain library
\cite{Jeannet09}.

\paragraph{\textbf{Parallelization.}} Both the forward and 
backward 
analyses are parallelized to run on multiple CPU cores. The pre-analysis uses a queue from
which each process draws a fraction $\partition$ of $Y$
(cf. Line~\ref{patterns:get}). Fractions that need to be
partitioned further are split in half along one of the non-sensitive
dimensions (in a round-robin fashion), and the resulting (sub)fractions
are put back into the queue (cf.  Line~\ref{patterns:partitioned}).
Feasible $\partition$s (with their corresponding abstract activation
pattern $p$) are put into another queue
(cf. Line~\ref{patterns:feasible}) for the 
backward
analysis. 

\paragraph{\textbf{Tool Outputs.}} The analysis 
returns 
the fractions of $Y$ that were 
analyzed and any (sub)regions of these where bias was found. It also
reports the percentage of the input space that was analyzed and (an
estimate of) the percentage that was found biased according to a given
probability distribution of the input space (uniform by default). To obtain the latter, we simply
use the size of a box wrapped around each biased
region. More precise but also costlier solutions exist
\cite{Barvinok94}.

\section{Experimental Evaluation}

In this section, we evaluate our approach by focusing on the following
research questions:
\begin{description}
\item[RQ1:] Can our analysis detect seeded
(i.e., injected) bias?
\item[RQ2:] Is our analysis able to answer specific 
bias queries?
\item[RQ3:] How does the model structure
affect 
the scalability of the analysis?
\item[RQ4:] How does the analyzed input-space size
affect the scalability of the analysis?
\item[RQ5:] How does the analysis budget affect 
the scalability-vs-precision tradeoff?
\item[RQ6:] Can our analysis effectively
leverage multiple CPUs?
\end{description}





\subsection{Data}

For our evaluation, we used 
public datasets from the UCI Machine Learning
Repository  and ProPublica
(see below for more details) to
train several neural-network models. We primarily 
focused on datasets 
discussed in the
literature~\cite{MehrabiMSLG-19} or used by
related
techniques (e.g., \cite{Galhotra17,UdeshiAC18,TramerAGHHHJL17,Bastani18,Albarghouthi17a,Albarghouthi17b,Albarghouthi19,Datta17}).

We pre-processed these datasets both to make them fair with respect to
a certain sensitive input feature as well as to seed 
bias. We describe how
we seeded bias in each particular dataset later on. 

Our methodology for making the data fair was common
across datasets.
In particular, given an original dataset and a sensitive feature (say,
race), we selected the largest population with a particular value for
this feature (say, Caucasian) from the dataset (and discarded all
others). We removed any duplicate or inconsistent entries from this
population. We then duplicated the population for every other value of
the sensitive feature (say, Asian and Hispanic). For example, assuming
the largest population was 500 Caucasians, we created 500 Asians and
500 Hispanics, and any two of these populations differ only in the
value of race. Consequently, the new dataset is causally fair because
there do not exist two inputs $k$ and $k'$ that differ only in the
value of the sensitive feature for which the classification outcomes
are different.

We define the \emph{causal-unfairness score} of a dataset as the
percentage of inputs $k$ in the dataset for which there exists another
input $k'$ that differs from $k$ only in the value of 
the sensitive feature and the classification
outcome. Our fair datasets have an unfairness score of 0\%.

All datasets used in our experiments are open 
source as part of \tool. 

\subsection{Setup}

Since neural-network training is non-deterministic, we typically train
eight neural networks on each dataset, unless 
stated otherwise. The model
sizes range from 2 hidden layers with 5 nodes each 
to 32
hidden layers with 40 nodes each. All  models used 
in our experiments
are open source as part of \tool. For
each model, we assume a uniform distribution of 
the input space.

We performed all experiments on a 12-core
\mbox{Intel}
\textregistered~Xeon \textregistered~X5650 
CPU~@~2.67GHz machine with $48$GB of 
memory, running Debian GNU/Linux 9.6 (stretch).

\subsection{Results}

In the following, we present our experimental results for each of the
above research questions.

\paragraph{\textbf{RQ1: Detecting Seeded Bias.}}
This research question focuses on detecting seeded bias by comparing
the analysis results for models trained with fair 
versus biased
data. 

For this experiment, we used the German Credit
dataset\footnote{\url{https://archive.ics.uci.edu/ml/datasets/Statlog+(German+Credit+Data)}}.
 This
dataset classifies 
creditworthiness into two categories,
``good'' and ``bad''. An input feature is age, which we consider
sensitive to bias. (Recall that this could also be an input feature
that the user considers indirectly sensitive to bias.)
We seeded bias in the fair dataset by randomly assigning a bad credit
score to people of age 60 and above who request a 
credit amount of
more than EUR 1~000 until we reached a 20\% causal-unfairness score of
the dataset.
The median classification accuracy of the models 
(17 inputs and 4 hidden layers with 5 nodes each) 
trained on fair and
biased data was 71\% and 65\%, respectively. 
Note that accuracy does
not improve by adding more layers or nodes per 
layer --- we tried up to
100 hidden layers with 100 nodes each.

\begin{table}[t]
        \caption{Analysis of Models Trained on Fair 
        and \{Age, Credit > 1000\}-Biased Data 
        (German Credit Data)} 
        \label{tbl:germanT}
\centering
\resizebox{0.96\textwidth}{!}{%
	\begin{tabular}{|c|cc|cc|cc|cc|cc|cc|c}
		\cline{1-13}
		\multirow{3}{*}{$\textsc{credit}$} & 
		\multicolumn{4}{c|}{\textsc{boxes}} 
		& 
		\multicolumn{4}{c|}{\textsc{symbolic}} 
		& 
		\multicolumn{4}{c|}{\textsc{deeppoly}} &
		\\
		&  
		\multicolumn{2}{c|}{\textsc{fair data} } 
	& 
	\multicolumn{2}{c|}{\textsc{biased data}} &
		\multicolumn{2}{c|}{\textsc{fair data} } 
		& 
		\multicolumn{2}{c|}{\textsc{biased data}} &
		\multicolumn{2}{c|}{\textsc{fair data} } 
		& 
		\multicolumn{2}{c|}{\textsc{biased data}} &
		\\
		& 
		\textsc{bias} & 
		\textsc{time} &
		\textsc{bias} & 
		\textsc{time} & 
		\textsc{bias} & 
		\textsc{time} &
		\textsc{bias} & 
		\textsc{time} & 
		\textsc{bias} & 
		\textsc{time} &
		\textsc{bias} & 
		\textsc{time} &\\
\hline%
%
\rowcolor{gray!10}
\cellcolor{mygreen!20} & {\small 
	$0.09\%$} &  {\small 
	47s} & {\small 
	$0.09\%$} &  
	{\small 
	2m 17s} 
& {\small 
	$0.09\%$} &  {\small 
	13s} & {\small $0.09\%$} & {\small 
	1m 10s} 
&  {\small 
	$\mathbf{0.09\%}$} &  \cellcolor{mygreen!20} 
	{\small 
	\textbf{10s}} & {\small $\mathbf{0.09\%}$} & 
	\cellcolor{mygreen!20} {\small 
	\textbf{39s}} & \multicolumn{1}{c|}{\textsc{min}}
\\
\cellcolor{mygreen!20} & {\small 
	$0.19\%$} &  {\small 
	5m 46s} & {\small $0.45\%$} & {\small 
	13m 2s} 
&  {\small 
	$\mathbf{0.19\%}$} &  \cellcolor{mygreen!20} 
	{\small 
	\textbf{1m 5s}} & {\small $0.45\%$} & {\small 
	2m 41s} 
& {\small 
	$0.19\%$} &  {\small 
	1m 12s} &  {\small 
	$\mathbf{0.45\%}$} & 
	\cellcolor{mygreen!20} {\small 
	\textbf{1m 46s}} & 
	\multicolumn{1}{c|}{\textsc{median}}
\\
\rowcolor{gray!10}
\multirow{-3}{*}{\cellcolor{mygreen!20}$\leq 
1000$} & 
{\small 
	$0.33\%$} & {\small 
	30m 59s} & {\small $0.95\%$} & {\small 
	1h 56m 57s} 
& 
 {\small 
	$\mathbf{0.33\%}$} & \cellcolor{mygreen!20} 
	{\small 
	\textbf{4m 8s}} & {\small 
	$\mathbf{0.95\%}$} & \cellcolor{mygreen!20} 
	{\small 
	\textbf{13m 16s}} 
& {\small 
	$0.33\%$} &  {\small 
	5m 45s} & {\small $0.95\%$} & {\small 
	18m 18s} & \multicolumn{1}{c|}{\textsc{max}}
\\
\hline%
\cellcolor{mypurple!20} & {\small 
	$2.21\%$} &  {\small  
	1m 42s} &  {\small 
	$4.52\%$} &
	{\small 
	21m 11s} 
&  {\small 
	$\mathbf{2.21\%}$} & \cellcolor{mypurple!20} 
	{\small  
	\textbf{38s}} & {\small $\mathbf{4.52\%}$} & 
	\cellcolor{mypurple!20} {\small 
	\textbf{3m 7s}} 
& {\small 
	$2.21\%$} &  {\small 
	39s} & {\small $4.52\%$} & {\small 
	4m 44s} & \multicolumn{1}{c|}{\textsc{min}}
\\
\rowcolor{gray!10}
\cellcolor{mypurple!20} & {\small 
	$6.72\%$} &{\small  
	31m 42s} & {\small $23.41\%$} & {\small 
	1h 36m 51s} 
& {\small 
	$6.72\%$} &{\small  
	8m 59s} & {\small $23.41\%$} & {\small 
	41m 44s} 
&  {\small 
	$\mathbf{6.63\%}$} & \cellcolor{mypurple!20} 
	{\small 
	\textbf{4m 58s}} & 
	{\small $\mathbf{23.41\%}$} & 
	\cellcolor{mypurple!20} {\small 
	\textbf{15m 39s}} & 
	\multicolumn{1}{c|}{\textsc{median}}
\\
\multirow{-3}{*}{\cellcolor{mypurple!20}$> 1000$} 
& 
{\small 
	$14.96\%$} & {\small  
	7h 7m 12s} & {\small $33.19\%$} & {\small 
	16h 50m 48s} 
& 
{\small 
	$14.96\%$} & {\small  
	4h 16m 52s} & {\small $33.19\%$} & {\small 
	8h 5m 14s} 
&  {\small 
	$\mathbf{14.96\%}$} & \cellcolor{mypurple!20} 
	{\small 
	\textbf{1h 9m 45s}} &  
	{\small $\mathbf{31.17\%}$} 
	&\cellcolor{mypurple!20} {\small 
	\textbf{6h 51m 50s}} & 
	\multicolumn{1}{c|}{\textsc{max}}
\\
\hline
\end{tabular}
}
\end{table}

To analyze these models, we set $\lowerbound = 
0$ to be sure to complete 
the analysis on $100\%$ of the input space. The 
drawback with this is that the pre-analysis might 
end up splitting input partitions endlessly. To 
counteract, for each model, we chose the smallest 
upper bound $\upperbound$ that did not cause this 
issue. 
Table~\ref{tbl:germanT} shows the analysis results 
for the different choices of domain used for the 
forward pre-analysis.
In particular, it shows whether the models are 
biased with respect to
age for credit requests of 1~000 or less as well as 
for credit
requests of over 1~000.
Columns \textsc{bias} and \textsc{time} show the 
detected bias (in 
percentage of
the entire input space) and the analysis running 
time. We show minimum, 
median, and maximum bias percentage and running 
time for each 
credit 
request group. For each line in 
Table~\ref{tbl:germanT}, we highlighted the choice 
of the abstract domain that entailed the shortest 
analysis time.
The analysis results for all models 
are shown in the appendix (cf. 
Tables~\ref{tbl:germanTfull0}-\ref{tbl:germanTfull2}).
 
For all models, 
the analysis finds little bias for small 
credit amounts, as intended.
Instead, for large credit amounts, the analysis finds 
significantly more bias (i.e., about three times as 
much 
median bias)
for the models trained 
on biased data in comparison to models trained on 
fair data.
This
demonstrates that \emph{our approach is able to 
effectively detect seeded bias}.

For the models trained on fair data, we observe 
a maybe unexpected difference in the bias found for 
small credit 
amounts compared to larger amounts. This is in 
part due to the fact that bias is given in 
percentage of the entire input space and not scaled 
with respect to the \emph{analyzed} input space. 
When considering the 
analyzed input space (small credit amounts 
correspond to a mere 
4\% of the input space), the difference is less 
marked: 
the median bias is 
0.19\% / 4\% = 4.75\% for small credit amounts 
and 6.72\% / 96\% = 7\% (or 6.63\% / 96\% = 
6.9\% for the \textsc{deeppoly} domain) for large 
credit amounts.
The remaining difference indicates that the models 
contain bias that does not necessarily depend
on the credit amount. The bias is introduced by the 
training process
itself (as explained in the Introduction) and is not 
due to imprecision
of our analysis. Recall that our approach is exact, and imprecision is
only introduced when estimating the bias 
percentage
(cf. Section~\ref{sec:implementation}).

\begin{toappendix}	
		\begin{table}[t]
		\caption{Analysis of Neural
			Networks Trained on Fair and \{Age, Credit 
			> 1000\}-Biased Data
			(German Credit
			Data) --- Full
			Table (\textsc{boxes} 
			Domain)}\label{tbl:germanTfull0}
\centering
\resizebox{0.75\textwidth}{!}{%
\begin{tabular}{|c|c|cc|cc|c|c|cc|cc|c|}
	\hline
		\multirow{3}{*}{$\textsc{credit}$} &  	
	\multicolumn{12}{c|}{\textsc{boxes} }  
	\\
	 &  
	\multicolumn{6}{c|}{\textsc{fair data} } 
	& 
	\multicolumn{6}{c|}{\textsc{biased data}}
	\\
	& $\upperbound$ &
	\textsc{bias} & \size{\text{C}} & 
	\multicolumn{2}{c|}{$\size{\text{F}}$} & 
	\textsc{time} & $\upperbound$ &
	\textsc{bias} & \size{\text{C}} & 
	\multicolumn{2}{c|}{$\size{\text{F}}$} & 
	\textsc{time} \\
	\hline%

\rowcolor{gray!10}
\cellcolor{mygreen!20} & $8$ & {\small $0.33\%$} 
& 
{\small 
$144$} & {\small $32$} &
{\small $39$} & {\small 7m 7s} 		
& $13$ & {\small $0.79\%$} & {\small $212$} & 
{\small 
$56$} & {\small $78$} & {\small 
1h 56m 57s} \\

\cellcolor{mygreen!20} & $9$ & {\small $0.17\%$} 
& 
{\small 
$182$} & {\small $35$} 
& {\small $56$} & 
{\small 30m 59s}		
& $5$ & {\small $0.31\%$} & {\small $166$} & 
{\small 
$26$} & {\small $49$} & {\small 
2m 17s} \\

\rowcolor{gray!10}
\cellcolor{mygreen!20} & $2$ & {\small $0.09\%$} 
& 
{\small 
$167$} & {\small $13$} 
& {\small $24$} & 
{\small 2m 2s} 		
& $12$ & {\small $0.90\%$} & {\small $202$} & 
{\small 
$67$} & {\small $81$} & {\small 
24m 2s} \\		

\cellcolor{mygreen!20} & $13$ & {\small $0.15\%$} 
& 
{\small 
$157$} & {\small $30 $} 
& {\small $34$} & 
{\small 17m 30s} 		
& $9$ & {\small $0.42\%$} & {\small $187$} & 
{\small 
$47$} & {\small $71$} & {\small 
17m 57s} \\		

\rowcolor{gray!10}
\cellcolor{mygreen!20} & $6$ & {\small $0.23\%$} 
& 
{\small 
$169$} & {\small $36$} 
& {\small $67$} & 
{\small 4m 24s}		
&  $10$ & 
 {\small $0.95\%$} 
& 
{\small $260$} 
& 
{\small 
$88$} & {\small 
$210$} & {\small 
1h 13m 14s} \\

\cellcolor{mygreen!20} & $13$ & {\small $0.30\%$} 
& 
{\small 
$173$} & {\small $57$} 
& {\small $82$} & 
{\small 12m 36s} 		
& $7$ & {\small $0.41\%$} & {\small $190$} & 
{\small 
$56$} & {\small $66$} & {\small 
8m 7s} \\		

\rowcolor{gray!10}
\cellcolor{mygreen!20} & $9$ & {\small $0.20\%$} 
& 
{\small 
$134$} & {\small $30$} 
& {\small $38$} & 
{\small 3m 13s} 		
& $9$ & {\small $0.48\%$} & {\small $189$} & 
{\small 
$39$} & {\small $59$} & {\small 
3m 2s}  \\		

\multirow{-8}{*}{\cellcolor{mygreen!20}$\leq 
1000$}
& $6$ & {\small $0.16\%$} & {\small $172$} & 
{\small 
$16$} & {\small $19$} & {\small 
47s}
& $3$ & {\small $0.09\%$} & {\small $200$} & 
{\small 
$18$} & {\small $21$} & {\small 
5m 23s}
	\\ \cline{2-13}
	
\cellcolor{mygreen!20}\textsc{min} && {\small 
$0.09\%$} &  &  &  & {\small 
		47s} && {\small $0.09\%$} &  &  &  
		& {\small 
		2m 17s} \\
\rowcolor{gray!10}
\cellcolor{mygreen!20}\textsc{median} && {\small 
$0.19\%$} &  &  &  & {\small 
		5m 46s} && {\small $0.45\%$} &  &  & 
		 & {\small 
		13m 2s} \\
\cellcolor{mygreen!20}\textsc{max} && {\small 
$0.33\%$} &  &  &  & {\small 
		30m 59s} && {\small $0.95\%$} &  &  & 
		 & {\small 
		1h 56m 57s} \\

	\hline%
	
\rowcolor{gray!10}
\cellcolor{mypurple!20} & $13$ & {\small 
$12.20\%$ } & 
{\small 
$208$} & 
{\small $76$} & {\small $139$} & 
{\small 53m 27s} 			
& $16$ & {\small $27.59\%$ } & {\small $285$} & 
{\small 
$140$} & {\small $270$} & {\small 
16h 50m 48s} \\	

\cellcolor{mypurple!20} & $15$ & {\small 
$7.43\%$  } & 
{\small 
$211$} & {\small $86$} 
& {\small $185$} & 
{\small 3h 45m 20s}		
& $10$ & {\small $30.77\%$ } & {\small $387$} & 
{\small 
$122$} & {\small $312$} & {\small 
36m 39s} \\	

\rowcolor{gray!10}
\cellcolor{mypurple!20} & $3$ & {\small $2.21\%$  } 
& 
{\small 
$207$} & {\small $23$} 
& {\small $42$} & 
{\small 1m 42s} 		
& $16$ & {\small $33.19\%$ } & {\small $273$} & 
{\small 
$122$} & {\small $260$} & {\small 
16h 49m 33s} \\	

\cellcolor{mypurple!20} & $13$ & {\small $4.29\%$  
} & 
{\small 
$180$} & {\small $45$} 
& {\small $75$} & 
{\small 36m 36s} 		
& $10$ & {\small $16.45\%$ } & {\small $397$} & 
{\small 
$198$} & {\small $389$} & {\small 
2h 25m 20s} \\	

\rowcolor{gray!10}
\cellcolor{mypurple!20} & $7$ & {\small $9.73\%$  } 
& 
{\small 
$433$} & {\small $139$} 
& {\small $329$} & 
{\small 16m 14s}		
&  $14$ & 
 {\small $30.27\%$  
} &  {\small $257$} 
& 
{\small 
$120$} &  {\small 
$253$} &  {\small 
2h 13m 36s} \\

\cellcolor{mypurple!20} & $16$ & {\small 
$14.96\%$ } & 
{\small 
$230$} & 
{\small $92$} & {\small $197$} & 
{\small 7h 7m 12s} 	
& $9$ & {\small $17.24\%$ } & {\small $417$} & 
{\small 
$169$} & {\small $337$} & {\small 
1h 0m 6s} \\	

\rowcolor{gray!10}
\cellcolor{mypurple!20} & $10$ & {\small $6.00\%$  
} & 
{\small 
$243$} & {\small $99$} 
& {\small $145$} & 
{\small 22m 1s} 		
& $11$ & {\small $19.23\%$ } & {\small $288$} & 
{\small 
$99$} & {\small $193$} & {\small 
28m 34s} \\

\multirow{-8}{*}{\cellcolor{mypurple!20}$> 1000$}
& $10$ & {\small $4.61\%$  } & {\small $237$} & 
{\small 
$63$} & {\small $96$} & {\small 
26m 48s}	
& $3$ & {\small $4.52\%$  } & {\small $618$} & 
{\small 
$83$} & {\small $240$} & {\small 
21m 11s} 
	\\\cline{2-13}
	
\cellcolor{mypurple!20} \textsc{min} && {\small 
$2.21\%$} &  &  &  & {\small  
	1m 42s} && {\small $4.52\%$} &  &  &  
	& {\small 
	21m 11s} \\
\rowcolor{gray!10}
\cellcolor{mypurple!20} \textsc{median} && {\small 
$6.72\%$} &  &  &  & {\small  
		31m 42s} && {\small $23.41\%$} &  &  & 
		 & {\small 
		1h 36m 51s} \\
\cellcolor{mypurple!20} \textsc{max} && {\small 
$14.96\%$} &  &  &  & {\small  
	7h 7m 12s} && {\small $33.19\%$} &  &  
	&  & {\small 
	16h 50m 48s} \\
	\hline
\end{tabular}
}
	\end{table}
	
	\begin{table}[t]
		\caption{Analysis of Neural
			Networks Trained on Fair and \{Age, Credit 
			> 1000\}-Biased Data
			(German Credit
			Data) --- Full
			Table (\textsc{symbolic} 
			Domain)}\label{tbl:germanTfull1}
\centering
\resizebox{0.75\textwidth}{!}{%
\begin{tabular}{|c|c|cc|cc|c|c|cc|cc|c|}
	\hline
		\multirow{3}{*}{$\textsc{credit}$} &  	
	\multicolumn{12}{c|}{\textsc{symbolic} }  
	\\
	 &  
	\multicolumn{6}{c|}{\textsc{fair data} } 
	& 
	\multicolumn{6}{c|}{\textsc{biased data}}
	\\
	& $\upperbound$ &
	\textsc{bias} & \size{\text{C}} & 
	\multicolumn{2}{c|}{$\size{\text{F}}$} & 
	\textsc{time} & $\upperbound$ &
	\textsc{bias} & \size{\text{C}} & 
	\multicolumn{2}{c|}{$\size{\text{F}}$} & 
	\textsc{time} \\
	\hline%

\rowcolor{gray!10}
\cellcolor{mygreen!20} &  
$\mathbf{7}$ &  {\small 
$\mathbf{0.33\%}$} &  
{\small 
$\mathbf{138}$} &  {\small 
$\mathbf{22}$} &  
{\small $\mathbf{32}$} &  \cellcolor{mygreen!20} 
{\small \textbf{52s}} 		
& $10$ & {\small $0.79\%$} & {\small $196$} & 
{\small 
$47$} & {\small $56$} & {\small 
7m 
9s} \\

\cellcolor{mygreen!20} & $6$ & {\small $0.17\%$} 
& 
{\small 
$165$} & {\small $19$} 
& {\small $23$} & 
{\small 4m 8s}		
&  $\mathbf{4}$ &  
 {\small $\mathbf{0.31\%}$} 
&   {\small $\mathbf{141}$} 
&  
{\small 
$\mathbf{17}$} &   {\small 
$\mathbf{26}$} &  \cellcolor{mygreen!20} {\small 
\textbf{1m 10s}} \\

\rowcolor{gray!10}
\cellcolor{mygreen!20} & $2$ & {\small $0.09\%$} 
& 
{\small 
$140$} & {\small $8$} 
& {\small $10$} & 
{\small 29s} 		
&  $\mathbf{12}$ &  
 {\small $\mathbf{0.90\%}$} 
&   {\small $\mathbf{198}$} 
& 
  {\small 
$\mathbf{52}$} &   {\small 
$\mathbf{59}$} &  \cellcolor{mygreen!20} {\small 
\textbf{13m 16s}} \\		

\cellcolor{mygreen!20} & $9$ & {\small $0.15\%$} 
& 
{\small 
$159$} & {\small $21 $} 
& {\small $22$} & 
{\small 2m 5s} 		
& $5$ & {\small $0.42\%$} & {\small $194$} & 
{\small 
$28$} & {\small $38$} & {\small 
3m 19s} \\		

\rowcolor{gray!10}
\cellcolor{mygreen!20} & $3$ & {\small $0.23\%$} 
& 
{\small 
$157$} & {\small $14$} 
& {\small $25$} & 
{\small 1m 49s}		
& $8$ & {\small $0.95\%$} & {\small $173$} 
& 
{\small 
$52$} & {\small $77$} & {\small 
10m 40s} \\

\cellcolor{mygreen!20} &  
$\mathbf{8}$ &   {\small 
$\mathbf{0.30\%}$} &  
{\small 
$\mathbf{173}$} &   {\small 
$\mathbf{23}$} 
&   {\small $\mathbf{32}$} &  
\cellcolor{mygreen!20}
{\small \textbf{1m 10s}} 		
& $2$ & {\small $0.41\%$} & {\small $182$} & 
{\small 
$24$} & {\small $33$} & {\small 
2m 3s} \\		

\rowcolor{gray!10}
\cellcolor{mygreen!20} & $6$ & {\small $0.20\%$} 
& 
{\small 
$135$} & {\small $23$} 
& {\small $25$} & 
{\small 1m 0s} 		
& $12$ & {\small $0.48\%$} & {\small $181$} & 
{\small 
$23$} & {\small $39$} & {\small 
1m 21s}  \\		

\multirow{-8}{*}{\cellcolor{mygreen!20}$\leq 
1000$}
& $5$ & {\small $0.16\%$} & {\small $168$} & 
{\small 
$13$} & {\small $14$} & {\small 
13s}
& $2$ & {\small $0.09\%$} & {\small $196$} & 
{\small 
$10$} & {\small $10$} & {\small 
1m 42s}
	\\ \cline{2-13}
	
\cellcolor{mygreen!20}\textsc{min} && {\small 
$0.09\%$} &  &  &  & {\small 
		13s} && {\small $0.09\%$} &  &  &  
		& {\small 
		1m 10s} \\
\rowcolor{gray!10}
\cellcolor{mygreen!20}\textsc{median} && {\small 
$0.19\%$} &  &  &  & {\small 
		1m 5s} && {\small $0.45\%$} &  &  & 
		 & {\small 
		2m 41s} \\
\cellcolor{mygreen!20}\textsc{max} && {\small 
$0.33\%$} &  &  &  & {\small 
		4m 8s} && {\small $0.95\%$} &  &  & 
		 & {\small 
		13m 16s} \\

	\hline%
	
\rowcolor{gray!10}
\cellcolor{mypurple!20} & $12$ & {\small 
$12.20\%$ } & 
{\small 
$202$} & 
{\small $56$} & {\small $101$} & 
{\small 32m 1s} 			
& $13$ & {\small $27.59\%$ } & {\small $412$} & 
{\small 
$189$} & {\small $334$} & {\small 
4h 50m 24s} \\	

\cellcolor{mypurple!20} & $15$ & {\small 
$7.43\%$  } & 
{\small 
$215$} & {\small $60$} 
& {\small $103$} & 
{\small 2h 28m 9s}		
& $6$ & {\small $30.77\%$ } & {\small $371$} & 
{\small 
$75$} & {\small $179$} & {\small 
11m 52s} \\	

\rowcolor{gray!10}
\cellcolor{mypurple!20} &  
$\mathbf{2}$ &  {\small 
$\mathbf{2.21\%}$  } 
& 
{\small 
$\mathbf{161}$} & {\small 
$\mathbf{11}$} 
&  {\small $\mathbf{18}$} & 
\cellcolor{mypurple!20}
{\small \textbf{38s}} 		
& $15$ & {\small $33.19\%$ } & {\small $309$} & 
{\small 
$126$} & {\small $257$} & {\small 
8h 5m 14s} \\	

\cellcolor{mypurple!20} & $9$ & {\small $4.29\%$  
} & 
{\small 
$203$} & {\small $41$} 
& {\small $54$} & 
{\small 6m 53s} 		
& $8$ & {\small $16.45\%$ } & {\small $324$} & 
{\small 
$136$} & {\small $229$} & {\small 
1h 4m 52s} \\	

\rowcolor{gray!10}
\cellcolor{mypurple!20} & 
$\mathbf{3}$ &  {\small 
$\mathbf{9.73\%}$  } 
& 
{\small 
$\mathbf{234}$} &  {\small 
$\mathbf{38}$} 
&  {\small $\mathbf{74}$} & 
\cellcolor{mypurple!20}
{\small \textbf{2m 56s}}		
& $14$ & {\small $30.27\%$  } & {\small $241$} & 
{\small 
$98$} & {\small $219$} & {\small 
1h 39m 34s} \\

\cellcolor{mypurple!20} & $16$ & {\small 
$14.96\%$ } & 
{\small 
$228$} & 
{\small $82$} & {\small $168$} & 
{\small 4h 16m 52s} 	
& $6$ & {\small $17.24\%$ } & {\small $389$} & 
{\small 
$76$} & {\small $162$} & {\small 
18m 36s} \\	

\rowcolor{gray!10}
\cellcolor{mypurple!20} & $8$ & {\small $6.00\%$  } 
& 
{\small 
$261$} & {\small $106$} 
& {\small $80$} & 
{\small 6m 6s} 		
&  $\mathbf{6}$ & 
 {\small $\mathbf{19.23\%}$ 
} & 
 {\small $\mathbf{340}$} & 
{\small 
$\mathbf{66}$} &  {\small 
$\mathbf{134}$} & 
\cellcolor{mypurple!20} {\small 
\textbf{4m 12s}} \\

\multirow{-8}{*}{\cellcolor{mypurple!20}$> 1000$}
& $9$ & {\small $4.61\%$  } & {\small $228$} & 
{\small 
$51$} & {\small $66$} & {\small 
11m 4s}	
&  $\mathbf{2}$ & 
 {\small $\mathbf{4.52\%}$  
} & 
 {\small $\mathbf{325}$} & 
 {\small 
$\mathbf{45}$} &  {\small 
$\mathbf{90}$} & 
\cellcolor{mypurple!20} {\small 
\textbf{11m 4s}} 
	\\\cline{2-13}
	
\cellcolor{mypurple!20} \textsc{min} && {\small 
$2.21\%$} &  &  &  & {\small  
	38s} && {\small $4.52\%$} &  &  &  
	& {\small 
	3m 7s} \\
\rowcolor{gray!10}
\cellcolor{mypurple!20} \textsc{median} && {\small 
$6.72\%$} &  &  &  & {\small  
		8m 59s} && {\small $23.41\%$} &  &  & 
		 & {\small 
		41m 44s} \\
\cellcolor{mypurple!20} \textsc{max} && {\small 
$14.96\%$} &  &  &  & {\small  
	4h 16m 52s} && {\small $33.19\%$} &  &  
	&  & {\small 
	8h 5m 14s} \\
	\hline
\end{tabular}
}
	\end{table}

\begin{table}[t]
	\caption{Analysis of Neural
		Networks Trained on Fair and \{Age, Credit 
		> 1000\}-Biased Data
		(German Credit
		Data) --- Full
		Table (\textsc{deeppoly} 
		Domain)}\label{tbl:germanTfull2}
\centering
\resizebox{0.75\textwidth}{!}{%
\begin{tabular}{|c|c|cc|cc|c|c|cc|cc|c|}
	\hline
		\multirow{3}{*}{$\textsc{credit}$} &  	
	\multicolumn{12}{c|}{\textsc{deeppoly} }  
	\\
	 &  
	\multicolumn{6}{c|}{\textsc{fair data} } 
	& 
	\multicolumn{6}{c|}{\textsc{biased data}}
	\\
	& $\upperbound$ &
	\textsc{bias} & \size{\text{C}} & 
	\multicolumn{2}{c|}{$\size{\text{F}}$} & 
	\textsc{time} & $\upperbound$ &
	\textsc{bias} & \size{\text{C}} & 
	\multicolumn{2}{c|}{$\size{\text{F}}$} & 
	\textsc{time} \\
	\hline%

\rowcolor{gray!10}
\cellcolor{mygreen!20} & $8$ & {\small $0.33\%$} 
& 
{\small 
$170$} & {\small $21$} &
{\small $25$} & {\small 3m 40s} 		
&  $\mathbf{8}$ & 
 {\small $\mathbf{0.79\%}$} 
&  {\small $\mathbf{260}$} & 
 {\small 
$\mathbf{42}$} &  {\small 
$\mathbf{53}$} & \cellcolor{mygreen!20} {\small 
\textbf{5m 
42s}} \\

\cellcolor{mygreen!20} &  
$\mathbf{6}$ &  {\small 
$\mathbf{0.17\%}$} 
& 
{\small 
$\mathbf{211}$} &  {\small 
$\mathbf{10}$} 
&  {\small $\mathbf{10}$} & 
\cellcolor{mygreen!20}
{\small \textbf{4m 5s}}		
& $4$ & {\small $0.31\%$} & {\small $218$} & 
{\small 
$9$} & {\small $20$} & {\small 
1m 6s} \\

\rowcolor{gray!10}
\cellcolor{mygreen!20} &  
$\mathbf{2}$ &  {\small 
$\mathbf{0.09\%}$} 
& 
{\small 
$\mathbf{176}$} &  {\small 
$\mathbf{4}$} 
&  {\small $\mathbf{5}$} & 
\cellcolor{mygreen!20}
{\small \textbf{14s}} 		
& $12$ & {\small $0.82\%$} & {\small $271$} & 
{\small 
$53$} & {\small $61$} & {\small 
18m 18s} \\		

\cellcolor{mygreen!20} & 
$\mathbf{7}$ &  {\small 
$\mathbf{0.15\%}$} 
& 
{\small 
$\mathbf{212}$} &  {\small 
$\mathbf{9} $} 
&  {\small $\mathbf{9}$} & 
\cellcolor{mygreen!20}
{\small \textbf{1m 31s}} 		
&  $\mathbf{4}$ & 
 {\small $\mathbf{0.42\%}$} 
&  {\small $\mathbf{242}$} & 
{\small 
$\mathbf{21}$} &  {\small 
$\mathbf{28}$} & \cellcolor{mygreen!20} {\small 
\textbf{1m 36s}} \\		

\rowcolor{gray!10}
\cellcolor{mygreen!20} & 
$\mathbf{3}$ & {\small 
$\mathbf{0.23\%}$} 
& 
{\small 
$\mathbf{217}$} &  {\small 
$\mathbf{8}$} 
& {\small $\mathbf{15}$} & 
\cellcolor{mygreen!20}
{\small \textbf{32s}}		
& $\mathbf{10}$ & {\small $\mathbf{0.95\%}$} & 
{\small $\mathbf{260}$} 
& 
{\small 
$\mathbf{42}$} & {\small $\mathbf{67}$} & 
\cellcolor{mygreen!20} {\small 
\textbf{3m 2s}} \\

\cellcolor{mygreen!20} & $12$ & {\small $0.30\%$} 
& 
{\small 
$213$} & {\small $17$} 
& {\small $23$} & 
{\small 5m 45s} 		
&  $\mathbf{2}$ & 
 {\small $\mathbf{0.41\%}$} 
&  {\small $\mathbf{226}$} & 
 {\small 
$\mathbf{20}$} &  {\small 
$\mathbf{26}$} & \cellcolor{mygreen!20} {\small 
\textbf{1m 56s}} \\		

\rowcolor{gray!10}
\cellcolor{mygreen!20} &  
$\mathbf{6}$ & {\small 
$\mathbf{0.20\%}$} 
& 
{\small 
$\mathbf{193}$} &  {\small 
$\mathbf{11}$} 
&  {\small $\mathbf{11}$} & 
\cellcolor{mygreen!20}
{\small \textbf{52s}} 		
&  $\mathbf{3}$ &  {\small $\mathbf{0.48\%}$} 
& 
 {\small $\mathbf{228}$} & 
{\small 
$\mathbf{19}$} &  {\small 
$\mathbf{34}$} & 
\cellcolor{mygreen!20} {\small 
\textbf{39s}}  \\		

\multirow{-8}{*}{\cellcolor{mygreen!20}$\leq 
1000$}
&  $\mathbf{5}$ & 
 {\small $\mathbf{0.16\%}$} 
& 
 {\small $\mathbf{193}$} & 
{\small 
$\mathbf{9}$} &  {\small 
$\mathbf{10}$} & 
\cellcolor{mygreen!20} {\small 
\textbf{10s}}
&  $\mathbf{1}$ & 
 {\small $\mathbf{0.09\%}$} 
& 
 {\small $\mathbf{206}$} & 
{\small 
$\mathbf{5}$} &  {\small 
$\mathbf{5}$} & 
\cellcolor{mygreen!20} {\small 
\textbf{51s}}
	\\ \cline{2-13}
	
\cellcolor{mygreen!20}\textsc{min} && {\small 
$0.09\%$} &  &  &  & {\small 
		10s} && {\small $0.09\%$} &  &  &  
		& {\small 
		39s} \\
\rowcolor{gray!10}
\cellcolor{mygreen!20}\textsc{median} && {\small 
$0.19\%$} &  &  &  & {\small 
		1m 12s} && {\small $0.45\%$} &  &  & 
		 & {\small 
		1m 46s} \\
\cellcolor{mygreen!20}\textsc{max} && {\small 
$0.33\%$} &  &  &  & {\small 
		5m 45s} && {\small $0.95\%$} &  &  & 
		 & {\small 
		18m 18s} \\

	\hline%
	
\rowcolor{gray!10}
\cellcolor{mypurple!20} &  
$\mathbf{10}$ &  {\small 
$\mathbf{12.08\%}$ } & 
 {\small 
$\mathbf{321}$} & 
{\small $\mathbf{85}$} & 
{\small 
$\mathbf{150}$} & \cellcolor{mypurple!20}
{\small \textbf{10m 30s}} 			
&  $\mathbf{11}$ & 
 {\small $\mathbf{27.59\%}$ 
} & 
 {\small $\mathbf{498}$} & 
{\small 
$\mathbf{234}$} &  {\small 
$\mathbf{333}$} & 
\cellcolor{mypurple!20} {\small 
\textbf{1h 16m 41s}} \\	

\cellcolor{mypurple!20} &  
$\mathbf{11}$ &  {\small 
$\mathbf{7.43\%}$  } 
& 
{\small 
$\mathbf{329}$} &  {\small 
$\mathbf{75}$} 
& {\small $\mathbf{125}$} 
& 
\cellcolor{mypurple!20}
{\small \textbf{22m 33s}}		
&  $\mathbf{7}$ & 
 {\small $\mathbf{30.77\%}$ 
} & 
 {\small $\mathbf{394}$} & 
 {\small 
$\mathbf{70}$} &  {\small 
$\mathbf{228}$} & 
\cellcolor{mypurple!20} {\small 
\textbf{6m 34s}} \\	

\rowcolor{gray!10}
\cellcolor{mypurple!20} & $2$ & {\small $2.21\%$  
} & 
{\small 
$217$} & {\small $15$} 
& {\small $16$} & 
{\small 39s} 		
& $\mathbf{7}$ & 
 {\small $\mathbf{33.17\%}$ 
} & 
 {\small $\mathbf{435}$} & 
 {\small 
$\mathbf{185}$} &  {\small 
$\mathbf{327}$} & 
\cellcolor{mypurple!20} {\small 
\textbf{6h 51m 50s}} \\	

\cellcolor{mypurple!20} &  
$\mathbf{10}$ &  {\small 
$\mathbf{4.29\%}$  } 
& 
{\small 
$\mathbf{239}$} &  {\small 
$\mathbf{24}$} 
&  {\small $\mathbf{33}$} & 
\cellcolor{mypurple!20}
{\small \textbf{4m 4s}} 		
&  $\mathbf{6}$ & 
{\small $\mathbf{16.45\%}$ 
} & 
 {\small $\mathbf{448}$} & 
 {\small 
$\mathbf{162}$} &  {\small 
$\mathbf{260}$} & 
\cellcolor{mypurple!20} {\small 
\textbf{18m 25s}} \\	

\rowcolor{gray!10}
\cellcolor{mypurple!20} & $4$ & {\small $9.73\%$  
} & 
{\small 
$268$} & {\small $29$} 
& {\small $87$} & 
{\small 4m 0s}		
& $\mathbf{13}$ & {\small $\mathbf{30.17\%}$  } & 
{\small $\mathbf{418}$} & 
{\small 
$\mathbf{141}$} & {\small $\mathbf{332}$} & 
\cellcolor{mypurple!20} {\small 
\textbf{43m 12s}} \\

\cellcolor{mypurple!20} & 
$\mathbf{14}$ &  {\small 
$\mathbf{14.96\%}$ } 
& 
{\small 
$\mathbf{403}$} & 
{\small $\mathbf{116}$} &  
{\small 
$\mathbf{231}$} & \cellcolor{mypurple!20}
{\small \textbf{1h 9m 45s}} 	
&  $\mathbf{5}$ & 
 {\small $\mathbf{17.24\%}$ 
} & 
 {\small $\mathbf{460}$} & 
 {\small 
$\mathbf{91}$} & {\small 
$\mathbf{217}$} & 
\cellcolor{mypurple!20} {\small 
\textbf{12m 53s}} \\	

\rowcolor{gray!10}
\cellcolor{mypurple!20} &  
$\mathbf{7}$ &  {\small 
$\mathbf{5.83\%}$  } & 
 {\small 
$\mathbf{313}$} &  {\small 
$\mathbf{92}$} 
&  {\small $\mathbf{115}$} 
& 
\cellcolor{mypurple!20} {\small \textbf{4m 
17s}} 		
& $8$ & {\small $19.23\%$ } & {\small $363$} & 
{\small 
$79$} & {\small $189$} & {\small 
7m 24s} \\

\multirow{-8}{*}{\cellcolor{mypurple!20}$> 1000$}
&  $\mathbf{9}$ & 
 {\small $\mathbf{4.61\%}$  
} & 
 {\small $\mathbf{264}$} & 

{\small 
$\mathbf{50}$} &  {\small 
$\mathbf{74}$} & 
\cellcolor{mypurple!20} {\small 
\textbf{5m 38s}}	
& $2$ & {\small $4.52\%$  } & {\small $331$} & 
{\small 
$45$} & {\small $95$} & {\small 
4m 44s} 
	\\\cline{2-13}
	
\cellcolor{mypurple!20} \textsc{min} && {\small 
$2.21\%$} &  &  &  & {\small  
	39s} && {\small $4.52\%$} &  &  &  
	& {\small 
	4m 44s} \\
\rowcolor{gray!10}
\cellcolor{mypurple!20} \textsc{median} && {\small 
$6.63\%$} &  &  &  & {\small  
		4m 58s} && {\small $23.41\%$} &  &  & 
		 & {\small 
		15m 39s} \\
\cellcolor{mypurple!20} \textsc{max} && {\small 
$14.96\%$} &  &  &  & {\small  
	1h 9m 45s} && {\small $31.17\%$} &  &  
	&  & {\small 
	6h 51m 50s} \\
	\hline
\end{tabular}
}
\end{table}

\subsection{RQ1: Detecting Seeded Bias}
Tables~\ref{tbl:germanTfull0},~\ref{tbl:germanTfull1}
and~\ref{tbl:germanTfull2} show the analysis 
results for all eight models trained on the German 
Credit
dataset.
Column $\upperbound$ shows the chosen 
	upper 
	bound for each model. 
As before,
	column \textsc{bias} shows the detected bias, 
in 
	percentage of
	the entire input space. We also again show 
	minimum, 
	median, and maximum bias percentage for each 
	credit 
	request group.
	Column $\size{\text{C}}$ shows the 
	total number of analyzed (i.e., completed) input 
	space partitions. Column $\size{\text{F}}$ 
shows 
	the total number of abstract activation patterns 
	(left) and feasible input partitions (right) that the 
	backward analysis had to explore. 
	Finally, column \textsc{time} shows the analysis 
	running time. Again, we also show minimum, 
	median, and maximum running time for each 
	credit 
	request group.
For all models, we 
highlighted across all tables the choice 
of the abstract domain that 
entailed the shortest 
analysis time.
\end{toappendix}

\paragraph{\textbf{RQ2: Answering Bias 
Queries.}}
To further evaluate the precision of our approach, we created queries
concerning bias within specific groups of people, 
each corresponding
to a subset of the entire input space.
We used the \textsc{compas}
dataset\footnote{\url{https://www.propublica.org/datastore/dataset/compas-recidivism-risk-score-data-and-analysis}}
from ProPublica for this experiment. The data assigns a
three-valued recidivism-risk score (high, medium, 
and low) indicating how likely criminals are to 
re-offend. 
The data includes both
personal attributes (e.g., age and race) as well as criminal history (e.g.,
number of priors and violent crimes). As for RQ1, we trained models
both on fair and biased data. Here, we considered race as the
sensitive feature. We seeded bias in the fair data by randomly
assigning high recidivism risk to African Americans until we reached a
20\% causal-unfairness score of the dataset.
The median classification accuracy of the 3-class models (19 inputs and 4 hidden layers with 5 nodes each) trained on
fair and biased data was 55\% and 56\%, 
respectively. Accuracy does
not improve with larger networks --- we tried up to
100 hidden layers with 100 nodes each.

\begin{table}[t]
	\caption{Queries on Models Trained
		on Fair and Race-Biased Data
		(ProPublica's \textsc{compas}
		Data)}\label{tbl:compasT}
\centering
\resizebox{\textwidth}{!}{%
\begin{tabular}{|c|cc|cc|cc|cc|cc|cc|c}
\cline{1-13}
\multirow{3}{*}{$\textsc{query}$} & 
\multicolumn{4}{c|}{\textsc{boxes}} & 
\multicolumn{4}{c|}{\textsc{symbolic}} & 
\multicolumn{4}{c|}{\textsc{deeppoly}} &
\\
 &  
 \multicolumn{2}{c|}{\textsc{fair data} } 
 & 
 \multicolumn{2}{c|}{\textsc{biased data}} &
\multicolumn{2}{c|}{\textsc{fair data} } 
& 
\multicolumn{2}{c|}{\textsc{biased data}} &
\multicolumn{2}{c|}{\textsc{fair data} } 
& 
\multicolumn{2}{c|}{\textsc{biased data}} &
\\
& 
\textsc{bias} & 
	\textsc{time} &
\textsc{bias} & 
	\textsc{time} & 
\textsc{bias} & 
\textsc{time} &
\textsc{bias} & 
\textsc{time} & 
\textsc{bias} & 
\textsc{time} &
\textsc{bias} & 
\textsc{time} &\\
\hline%
\rowcolor{gray!10}
\cellcolor{myyellow!20}
 & {\small $0.22\%$} & {\small 
	24m 32s} & {\small $0.12\%$} & {\small 
	14m 53s} & 
{\small $0.22\%$} & {\small 
	11m 34s} &  {\small 
	$\mathbf{0.12\%}$} & \cellcolor{myyellow!20} 
	{\small 
	\textbf{7m 14s}} & 
 {\small $\mathbf{0.22\%}$} 
& \cellcolor{myyellow!20} {\small 
	\textbf{5m 18s}} & {\small $0.12\%$} & {\small 
	8m 46s} & 
	\multicolumn{1}{c|}{\textsc{min}} \\	
\multirow{-2}{*}{\cellcolor{myyellow!20}$\textsc{age}
	< 
	25$} 
& {\small $0.31\%$} & 
{\small 
	1h 54m 48s} & {\small $0.99\%$} & {\small 
	57m 33s} & 
 {\small $\mathbf{0.32\%}$} 
& \cellcolor{myyellow!20}
{\small 
	\textbf{36m 0s}} & {\small $0.99\%$} & {\small 
	20m 43s} & 
{\small $0.32\%$} & {\small 
	47m 16s} &  {\small 
	$\mathbf{0.99\%}$} & \cellcolor{myyellow!20} 
	{\small 
	\textbf{16m 38s}} & 
\multicolumn{1}{c|}{\textsc{median}} 
	\\	
\rowcolor{gray!10}
\multirow{-2}{*}{\cellcolor{myyellow!20}$\textsc{race
 bias}?$}
& {\small $2.46\%$} & {\small 
	2h 44m 11s} & {\small $8.33\%$} & {\small 
	5h 29m 19s} &
{\small $2.46\%$} & {\small 
	2h 17m 3s} & {\small $8.50\%$} & {\small 
	3h 34m 50s} &
 {\small $\mathbf{2.12\%}$} 
& \cellcolor{myyellow!20} {\small 
	\textbf{1h 11m 43s}} & 
	{\small $\mathbf{6.48\%}$} & 
	\cellcolor{myyellow!20} {\small 
	\textbf{2h 5m 5s}} & 
\multicolumn{1}{c|}{\textsc{max}} 
	\\	
\hline%
\cellcolor{mygreen!20}
 & 
 {\small $2.60\%$} & {\small 
 	24m 14s} & {\small $4.51\%$} & {\small 
 	34m 23s} &
{\small $2.64\%$} & {\small 
	25m 13s} & {\small $5.20\%$} & {\small 
	29m 19s} &
 {\small $\mathbf{2.70\%}$} 
& \cellcolor{mygreen!20} {\small 
	\textbf{19m 47s}} & 
	{\small $\mathbf{5.22\%}$} & 
	\cellcolor{mygreen!20} {\small 
	\textbf{20m 51s}} & 
\multicolumn{1}{c|}{\textsc{min}} \\	
\rowcolor{gray!10}
\cellcolor{mygreen!20} & {\small $6.08\%$} & 
{\small 
	1h 49m 42s} & {\small $6.95\%$} & {\small 
	2h 3m 39s} &
 {\small $\mathbf{6.77\%}$} 
& \cellcolor{mygreen!20} {\small 
	\textbf{1h 1m 51s}} & {\small $7.02\%$} & 
	{\small 
	1h 2m 26s} &
{\small $6.77\%$} & {\small 
	1h 13m 31s} &  {\small 
	$\mathbf{7.00\%}$} & \cellcolor{mygreen!20} 
	{\small 
	\textbf{47m 28s}} & 
\multicolumn{1}{c|}{\textsc{median}} 
	\\	
\multirow{-3}{*}{\cellcolor{mygreen!20}$\begin{matrix}\textsc{male}\\
	\textsc{age bias}?\end{matrix}$}
& {\small $8.00\%$} & {\small 
	5h 56m 6s} & {\small $12.56\%$} & {\small 
	8h 26m 55s} &
 {\small $\mathbf{8.40\%}$} 
& \cellcolor{mygreen!20} {\small 
	\textbf{2h 2m 22s}} & {\small $12.71\%$} & 
	{\small 
	4h 55m 35s} &
{\small $8.84\%$} & {\small 
	2h 20m 23s} &  {\small 
	$\mathbf{12.88\%}$} & \cellcolor{mygreen!20} 
	{\small 
	\textbf{3h 25m 21s}} & 
\multicolumn{1}{c|}{\textsc{max}} 
	\\	
\hline%
\rowcolor{gray!10}
\cellcolor{mypurple!20} 
& {\small $2.18\%$} & 
{\small 2h 54m 18s} & {\small $2.92\%$} & 
{\small 
	46m 53s} &
{\small $2.18\%$} & 
{\small 1h 20m 41s} & {\small $2.92\%$} & 
{\small 
	30m 23s} &
  {\small $\mathbf{2.18\%}$} 
& \cellcolor{mypurple!20}  {\small 
	\textbf{18m 26s}} &  
	{\small 
	$\mathbf{2.92\%}$} & \cellcolor{mypurple!20}  
	{\small 
	\textbf{15m 29s}} & 
\multicolumn{1}{c|}{\textsc{min}} \\	
\multirow{-2}{*}{\cellcolor{mypurple!20}$\textsc{caucasian}$}
& {\small $2.95\%$} & 
{\small 6h 56m 44s} & {\small $4.21\%$} & 
{\small 
	3h 50m 38s} & 
{\small $2.95\%$} & 
{\small 4h 12m 28s} & {\small $4.21\%$} & 
{\small 
	3h 32m 52s} & 
  {\small $\mathbf{2.95\%}$} 
& \cellcolor{mypurple!20}  {\small 
	\textbf{2h 36m 1s}} & 
	{\small 
	$\mathbf{4.21\%}$} & \cellcolor{mypurple!20}  
	{\small 
	\textbf{1h 34m 7s}} & 
	\multicolumn{1}{c|}{\textsc{median}} \\	
\rowcolor{gray!10}
\multirow{-2}{*}{\cellcolor{mypurple!20}$\textsc{priors
 bias}?$}
& {\small $\mathbf{5.36\%}$} & 
\cellcolor{mypurple!20}
{\small \textbf{45h 2m 12s}} & {\small $6.98\%$} & 
{\small 
	70h 50m 10s} & 
{\small $5.36\%$} & 
{\small 60h 53m 6s} & {\small $6.98\%$} & 
{\small 
	49h 51m 42s} & 
  {\small $5.36\%$} 
& {\small 
	52h 10m 2s} &  
	{\small 
	$\mathbf{6.95\%}$} & \cellcolor{mypurple!20}  
	{\small 
	\textbf{17h 48m 22s}} & 
	\multicolumn{1}{c|}{\textsc{max}} \\	
\hline
\end{tabular}
}
\end{table}

To analyze these models, we used a lower bound 
$\lowerbound$ of 0, and an
upper bound $\upperbound$ between 7 and 19.
Table~\ref{tbl:compasT} shows the results of our 
analysis (i.e.,
columns shown as in Table~\ref{tbl:germanT})
for three queries:
\begin{description}
\item[$Q_A$:] Is there bias with respect to 
\emph{race}
for people younger than 25?
\item[$Q_B$:] Is there bias 
with respect to \emph{age}
for males?
\item[$Q_C$:] Is there bias 
with respect to the number of priors 
for Caucasians?
\end{description}
%
%
For $Q_A$, the analysis detects only a small 
percentage of race bias in the
fair models, but as intended, the race bias is 
found to be
  significantly higher (about three times as much 
  median bias) for the biased models. 
In contrast,
for $Q_B$, the analysis finds a comparable amount 
of age bias across both sets of models. This 
becomes more evident when scaling the median 
bias with respect to the analyzed input space (males 
correspond to $50\%$ of the input space): the 
smallest median bias for the models trained on fair 
data is 
$12.16\%$ (for the \textsc{boxes} domain) and the 
largest median bias for the models trained on biased 
data is 
$14.04\%$ (for the \textsc{symbolic} domain).
This bias is not intended and was either present in 
the original data or introduced by the 
training process (or both).
Finally, for $Q_C$, the
  analysis detects significant bias across both sets 
  of models with respect to the number of
  priors. 
When considering the analyzed input space 
(Caucasians represent $1/6$ of the entire
input 
space), this translates to $17.7\%$ median bias for 
the models trained on fair data and $25.26\%$ for 
the models trained on biased data.
This
bias is intended and present in the original data: as 
one would expect,
recidivism risk differs for different numbers of priors.
Overall, these results \emph{demonstrate the effectiveness of our
  analysis in answering specific bias queries}.

\begin{toappendix}
	\begin{table}[t]
		\caption{Queries on Neural Networks Trained
			on Fair and Race-Biased Data
			(ProPublica's \textsc{compas}
			Data) --- Full
			Table (\textsc{boxes} 
			Domain)}\label{tbl:compasTfull0}
\centering
\resizebox{0.75\textwidth}{!}{%
\begin{tabular}{|c|c|cc|cc|c|c|cc|cc|c|}
	\hline
	\multirow{3}{*}{$\textsc{query}$} &  	
	\multicolumn{12}{c|}{\textsc{boxes} }  
	\\
	& \multicolumn{6}{c|}{\textsc{fair data} } 
	& 
	\multicolumn{6}{c|}{\textsc{biased data}}
	\\
	& $\upperbound$ &
	\textsc{bias} & \size{\text{C}} & 
	\multicolumn{2}{c|}{$\size{\text{F}}$} & 
	\textsc{time} & $\upperbound$ &
	\textsc{bias} & \size{\text{C}} & 
	\multicolumn{2}{c|}{$\size{\text{F}}$} & 
	\textsc{time} \\
	\hline%

	\rowcolor{gray!10}
\cellcolor{myyellow!20} 
	& $10$ & {\small $0.22\%$} & {\small $93$} & 
	{\small $83$} & {\small $46$} &  {\small 2h 0m 
	58s} 		
& $10$ & {\small $0.83\%$} & {\small $65$} & 
{\small $27$} & {\small $64$} &  {\small 5h 29m 
19s} 
\\

\cellcolor{myyellow!20}
& $10$ & {\small $0.64\%$} & {\small $98$} & 
{\small $60$} & {\small $95$} &  {\small 1h 48m 
37s  
}		
& $10$ & {\small $8.33\%$} & {\small $66$} & 
{\small $37$} & {\small $65$} &  {\small 27m 14s} 
\\
\rowcolor{gray!10}
\cellcolor{myyellow!20}
& $10$ & {\small $0.22\%$} & {\small $51$} & 
{\small $22$} & {\small $30$} &  {\small 24m 32s   
} 		
& $10$ & {\small $1.15\%$} & {\small $28$} & 
{\small $12$} & {\small $20$} &  {\small 14m 53s} 
\\		
\cellcolor{myyellow!20}
& $10$ & {\small $0.23\%$} & {\small $191$} & 
{\small $85$} & {\small $104$} &  {\small 2h 44m 
11s   
} 		
& $10$ & {\small $0.42\%$} & {\small $21$} & 
{\small $12$} & {\small $20$} &  {\small 44m 55s} 
\\		
\rowcolor{gray!10}
\cellcolor{myyellow!20}
& $10$ & {\small $0.29\%$} & {\small $221$} & 
{\small $113$} & {\small $169$} &  {\small 2h 34m 
6s   
}		
& $10$ & {\small $0.12\%$} & {\small $70$} & 
{\small $34$} & {\small $69$} &  {\small 26m 0s} 
\\
\cellcolor{myyellow!20}
& $10$ & {\small $0.33\%$} & {\small $107$} & 
{\small $56$} & {\small $84$} &  {\small 2h 30m 
28s   
} 		
& $10$ & {\small $1.54\%$} & {\small $60$} & 
{\small $33$} & {\small $59$} &  {\small 1h 17m 
56s} 
\\		
\rowcolor{gray!10}
\cellcolor{myyellow!20}
& $10$ & {\small $1.19\%$} & {\small $70$} & 
{\small $28$} & {\small $69$} &  {\small 41m 20s   
} 		
& $10$ & {\small $3.25\%$} & {\small $206$} & 
{\small $155$} & {\small $205$} &  {\small 1h 10m 
10s} 
\\		
\multirow{-8}{*}{\cellcolor{myyellow!20}$\begin{matrix}\textsc{age}
	 < 
	25\\\textsc{race bias}?\end{matrix}$}
& $10$ & {\small $2.46\%$} & {\small $32$} & 
{\small $20$} & {\small $31$} &  {\small 36m 6s   }
& $10$ & {\small $0.18\%$} & {\small $28$} & 
{\small $13$} & {\small $27$} &  {\small 3h 8m 
10s} 
	\\\cline{2-13}
\cellcolor{myyellow!20}\textsc{min} && {\small 
$0.22\%$} &&&& {\small 
	24m 32s} && {\small $0.12\%$} &&&& {\small 
	14m 53s} \\	
\rowcolor{gray!10}
\cellcolor{myyellow!20}\textsc{median} && {\small 
$0.31\%$} &&&& {\small 
	1h 54m 48s} && {\small $0.99\%$} &&&& {\small 
	57m 33s} \\	
\cellcolor{myyellow!20}\textsc{max} && {\small 
$2.46\%$} &&&& {\small 
	2h 44m 11s} && {\small $8.33\%$} &&&& {\small 
	5h 29m 19s} \\	
	\hline%

	\rowcolor{gray!10}
\cellcolor{mygreen!20}
	&   $\mathbf{10}$ 
	&   {\small 
	$\mathbf{3.68\%}$} &  
	{\small $\mathbf{282}$} 
	& 
	{\small 
	$\mathbf{115}$} & 
	{\small $\mathbf{281}$} &  
	\cellcolor{mygreen!20} 
	{\small  \textbf{1h 23m 52s}       }
& $10$ & {\small $4.51\%$} & {\small $336$} & 
{\small 
$111$} & {\small $335$} &  {\small  
6h 36m 43s    } 
\\	
\cellcolor{mygreen!20} & $10$ & {\small $7.00\%$} 
& 
{\small 
$358$} & {\small $117$} & {\small 
$357$} &  
{\small  1h 55m 55s     }
& $10$ & {\small $12.56\%$} & {\small $478$} & 
{\small 
$135$} & {\small $477$} &  {\small  
1h 3m 18s   } \\	
\rowcolor{gray!10}
\cellcolor{mygreen!20} 
&  $\mathbf{10}$ 
&   {\small 
$\mathbf{7.92\%}$} &   
{\small $\mathbf{237}$} & 
  {\small 
$\mathbf{57}$} &  {\small 
$\mathbf{232}$} &  \cellcolor{mygreen!20}  {\small  
\textbf{24m 14s}      }
& $10$ & {\small $7.00\%$} & {\small $179$} & 
{\small 
$57$} & {\small $172$} &  {\small  
34m 23s   } \\	
\cellcolor{mygreen!20} 
& $10$ & {\small $2.60\%$} & 
{\small $776$} 
& 
{\small 
$265 $} & {\small $478$} &  {\small  
3h 24m 34s       }
& $10$ & {\small $6.90\%$} & {\small $119$} & 
{\small 
$35$} & {\small $118$} &  {\small  
4h 1m 53s    } \\	
\rowcolor{gray!10}
\cellcolor{mygreen!20} 
& $10$ & {\small $4.29\%$} & {\small $1175$} & 
{\small 
$410$} & {\small $951$} &  {\small  
3h 32m 8s    }
& $10$ & {\small $4.96\%$} & {\small $303$} & 
{\small 
$75$} & {\small $264$} &  {\small  
1h 41m 39s   } \\	
\cellcolor{mygreen!20} 
& $10$ & {\small $5.16\%$} & {\small $397$} & 
{\small 
$100$} & {\small $306$} &  {\small  
5h 56m 6s   }
& $10$ & {\small $7.89\%$} & {\small $294$} & 
{\small 
$50$} & {\small $293$} &  {\small  
8h 26m 55s } 
\\	
\rowcolor{gray!10}
\cellcolor{mygreen!20} 
& $10$ & {\small $7.54\%$} & {\small $338$} & 
{\small 
$84$} & {\small $337$} &  {\small  
1h 1m 14s     }
& $10$ & {\small $6.31\%$} & {\small $484$} & 
{\small 
$228$} & {\small $483$} &  {\small  
1h 46m 51s   } \\	
	\multirow{-8}{*}{\cellcolor{mygreen!20} 
	$\begin{matrix}\textsc{male}\\\textsc{age
		bias}?\end{matrix}$}
& $10$ & {\small $8.00\%$} & {\small $415$} & 
{\small 
$103$} & {\small $414$} &  {\small  
1h 43m 28s       }
& $10$ & {\small $12.24\%$} & {\small $377$} & 
{\small 
$143$} & {\small $376$} &  {\small  
2h 20m 27s   }
	\\\cline{2-13}
\cellcolor{mygreen!20}\textsc{min} && {\small 
$2.60\%$} &&&& {\small 
	24m 14s} && {\small $4.51\%$} &&&& {\small 
	34m 23s} \\	
\rowcolor{gray!10}
\cellcolor{mygreen!20} \textsc{median} && {\small 
$6.08\%$} &&&& {\small 
	1h 49m 42s} && {\small $6.95\%$} &&&& {\small 
	2h 3m 39s} \\	
\cellcolor{mygreen!20} \textsc{max} && {\small 
$8.00\%$} &&&& {\small 
	5h 56m 6s} && {\small $12.56\%$} &&&& 
	{\small 
	8h 26m 55s} \\	
	\hline%
	
	\rowcolor{gray!10}
\cellcolor{mypurple!20}
	& $12$ & {\small $2.18\%$} & {\small $75$} & 
	{\small 
		$29$} & {\small $74$} &  {\small  7h 3m 17s     
		}
& $14$ & {\small $2.92\%$} & {\small $35$} & 
{\small 
	$16$} & {\small $34$} &  {\small     5h 22m 42s  
} 
\\	
\cellcolor{mypurple!20}
& $14$ & {\small $3.66\%$} & {\small $76$} & 
{\small 
	$39$} & {\small $75$} &  {\small     6h 50m 
	10s     
}
& $14$ & {\small $6.98\%$} & {\small $53$} & 
{\small 
	$23$} & {\small $52$} &  {\small       1h 38m 57s
} 
\\	
\rowcolor{gray!10}
\cellcolor{mypurple!20}
& $14$ & {\small $2.73\%$} & {\small $51$} & 
{\small 
	$25$} & {\small $46$} &  {\small    2h 54m 18s   }
& $13$ & {\small $4.43\%$} & {\small $40$} & 
{\small 
	$11$} & {\small $39$} &  {\small       1h 8m 37s  } 
\\	
\cellcolor{mypurple!20}
& $18$ & {\small $2.19\%$} & {\small $46$} & 
{\small 
	$16$} & {\small $45$} &  {\small     37h 15m 28s
}
& $9$ & {\small $3.40\%$} & {\small $67$} & 
{\small 
	$23$} & {\small $66$} &  {\small       46m 53s   } 
\\	
\rowcolor{gray!10}
\cellcolor{mypurple!20}
& $\mathbf{19}$ & {\small $\mathbf{3.17\%}$} & 
{\small $\mathbf{34}$} & 
{\small 
	$\mathbf{11}$} & {\small $\mathbf{33}$} &  
	\cellcolor{mypurple!20} {\small     \textbf{45h 
	2m 12s}         }
& $14$ & {\small $3.09\%$} & {\small $54$} & 
{\small 
	$21$} & {\small $53$} &  {\small   2h 29m 32s
} 
\\	
\cellcolor{mypurple!20}
& $12$ & {\small $2.45\%$} & {\small $128$} & 
{\small 
	$42$} & {\small $110$} &  {\small    8h 41m 
	43s    }
& $15$ & {\small $5.79\%$} & {\small $57$} & 
{\small 
	$32$} & {\small $56$} &  {\small         5h 11m 44s
} 
\\	
\rowcolor{gray!10}
\cellcolor{mypurple!20}
& $13$ & {\small $3.94\%$} & {\small $62$} & 
{\small 
	$28$} & {\small $61$} &  {\small     3h 7m 59s   }
& $19$ & {\small $5.10\%$} & {\small $47$} & 
{\small 
	$30$} & {\small $46$} &  {\small      70h 50m 10s
} \\	
	
\multirow{-8}{*}{\cellcolor{mypurple!20}$\begin{matrix}\textsc{caucasian}\\\textsc{priors
		bias}?\end{matrix}$}
& $15$ & {\small $5.36\%$} & {\small $49$} & 
{\small 
	$20$} & {\small $48$} &  
	{\small          6h 16m 33s     
}
& $17$ & {\small $3.99\%$} & {\small $46$} & 
{\small 
	$22$} & {\small $45$} &  {\small      15h 1m 10s  }
	\\\cline{2-13}
\cellcolor{mypurple!20}\textsc{min} && {\small 
$2.18\%$} &&&& 
{\small 2h 54m 18s} && {\small $2.92\%$} &&&& 
{\small 
	46m 53s} \\	
\rowcolor{gray!10}
\cellcolor{mypurple!20}\textsc{median} && {\small 
$2.95\%$} &&&& 
	{\small 6h 56m 44s} && {\small $4.21\%$} &&&& 
	{\small 
	3h 50m 28s} \\	
\cellcolor{mypurple!20}\textsc{max} && {\small 
$5.36\%$} &&&& 
{\small 45h 2m 12s} && {\small $6.98\%$} &&&& 
{\small 
	70h 50m 10s} \\	
	\hline
\end{tabular}
}
	\end{table}
	\begin{table}[t]
		\caption{Queries on Neural Networks Trained
			on Fair and Race-Biased Data
			(ProPublica's \textsc{compas}
			Data) --- Full
			Table (\textsc{symbolic} 
			Domain)}\label{tbl:compasTfull1}
\centering
\resizebox{0.75\textwidth}{!}{%
\begin{tabular}{|c|c|cc|cc|c|c|cc|cc|c|}
	\hline
	\multirow{3}{*}{$\textsc{query}$} &  	
	\multicolumn{12}{c|}{\textsc{symbolic} }  
	\\
	& \multicolumn{6}{c|}{\textsc{fair data} } 
	& 
	\multicolumn{6}{c|}{\textsc{biased data}}
	\\
	& $\upperbound$ &
	\textsc{bias} & \size{\text{C}} & 
	\multicolumn{2}{c|}{$\size{\text{F}}$} & 
	\textsc{time} & $\upperbound$ &
	\textsc{bias} & \size{\text{C}} & 
	\multicolumn{2}{c|}{$\size{\text{F}}$} & 
	\textsc{time} \\
	\hline%

	\rowcolor{gray!10}
\cellcolor{myyellow!20} 
	& $10$ & {\small $0.23\%$} & {\small $57$} & 
	{\small $17$} & {\small $24$} &  {\small 2h 17m 
	3s} 		
&   $\mathbf{10}$ & 
  {\small $\mathbf{0.83\%}$} 
&  {\small $\mathbf{25}$} & 
  {\small $\mathbf{11}$} & 
  {\small $\mathbf{24}$} & 
\cellcolor{myyellow!20}  {\small \textbf{1h 32m 
10s}} 
\\

\cellcolor{myyellow!20}
& $10$ & {\small $0.75\%$} & {\small $44$} & 
{\small $20$} & {\small $24$} &  {\small 19m 16s   
}		
& $10$ & {\small $8.50\%$} & {\small $60$} & 
{\small $27$} & {\small $34$} &  {\small 18m 48s} 
\\
\rowcolor{gray!10}
\cellcolor{myyellow!20}
&   $\mathbf{10}$ & 
 {\small $\mathbf{0.22\%}$} 
&  {\small $\mathbf{41}$} & 
 {\small $\mathbf{13}$} & 
 {\small $\mathbf{15}$} &  
\cellcolor{myyellow!20} {\small \textbf{11m 34s}   
} 		
& $10$ & {\small $1.15\%$} & {\small $24$} & 
{\small $8$} & {\small $14$} &  {\small 19m 50s} 
\\		
\cellcolor{myyellow!20}
& $10$ & {\small $0.26\%$} & {\small $122$} & 
{\small $33$} & {\small $36$} &  {\small 54m 19s   
} 		
&   $\mathbf{10}$ & 
 {\small $\mathbf{0.42\%}$} 
&  {\small $\mathbf{17}$} & 
 {\small $\mathbf{13}$} & 
 {\small $\mathbf{16}$} &  
\cellcolor{myyellow!20} {\small \textbf{7m 14s}} 
\\		
\rowcolor{gray!10}
\cellcolor{myyellow!20}
&  $\mathbf{10}$ & 
 {\small $\mathbf{0.30\%}$} 
&  {\small $\mathbf{148}$} 
& 
 {\small $\mathbf{54}$} & 
 {\small $\mathbf{63}$} & 
\cellcolor{myyellow!20}  {\small \textbf{50m 54s}   
}		
& $10$ & {\small $0.12\%$} & {\small $32$} & 
{\small $9$} & {\small $14$} &  {\small 21m 35s} 
\\
\cellcolor{myyellow!20}
&  $\mathbf{10}$ & 
  {\small $\mathbf{0.33\%}$} 
&   {\small $\mathbf{64}$} & 
 {\small $\mathbf{15}$} & 
 {\small $\mathbf{17}$} & 
\cellcolor{myyellow!20}  {\small \textbf{53m 14s}   
} 		
& $10$ & {\small $1.59\%$} & {\small $34$} & 
{\small $19$} & {\small $30$} &  {\small 3h 34m 
50s} 
\\		
\rowcolor{gray!10}
\cellcolor{myyellow!20}
& $10$ & {\small $1.19\%$} & {\small $27$} & 
{\small $14$} & {\small $24$} &  {\small 12m 38s   
} 		
& $10$ & {\small $3.36\%$} & {\small $162$} & 
{\small $96$} & {\small $122$} &  {\small 41m 
9s} 
\\		
\multirow{-8}{*}{\cellcolor{myyellow!20}$\begin{matrix}\textsc{age}
	 < 
	25\\\textsc{race bias}?\end{matrix}$}
& $10$ & {\small $2.46\%$} & {\small $17$} & 
{\small $13$} & {\small $16$} &  {\small 21m 6s   }
& $10$ & {\small $0.18\%$} & {\small $17$} & 
{\small $11$} & {\small $16$} &  {\small 15m 8s} 
	\\\cline{2-13}
\cellcolor{myyellow!20}\textsc{min} && {\small 
$0.22\%$} &&&& {\small 
	11m 34s} && {\small $0.12\%$} &&&& {\small 
	7m 14s} \\	
\rowcolor{gray!10}
\cellcolor{myyellow!20}\textsc{median} && {\small 
$0.32\%$} &&&& {\small 
	36m 0s} && {\small $0.99\%$} &&&& {\small 
	20m 43s} \\	
\cellcolor{myyellow!20}\textsc{max} && {\small 
$2.46\%$} &&&& {\small 
	2h 17m 3s} && {\small $8.50\%$} &&&& {\small 
	3h 34m 50s} \\	
	\hline%

	\rowcolor{gray!10}
\cellcolor{mygreen!20}
	& $10$ & {\small $4.27\%$} & {\small $185$} 
	& 
	{\small 
	$49$} & {\small $136$} &  
	{\small  1h 46m 28s       }
& $10$ & {\small $5.20\%$} & {\small $205$} & 
{\small 
$55$} & {\small $186$} &  {\small  
4h 55m 35s    } 
\\	
\cellcolor{mygreen!20} & $10$ & {\small $7.93\%$} 
& 
{\small 
$166$} & {\small $60$} & {\small 
$127$} &  
{\small  30m 13s       }
& $10$ & {\small $12.71\%$} & {\small $404$} & 
{\small 
$120$} & {\small $359$} &  {\small  
50m 53s   } \\	
\rowcolor{gray!10}
\cellcolor{mygreen!20} 
& $10$ & {\small $8.36\%$} & {\small $197$} & 
{\small 
$39$} & {\small $135$} &  {\small  
38m 46s      }
& $10$ & {\small $7.09\%$} & {\small $168$} & 
{\small 
$59$} & {\small $133$} &  {\small  
29m 19s   } \\	
\cellcolor{mygreen!20} 
& $10$ & {\small $2.64\%$} & 
{\small $734$} 
& 
{\small 
$170 $} & {\small $322$} &  {\small  
2h 2m 22s       }
&  $\mathbf{10}$ & 
 {\small $\mathbf{6.95\%}$} 
&  {\small $\mathbf{70}$} & 
 {\small 
$\mathbf{22}$} &  {\small 
$\mathbf{69}$} &  \cellcolor{mygreen!20} {\small  
\textbf{1h 8m 47s}    } \\	
\rowcolor{gray!10}
\cellcolor{mygreen!20} 
& $10$ & {\small $4.54\%$} & {\small $706$} & 
{\small 
$269$} & {\small $497$} &  {\small  
1h 35m 54s      }
& $10$ & {\small $5.47\%$} & {\small $200$} & 
{\small 
$62$} & {\small $166$} &  {\small  
34m 15s   } \\	
\cellcolor{mygreen!20} 
&  $\mathbf{10}$ & 
 {\small $\mathbf{5.69\%}$} 
&  {\small $\mathbf{227}$} & 
 {\small 
$\mathbf{70}$} &  {\small 
$\mathbf{159}$} &  \cellcolor{mygreen!20} {\small  
\textbf{1h 21m 58s}    }
& $10$ & {\small $8.49\%$} & {\small $200$} & 
{\small 
$47$} & {\small $198$} &  {\small  
1h 21m 58s    } 
\\	
\rowcolor{gray!10}
\cellcolor{mygreen!20} 
& $10$ & {\small $7.84\%$} & {\small $276$} & 
{\small 
$61$} & {\small $249$} &  {\small  
25m 13s       }
& $10$ & {\small $6.64\%$} & {\small $332$} & 
{\small 
$139$} & {\small $262$} &  {\small  
56m 4s   } \\	
	\multirow{-8}{*}{\cellcolor{mygreen!20} 
	$\begin{matrix}\textsc{male}\\\textsc{age
		bias}?\end{matrix}$}
& $\mathbf{10}$ & 
 {\small $\mathbf{8.40\%}$} 
&  {\small $\mathbf{318}$} & 
 {\small 
$\mathbf{92}$} & {\small 
$\mathbf{317}$} &  \cellcolor{mygreen!20} {\small  
\textbf{41m 44s}       }
& $10$ & {\small $12.68\%$} & {\small $286$} & 
{\small 
$101$} & {\small $246$} &  {\small  
1h 39m 0s    }
	\\\cline{2-13}
\cellcolor{mygreen!20}\textsc{min} && {\small 
$2.64\%$} &&&& {\small 
	25m 13s} && {\small $5.20\%$} &&&& {\small 
	29m 19s} \\	
\rowcolor{gray!10}
\cellcolor{mygreen!20} \textsc{median} && {\small 
$6.77\%$} &&&& {\small 
	1h 1m 51s} && {\small $7.02\%$} &&&& {\small 
	1h 2m 26s} \\	
\cellcolor{mygreen!20} \textsc{max} && {\small 
$8.40\%$} &&&& {\small 
	2h 2m 22s} && {\small $12.71\%$} &&&& 
	{\small 
	4h 55m 35s} \\	
	\hline%
	
	\rowcolor{gray!10}
\cellcolor{mypurple!20}
	& $12$ & {\small $2.18\%$} & {\small $46$} & 
	{\small 
		$14$} & {\small $39$} &  {\small         4h 30m 
		18s     }
& $14$ & {\small $2.92\%$} & {\small $44$} & 
{\small 
	$13$} & {\small $43$} &  {\small     5h 29m 22s
} 
\\	
\cellcolor{mypurple!20}
& $12$ & {\small $3.66\%$} & {\small $68$} & 
{\small 
	$34$} & {\small $57$} &  {\small          2h 26m 
43s      
}
& $14$ & {\small $6.98\%$} & {\small $50$} & 
{\small 
	$22$} & {\small $44$} &  {\small          1h 
16m 36s   
} 
\\	
\rowcolor{gray!10}
\cellcolor{mypurple!20}
& $14$ & {\small $2.73\%$} & {\small $51$} & 
{\small 
	$22$} & {\small $43$} &  {\small        2h 
17m 
42s     }
& $\mathbf{12}$ & {\small $\mathbf{4.43\%}$} & 
{\small $\mathbf{45}$} & 
{\small 
	$\mathbf{17}$} & {\small $\mathbf{39}$} & 
	\cellcolor{mypurple!20} {\small          
	\textbf{30m 
23s}   } 
\\	
\cellcolor{mypurple!20}
& $18$ & {\small $2.19\%$} & {\small $47$} & 
{\small 
	$18$} & {\small $46$} &  {\small      35h 44m 
	27s   
}
& $9$ & {\small $3.40\%$} & {\small $46$} & 
{\small 
	$17$} & {\small $45$} &  {\small          35m 
28s   } 
\\	
\rowcolor{gray!10}
\cellcolor{mypurple!20}
& $19$ & {\small $3.17\%$} & {\small $40$} & 
{\small 
	$10$} & {\small $39$} &  {\small     60h 54m 
	6s       }
& $14$ & {\small $3.09\%$} & {\small $51$} & 
{\small 
	$18$} & {\small $48$} &  {\small        1h 40m 33s
} 
\\	
\cellcolor{mypurple!20}
& $12$ & {\small $2.45\%$} & {\small $57$} & 
{\small 
	$21$} & {\small $43$} &  {\small        3h 54m 
	37s     }
& $15$ & {\small $5.79\%$} & {\small $54$} & 
{\small 
	$34$} & {\small $53$} &  {\small         5h 
25m 11s   
} 
\\	
\rowcolor{gray!10}
\cellcolor{mypurple!20}
& $\mathbf{13}$ & {\small $\mathbf{3.94\%}$} & 
{\small $\mathbf{61}$} & 
{\small 
	$\mathbf{28}$} & {\small $\mathbf{53}$} &  
	\cellcolor{mypurple!20} {\small      \textbf{1h 
	20m 
	41s}     }
& $19$ & {\small $5.10\%$} & {\small $49$} & 
{\small 
	$33$} & {\small $48$} &  {\small         49h 
51m 42s 
} \\	
	
\multirow{-8}{*}{\cellcolor{mypurple!20}$\begin{matrix}\textsc{caucasian}\\\textsc{priors
		bias}?\end{matrix}$}
& $15$ & {\small $5.36\%$} & {\small $47$} & 
{\small 
	$18$} & {\small $46$} &  
	{\small          7h 50m 
23s      
}
& $17$ & {\small $3.99\%$} & {\small $46$} & 
{\small 
	$19$} & {\small $44$} &  {\small          13h 
5m 34s   }
	\\\cline{2-13}
\cellcolor{mypurple!20}\textsc{min} && {\small 
$2.18\%$} &&&& 
{\small 1h 20m 41s} && {\small $2.92\%$} &&&& 
{\small 
	30m 23s} \\	
\rowcolor{gray!10}
\cellcolor{mypurple!20}\textsc{median} && {\small 
$2.95\%$} &&&& 
	{\small 4h 12m 28s} && {\small $4.21\%$} &&&& 
	{\small 
	3h 32m 52s} \\	
\cellcolor{mypurple!20}\textsc{max} && {\small 
$5.36\%$} &&&& 
{\small 60h 53m 6s} && {\small $6.98\%$} &&&& 
{\small 
	49h 51m 42s} \\	
	\hline
\end{tabular}
}
	\end{table}
	\begin{table}[t]
	\caption{Queries on Neural Networks Trained
		on Fair and Race-Biased Data
		(ProPublica's \textsc{compas}
		Data) --- Full
		Table (\textsc{deeppoly} 
		Domain)}\label{tbl:compasTfull2}
\centering
\resizebox{0.75\textwidth}{!}{%
\begin{tabular}{|c|c|cc|cc|c|c|cc|cc|c|}
	\hline
	\multirow{3}{*}{$\textsc{query}$} &  	
	\multicolumn{12}{c|}{\textsc{deeppoly} }  
	\\
	& \multicolumn{6}{c|}{\textsc{fair data} } 
	& 
	\multicolumn{6}{c|}{\textsc{biased data}}
	\\
	& $\upperbound$ &
	\textsc{bias} & \size{\text{C}} & 
	\multicolumn{2}{c|}{$\size{\text{F}}$} & 
	\textsc{time} & $\upperbound$ &
	\textsc{bias} & \size{\text{C}} & 
	\multicolumn{2}{c|}{$\size{\text{F}}$} & 
	\textsc{time} \\
	\hline%

	\rowcolor{gray!10}
\cellcolor{myyellow!20} 
	&  $\mathbf{10}$ & 
	 {\small $\mathbf{0.23\%}$} & 
	 {\small $\mathbf{71}$} & 
	{\small $\mathbf{18}$} & 
 {\small $\mathbf{20}$} &  
\cellcolor{myyellow!20} {\small 
\textbf{1h 11m 
	43s}} 		
& $10$ & {\small $0.83\%$} & {\small $43$} & 
{\small $15$} & {\small $33$} &  {\small 2h 5m 
5s} 
\\

\cellcolor{myyellow!20}
&  $\mathbf{10}$ & 
 {\small $\mathbf{0.75\%}$} & 
{\small $\mathbf{33}$} & 
{\small $\mathbf{14}$} &  {\small 
$\mathbf{16}$} & \cellcolor{myyellow!20} {\small 
\textbf{10m 33s}   
}		
&  $\mathbf{10}$ & 
 {\small $\mathbf{6.48\%}$} & 
 {\small $\mathbf{63}$} & 
{\small $\mathbf{25}$} &  {\small 
$\mathbf{34}$} & \cellcolor{myyellow!20} {\small 
\textbf{8m 
46s}} 
\\
\rowcolor{gray!10}
\cellcolor{myyellow!20}
& $10$ & {\small $0.22\%$} & {\small $34$} & 
{\small $17$} & {\small $22$} &  {\small 52m 29s   
} 		
& $\mathbf{10}$ & 
 {\small $\mathbf{1.15\%}$} & 
 {\small $\mathbf{33}$} & 
 {\small $\mathbf{10}$} & 
 {\small $\mathbf{14}$} & 
\cellcolor{myyellow!20} {\small \textbf{11m 58s}} 
\\		
\cellcolor{myyellow!20}
& $\mathbf{10}$ & 
 {\small $\mathbf{0.24\%}$} & 
 {\small $\mathbf{118}$} & 
 {\small $\mathbf{28}$} & 
 {\small $\mathbf{29}$} &  
\cellcolor{myyellow!20} {\small \textbf{42m 2s}  
} 		
& $10$ & {\small $0.42\%$} & {\small $31$} & 
{\small $13$} & {\small $30$} &  {\small 10m 51s} 
\\		
\rowcolor{gray!10}
\cellcolor{myyellow!20}
& $10$ & {\small $0.31\%$} & {\small $117$} & 
{\small $49$} & {\small $54$} &  {\small 1h 0m 2s
}		
& $\mathbf{10}$ & 
 {\small $\mathbf{0.12\%}$} & 
 {\small $\mathbf{37}$} & 
 {\small $\mathbf{11}$} & 
 {\small $\mathbf{16}$} & 
\cellcolor{myyellow!20} {\small \textbf{18m 18s}} 
\\
\cellcolor{myyellow!20}
& $10$ & {\small $0.33\%$} & {\small $59$} & 
{\small $18$} & {\small $21$} &  {\small 53m 29s 
} 		
& $\mathbf{10}$ & 
 {\small $\mathbf{2.27\%}$} & 
 {\small $\mathbf{33}$} & 
{\small $\mathbf{16}$} &  {\small 
$\mathbf{24}$} & \cellcolor{myyellow!20} {\small 
\textbf{1h 
4m 
35s}} 
\\		
\rowcolor{gray!10}
\cellcolor{myyellow!20}
& $\mathbf{10}$ & 
 {\small $\mathbf{1.19\%}$} & 
 {\small $\mathbf{39}$} & 
 {\small $\mathbf{17}$} & 
 {\small $\mathbf{23}$} & 
\cellcolor{myyellow!20} {\small \textbf{9m 39s}   
} 		
& $\mathbf{10}$ & 
 {\small $\mathbf{3.41\%}$} & 
 {\small $\mathbf{133}$} & 
{\small $\mathbf{92}$} &  {\small 
$\mathbf{102}$} & \cellcolor{myyellow!20} {\small 
\textbf{33m 
43s}} 
\\		
\multirow{-8}{*}{\cellcolor{myyellow!20}$\begin{matrix}\textsc{age}
	 < 
	25\\\textsc{race bias}?\end{matrix}$}
& $\mathbf{10}$ & 
 {\small $\mathbf{2.12\%}$} & 
 {\small $\mathbf{33}$} & 
 {\small $\mathbf{17}$} & 
 {\small $\mathbf{31}$} & 
\cellcolor{myyellow!20} {\small \textbf{5m 18s}   }
& $\mathbf{10}$ & 
 {\small $\mathbf{0.18\%}$} & 
 {\small $\mathbf{33}$} & 
 {\small $\mathbf{12}$} & 
 {\small $\mathbf{17}$} & 
\cellcolor{myyellow!20} {\small \textbf{14m 58s}} 
	\\\cline{2-13}
\cellcolor{myyellow!20}\textsc{min} && {\small 
$0.22\%$} &&&& {\small 
	5m 18s} && {\small $0.12\%$} &&&& {\small 
	8m 46s} \\	
\rowcolor{gray!10}
\cellcolor{myyellow!20}\textsc{median} && {\small 
$0.32\%$} &&&& {\small 
	47m 16s} && {\small $0.99\%$} &&&& {\small 
	16m 38s} \\	
\cellcolor{myyellow!20}\textsc{max} && {\small 
$2.12\%$} &&&& {\small 
	1h 11m 43s} && {\small $6.48\%$} &&&& {\small 
	2h 5m 5s} \\	
	\hline%

	\rowcolor{gray!10}
\cellcolor{mygreen!20}
	& $10$ & {\small $3.86\%$} & {\small $242$} 
	& 
	{\small 
	$96$} & {\small $180$} &  
	{\small  2h 30m 23s    }
& $\mathbf{10}$ & {\small $\mathbf{5.22\%}$} & 
{\small $\mathbf{204}$} & 
{\small 
$\mathbf{65}$} & {\small $\mathbf{180}$} & 
\cellcolor{mygreen!20} {\small  
\textbf{3h 25m 21s}    } 
\\	
\cellcolor{mygreen!20} & $\mathbf{10}$ & {\small 
$\mathbf{8.84\%}$} 
& 
{\small 
$\mathbf{100}$} & {\small $\mathbf{45}$} & 
{\small 
$\mathbf{77}$} & \cellcolor{mygreen!20} 
{\small  \textbf{19m 47s}     }
& $\mathbf{10}$ & {\small $\mathbf{12.38\%}$} & 
{\small $\mathbf{387}$} & 
{\small 
$\mathbf{152}$} & {\small $\mathbf{318}$} & 
\cellcolor{mygreen!20} {\small  
\textbf{40m 49s}   } \\	
\rowcolor{gray!10}
\cellcolor{mygreen!20} 
& $10$ & {\small $8.14\%$} & {\small $204$} & 
{\small 
$47$} & {\small $143$} &  {\small  
28m 12s      }
& $\mathbf{10}$ & {\small $\mathbf{7.10\%}$} & 
{\small $\mathbf{181}$} & 
{\small 
$\mathbf{63}$} & {\small $\mathbf{142}$} & 
\cellcolor{mygreen!20} {\small  \textbf{
20m 51s}   } \\	
\cellcolor{mygreen!20} 
& $\mathbf{10}$ & {\small $\mathbf{2.70\%}$} & 
{\small $\mathbf{563}$} 
& 
{\small 
$\mathbf{168} $} & {\small $\mathbf{232}$} & 
\cellcolor{mygreen!20} {\small  \textbf{
1h 49m 9s}       }
& $10$ & {\small $6.90\%$} & {\small $96$} & 
{\small 
$23$} & {\small $95$} &  {\small  
1h 21m 37s    } \\	
\rowcolor{gray!10}
\cellcolor{mygreen!20} 
& $\mathbf{10}$ & {\small $\mathbf{4.65\%}$} & 
{\small $\mathbf{545}$} & 
{\small 
$\mathbf{280}$} & {\small $\mathbf{415}$} & 
\cellcolor{mygreen!20} {\small  
\textbf{1h 33m 36s}     }
& $\mathbf{10}$ & {\small $\mathbf{6.14\%}$} & 
{\small $\mathbf{157}$} & 
{\small 
$\mathbf{62}$} & {\small $\mathbf{110}$} & 
\cellcolor{mygreen!20} {\small  
\textbf{27m 43s}   } \\	
\cellcolor{mygreen!20} 
& $10$ & {\small $5.77\%$} & {\small $217$} & 
{\small 
$68$} & {\small $154$} &  {\small  
1h 35m 25s   }
& $\mathbf{10}$ & {\small $\mathbf{8.10\%}$} & 
{\small $\mathbf{345}$} & 
{\small 
$\mathbf{61}$} & {\small $\mathbf{284}$} & 
\cellcolor{mygreen!20} {\small  
\textbf{47m 9s}    } 
\\	
\rowcolor{gray!10}
\cellcolor{mygreen!20} 
& $\mathbf{10}$ & {\small $\mathbf{7.76\%}$} & 
{\small $\mathbf{252}$} & 
{\small 
$\mathbf{62}$} & {\small $\mathbf{226}$} & 
\cellcolor{mygreen!20} {\small  
\textbf{23m 10s}      }
& $\mathbf{10}$ & {\small $\mathbf{6.78\%}$} & 
{\small $\mathbf{251}$} & 
{\small 
$\mathbf{141}$} & {\small $\mathbf{223}$} &  
\cellcolor{mygreen!20}
{\small  
\textbf{50m 13s}   } \\	
	\multirow{-8}{*}{\cellcolor{mygreen!20} 
	$\begin{matrix}\textsc{male}\\\textsc{age
		bias}?\end{matrix}$}
& $10$ & {\small $8.70\%$} & {\small $267$} & 
{\small 
$90$} & {\small $266$} &  {\small  
53m 26s     }
& $\mathbf{10}$ & {\small $\mathbf{12.88\%}$} & 
{\small $\mathbf{257}$} & 
{\small 
$\mathbf{124}$} & {\small $\mathbf{228}$} & 
\cellcolor{mygreen!20} {\small  
\textbf{47m 46s}   }
	\\\cline{2-13}
\cellcolor{mygreen!20}\textsc{min} && {\small 
$2.70\%$} &&&& {\small 
	19m 47s} && {\small $5.22\%$} &&&& {\small 
	20m 51s} \\	
\rowcolor{gray!10}
\cellcolor{mygreen!20} \textsc{median} && {\small 
$6.77\%$} &&&& {\small 
	1h 13m 31s} && {\small $7.00\%$} &&&& {\small 
	47m 28s} \\	
\cellcolor{mygreen!20} \textsc{max} && {\small 
$8.84\%$} &&&& {\small 
	2h 20m 23s} && {\small $12.88\%$} &&&& 
	{\small 
	3h 25m 21s} \\	
	\hline%
	
	\rowcolor{gray!10}
\cellcolor{mypurple!20}
	& $\mathbf{11}$ & {\small $\mathbf{2.18\%}$} & 
	{\small $\mathbf{106}$} & 
	{\small 
		$\mathbf{21}$} & {\small $\mathbf{53}$} & 
		\cellcolor{mypurple!20} {\small \textbf{2h 
		32m 
		44s}    
		}
& $\mathbf{11}$ & {\small $\mathbf{2.92\%}$} & 
{\small $\mathbf{86}$} & 
{\small 
	$\mathbf{26}$} & {\small $\mathbf{69}$} & 
	\cellcolor{mypurple!20} {\small  \textbf{2h 26m 
	20s}  
} 
\\	
\cellcolor{mypurple!20}
& $\mathbf{7}$ & {\small $\mathbf{3.66\%}$} & 
{\small $\mathbf{105}$} & 
{\small 
	$\mathbf{38}$} & {\small $\mathbf{55}$} & 
	\cellcolor{mypurple!20} {\small  \textbf{18m 26s}
}
& $\mathbf{11}$ & {\small $\mathbf{6.95\%}$} & 
{\small $\mathbf{108}$} & 
{\small 
	$\mathbf{33}$} & {\small $\mathbf{71}$} & 
	\cellcolor{mypurple!20} {\small      \textbf{15m 
	29s}
} 
\\	
\rowcolor{gray!10}
\cellcolor{mypurple!20}
& $\mathbf{11}$ & {\small $\mathbf{2.73\%}$} & 
{\small $\mathbf{100}$} & 
{\small 
	$\mathbf{32}$} & {\small $\mathbf{57}$} & 
	\cellcolor{mypurple!20} {\small     \textbf{39m 
	5s}  
	}
& $14$ & {\small $4.43\%$} & {\small $69$} & 
{\small 
	$12$} & {\small $51$} &  {\small    1h 47m 5s  } 
\\	
\cellcolor{mypurple!20}
& $\mathbf{17}$ & {\small $\mathbf{2.19\%}$} & 
{\small $\mathbf{101}$} & 
{\small 
	$\mathbf{28}$} & {\small $\mathbf{57}$} &  
	\cellcolor{mypurple!20} {\small    \textbf{16h 
	19m 14s}
}
& $\mathbf{7}$ & {\small $\mathbf{3.40\%}$} & 
{\small $\mathbf{83}$} & 
{\small 
	$\mathbf{21}$} & {\small $\mathbf{82}$} &  
	\cellcolor{mypurple!20} {\small    \textbf{20m 
	1s}  } 
\\	
\rowcolor{gray!10}
\cellcolor{mypurple!20}
& $19$ & {\small $3.17\%$} & {\small $86$} & 
{\small 
	$30$} & {\small $53$} &  {\small      
52h 10m 2s     }
& $\mathbf{13}$ & {\small $\mathbf{3.09\%}$} & 
{\small $\mathbf{96}$} & 
{\small 
	$\mathbf{24}$} & {\small $\mathbf{58}$} &  
	\cellcolor{mypurple!20} {\small     \textbf{1h 8m 
	4s} 
} 
\\	
\cellcolor{mypurple!20}
& $\mathbf{11}$ & {\small $\mathbf{2.45\%}$} & 
{\small $\mathbf{94}$} & 
{\small 
	$\mathbf{26}$} & {\small $\mathbf{52}$} &  
	\cellcolor{mypurple!20} {\small     \textbf{2h 
	18m 
	42s}    }
& $\mathbf{14}$ & {\small $\mathbf{5.79\%}$} & 
{\small $\mathbf{99}$} & 
{\small 
	$\mathbf{45}$} & {\small $\mathbf{87}$} &  
	\cellcolor{mypurple!20} {\small    \textbf{1h 51m 
	2s}
} 
\\	
\rowcolor{gray!10}
\cellcolor{mypurple!20}
& $15$ & {\small $3.94\%$} & {\small $87$} & 
{\small 
	$29$} & {\small $52$} &  {\small        2h 39m 
	18s   
	}
& $\mathbf{17}$ & {\small $\mathbf{5.10\%}$} & 
{\small $\mathbf{110}$} & 
{\small 
	$\mathbf{73}$} & {\small $\mathbf{94}$} &  
	\cellcolor{mypurple!20} {\small     \textbf{17h 
	48m 22s}
} \\	
	
\multirow{-8}{*}{\cellcolor{mypurple!20}$\begin{matrix}\textsc{caucasian}\\\textsc{priors
		bias}?\end{matrix}$}
& $\mathbf{15}$ & {\small $\mathbf{5.36\%}$} & 
{\small $\mathbf{90}$} & 
{\small 
	$\mathbf{35}$} & {\small $\mathbf{89}$} & 
	\cellcolor{mypurple!20} 
	{\small     \textbf{3h 41m 16s}   
}
& $\mathbf{14}$ & {\small $\mathbf{3.99\%}$} & 
{\small $\mathbf{97}$} & 
{\small 
	$\mathbf{38}$} & {\small $\mathbf{65}$} & 
	\cellcolor{mypurple!20} {\small   \textbf{1h 21m 
	8s}   }
	\\\cline{2-13}
\cellcolor{mypurple!20}\textsc{min} && {\small 
$2.18\%$} &&&& 
{\small 18m 26s} && {\small $2.92\%$} &&&& 
{\small 
	15m 29s} \\	
\rowcolor{gray!10}
\cellcolor{mypurple!20}\textsc{median} && {\small 
$2.95\%$} &&&& 
	{\small 2h 36m 1s} && {\small $4.21\%$} 
	&&&& 
	{\small 
	1h 34m 7s} \\	
\cellcolor{mypurple!20}\textsc{max} && {\small 
$5.36\%$} &&&& 
{\small 52h 10m 2s} && {\small $6.95\%$} &&&& 
{\small 
	17h 48m 22s} \\	
	\hline
\end{tabular}
}
\end{table}

\subsection{RQ2: Answering Bias 
Queries}

Table~\ref{tbl:compasTfull0},~\ref{tbl:compasTfull1}
  and~\ref{tbl:compasTfull2} show the 
analysis 
results for all eight models trained on the 
\textsc{compas} dataset from ProPublica. 
All columns are shown as before and, 
again, we 
highlighted across all tables the choice 
of the abstract domain that 
entailed the shortest 
analysis time.
\end{toappendix}

For each line in 
Table~\ref{tbl:compasT}, we highlighted the choice 
of abstract domain that entailed the shortest 
analysis time.
We observe that \textsc{deeppoly} seems generally
the better choice. The difference in performance 
becomes more striking as the analyzed input space 
becomes smaller, i.e., for $Q_C$. This is because 
\textsc{deeppoly} is specifically designed for 
proving \emph{local} robustness of neural 
networks. Thus, our input partitioning, in addition 
to allowing for
parallelism, is also 
enabling analyses designed for local properties to 
prove global properties, like causal fairness.

The analysis results for all models are shown in the 
appendix (see 
Tables~\ref{tbl:compasTfull0},~\ref{tbl:compasTfull1},
 and~\ref{tbl:compasTfull2}).

\paragraph{\textbf{RQ3: Effect of Model 
Structure on Scalability.}}
To evaluate the effect of the 
model structure on the
scalability of our analysis, we trained 
models
on the Adult Census
dataset\footnote{\url{https://archive.ics.uci.edu/ml/datasets/adult}}
 by
varying the number of layers and nodes per layer. 
The dataset
assigns a yearly income ($>$ or $\leq$ USD 50K) based on personal attributes
such as gender, race, and occupation. We trained all 
models (with 23 inputs) on a fair
dataset with respect to gender and ensured that each model reached a
minimum classification accuracy of 78\%. Accuracy 
does not
increase by adding more layers or nodes per layer, 
in fact, it may
significantly decrease --- we tried up to 100 hidden layers with 100
nodes each.

\begin{table}[t]
	\caption{Comparison of
		Different Model Structures (Adult
		Census Data)}\label{tbl:censusT}
\centering
\resizebox{\textwidth}{!}{%
	\begin{tabular}{|c|c|cc|cc|c|cc|cc|c|cc|cc|c|}
		\hline
		\multirow{2}{*}{$\size{\model}$} & 
		\multirow{2}{*}{\textsc{U}} &     
		\multicolumn{5}{c|}{\textsc{boxes}} &     
		\multicolumn{5}{c|}{\textsc{symbolic}} 
		&     \multicolumn{5}{c|}{\textsc{deeppoly}} \\
		& & \textsc{input} & $\size{\text{C}}$ & 		
		\multicolumn{2}{c|}{$\size{\text{F}}$} & 		
		\textsc{time} &         \textsc{input} & 
		$\size{\text{C}}$ &         		
		\multicolumn{2}{c|}{$\size{\text{F}}$} 
		&         		\textsc{time} & 		\textsc{input} 
		& $\size{\text{C}}$ & 		
		\multicolumn{2}{c|}{$\size{\text{F}}$} & 		
		\textsc{time} \\
		\hline
		\rowcolor{gray!10} \cellcolor{mygreen!20} & 
		$4$ & {\small $88.26\%$} & {\small $1482$} & 
		{\small $77$} & {\small $1136$} & {\small 
		33m 55s} & {\small $95.14\%$} & {\small 
		$1132$} & {\small $65$} & {\small $686$} & 
		{\small 19m 5s} & {\small $93.99\%$} & 
		{\small $1894$} & {\small $77$} & {\small 
		$992$} & {\small 29m 55s} \\
		\cellcolor{mygreen!20} & $6$ & {\small 
		$99.51\%$} & {\small $769$} & {\small $51$} 
		& {\small $723$} & {\small 1h 10m 25s} & 
		{\small $99.93\%$} & {\small $578$} & {\small 
		$47$} & {\small $447$} & {\small 39m 8s} & 
		{\small $99.83\%$} & {\small $1620$} & 
		{\small $54$} & {\small $1042$} & {\small 1h 
		24m 24s} \\
		\rowcolor{gray!10} \cellcolor{mygreen!20} & 
		$8$ & {\small $100.00\%$} & {\small $152$} & 
		{\small $19$} & {\small $143$} & {\small 3h 
		47m 23s} & {\small $100.00\%$} & {\small 
		$174$} & {\small $18$} & {\small $146$} & 
		{\small 1h 51m 2s} & {\small $100.00\%$} & 
		{\small $1170$} & {\small $26$} & {\small 
		$824$} & {\small 8h 2m 27s} \\
		\multirow{-4}{*}{\cellcolor{mygreen!20}\begin{tabular}{c}$10$\\$\showmark{o}\showmark{*}\showmark{oplus}$\end{tabular}}
		 &$10$ & \cellcolor{mygreen!20} {\small 
		 $\mathbf{100.00\%}$} & 
		 \cellcolor{mygreen!20} {\small 
		 $\mathbf{1}$} 
		& \cellcolor{mygreen!20} {\small 
		$\mathbf{1}$} & 
		\cellcolor{mygreen!20} {\small $\mathbf{1}$} 
		& 
		\cellcolor{mygreen!20} {\small \textbf{55m 
		58s}} & {\small $100.00\%$} & {\small $1$} & 
		{\small $1$} & {\small $1$} & {\small 56m 8s} 
		& {\small $100.00\%$} & {\small $1$} & 
		{\small $1$} & {\small $1$} & {\small 56m 
		43s} \\
		\hline
		\rowcolor{gray!10} \cellcolor{myyellow!20} & 
		$4$ & {\small $49.83\%$} & {\small $719$} & 
		{\small $9$} & {\small $329$} & {\small 13m 
		43s} & {\small $72.29\%$} & {\small $1177$} 
		& {\small $11$} & {\small $559$} & {\small 
		24m 9s} & {\small $60.52\%$} & {\small 
		$1498$} & {\small $14$} & {\small $423$} & 
		{\small 10m 32s} \\
		\cellcolor{myyellow!20} & $6$ & {\small 
		$72.74\%$} & {\small $1197$} & {\small $15$} 
		& {\small $929$} & {\small 2h 6m 49s} & 
		{\small $98.54\%$} & {\small $333$} & {\small 
		$7$} & {\small $195$} & {\small 20m 46s} & 
		{\small $66.46\%$} & {\small $1653$} & 
		{\small $17$} & {\small $594$} & {\small 15m 
		44s} \\
		\rowcolor{gray!10} \cellcolor{myyellow!20} & 
		$8$ & {\small $98.68\%$} & {\small $342$} & 
		{\small $9$} & {\small $284$} & {\small 1h 
		46m 43s} & {\small $98.78\%$} & {\small 
		$323$} & {\small $9$} & {\small $190$} & 
		{\small 1h 27m 18s} & {\small $70.87\%$} & 
		{\small $1764$} & {\small $18$} & {\small 
		$724$} & {\small 2h 19m 11s} \\
		\multirow{-4}{*}{\cellcolor{myyellow!20}\begin{tabular}{c}$12$\\$\showmark{triangle}\showmark{triangle*}\showmark{Mercedes
		 star}$\end{tabular}} &$10$ & {\small 
		$99.06\%$} & {\small $313$} & {\small $7$} 
		& {\small $260$} & {\small 1h 21m 47s} & 
		\cellcolor{myyellow!20} {\small 
		$\mathbf{99.06\%}$} & 
		\cellcolor{myyellow!20} {\small 
		$\mathbf{307}$} & 
		\cellcolor{myyellow!20} {\small 
		$\mathbf{5}$} & \cellcolor{myyellow!20} 
		{\small 
		$\mathbf{182}$} & \cellcolor{myyellow!20} 
		{\small \textbf{1h 
		13m 55s}} 
		& {\small $80.76\%$} & {\small $1639$} & 
		{\small $18$} & {\small $1007$} & {\small 3h 
		22m 11s} \\
		\hline
		\rowcolor{gray!10} \cellcolor{myorange!20} & 
		$4$ & {\small $38.92\%$} & {\small $1044$} & 
		{\small $18$} & {\small $39$} & {\small 2m 
		6s} & {\small $51.01\%$} & {\small $933$} & 
		{\small $31$} & {\small $92$} & {\small 15m 
		28s} & {\small $49.62\%$} & {\small $1081$} 
		& {\small $34$} & {\small $79$} & {\small 3m 
		2s} \\
		\cellcolor{myorange!20} & $6$ & {\small 
		$46.22\%$} & {\small $1123$} & {\small 
		$62$} & {\small $255$} & {\small 20m 51s} & 
		{\small $61.60\%$} & {\small $916$} & {\small 
		$67$} & {\small $405$} & {\small 44m 40s} & 
		{\small $59.20\%$} & {\small $1335$} & 
		{\small $90$} & {\small $356$} & {\small 
		22m 13s} \\
		\rowcolor{gray!10} \cellcolor{myorange!20} & 
		$8$ & {\small $64.24\%$} & {\small $1111$} & 
		{\small $96$} & {\small $792$} & {\small 2h 
		24m 51s} & {\small $74.27\%$} & {\small 
		$1125$} & {\small $78$} & {\small $780$} & 
		{\small 3h 26m 20s} & {\small $69.69\%$} & 
		{\small $1574$} & {\small $127$} & {\small 
		$652$} & {\small 5h 6m 7s} \\
		\multirow{-4}{*}{\cellcolor{myorange!20}\begin{tabular}{c}$20$\\$\showmark{diamond}\showmark{diamond*}\testmark{myorange!20}{halfdiamond*}$\end{tabular}}
		 &$10$ & {\small $85.90\%$} & {\small 
		$1390$} & {\small $71$} & {\small $1339$} & 
		{\small >13h} & {\small $89.27\%$} & {\small 
		$1435$} & {\small $60$} & {\small $1157$} & 
		{\small >13h} & \cellcolor{myorange!20} 
		{\small $\mathbf{76.25\%}$} & 
		\cellcolor{myorange!20} 
		{\small 
		$\mathbf{1711}$} & \cellcolor{myorange!20} 
		{\small 
		$\mathbf{148}$} & \cellcolor{myorange!20} 
		{\small 
		$\mathbf{839}$} & 
		\cellcolor{myorange!20} {\small \textbf{4h 
		36m 23s}} \\
		\hline
		\rowcolor{gray!10} \cellcolor{myred!20} & $4$ 
		& {\small $0.35\%$} & {\small $10$} & {\small 
		$0$} & {\small $0$} & {\small 1m 39s} & 
		{\small $34.62\%$} & {\small $768$} & {\small 
		$1$} & {\small $1$} & {\small 6m 56s} & 
		{\small $26.39\%$} & {\small $648$} & {\small 
		$2$} & {\small $3$} & {\small 10m 11s} \\
		\cellcolor{myred!20} & $6$ & {\small 
		$0.35\%$} & {\small $10$} & {\small $0$} & 
		{\small $0$} & {\small 1m 38s} & {\small 
		$34.76\%$} & {\small $817$} & {\small $4$} & 
		{\small $5$} & {\small 43m 53s} & {\small 
		$26.74\%$} & {\small $592$} & {\small $8$} 
		& {\small $10$} & {\small 1h 23m 11s} \\
		\rowcolor{gray!10} \cellcolor{myred!20} & $8$ 
		& {\small $0.42\%$} & {\small $12$} & {\small 
		$1$} & {\small $2$} & {\small 14m 37s} & 
		{\small $35.56\%$} & {\small $840$} & {\small 
		$21$} & {\small $28$} & {\small 2h 48m 15s} 
		& {\small $27.74\%$} & {\small $686$} & 
		{\small $32$} & {\small $42$} & {\small 2h 
		43m 2s} \\
		\multirow{-4}{*}{\cellcolor{myred!20}\begin{tabular}{c}$40$\\$\showmark{square}\showmark{square*}\testmark{myred!20}{halfsquare*}$\end{tabular}}
		 &$10$ & {\small $0.80\%$} & {\small $23$} & 
		{\small $10$} & {\small $13$} & {\small 1h 
		48m 43s} & \cellcolor{myred!20} {\small 
		$\mathbf{37.19\%}$} & \cellcolor{myred!20} 
		{\small 
		$\mathbf{880}$} & \cellcolor{myred!20} 
		{\small $\mathbf{50}$} & 
		\cellcolor{myred!20} {\small $\mathbf{75}$} & 
		\cellcolor{myred!20}
		{\small \textbf{11h 32m 21s}} & {\small 
		$30.56\%$} & 
		{\small $699$} & {\small $83$} & {\small 
		$121$} & {\small >13h} \\
		\hline
		\rowcolor{gray!10} \cellcolor{mypurple!20} & 
		$4$ & {\small $1.74\%$} & {\small $50$} & 
		{\small $0$} & {\small $0$} & {\small 1m 38s} 
		& {\small $41.98\%$} & {\small $891$} & 
		{\small $14$} & {\small $49$} & {\small 10m 
		14s} & {\small $36.60\%$} & {\small $805$} & 
		{\small $6$} & {\small $8$} & {\small 2m 47s} 
		\\
		\cellcolor{mypurple!20} & $6$ & {\small 
		$2.50\%$} & {\small $72$} & {\small $3$} & 
		{\small $22$} & {\small 4m 35s} & {\small 
		$45.00\%$} & {\small $822$} & {\small $32$} 
		& {\small $143$} & {\small 45m 42s} & {\small 
		$38.06\%$} & {\small $847$} & {\small $25$} 
		& {\small $50$} & {\small 5m 7s} \\
		\rowcolor{gray!10} \cellcolor{mypurple!20} & 
		$8$ & {\small $9.83\%$} & {\small $282$} & 
		{\small $25$} & {\small $234$} & {\small 25m 
		30s} & {\small $47.78\%$} & {\small $651$} & 
		{\small $46$} & {\small $229$} & {\small 1h 
		14m 5s} & {\small $42.53\%$} & {\small 
		$975$} & {\small $74$} & {\small $180$} & 
		{\small 25m 1s} \\
		\multirow{-4}{*}{\cellcolor{mypurple!20}\begin{tabular}{c}$45$\\$\showmark{pentagon}\showmark{pentagon*}\showmark{10-pointed
		 star}$\end{tabular}} &$10$ & {\small 
		$18.68\%$} & {\small $522$} & {\small $33$} 
		& {\small $488$} & {\small 1h 51m 24s} & 
		\cellcolor{mypurple!20} {\small 
		$\mathbf{49.62\%}$} & 
		\cellcolor{mypurple!20} {\small 
		$\mathbf{714}$} & 
		\cellcolor{mypurple!20} {\small 
		$\mathbf{51}$} & \cellcolor{mypurple!20} 
		{\small 
		$\mathbf{294}$} & \cellcolor{mypurple!20} 
		{\small \textbf{3h 
		23m 
		20s}} & {\small $48.68\%$} & {\small $1087$} 
		& {\small $110$} & {\small $373$} & {\small 
		1h 58m 34s} \\
		\hline
	\end{tabular}
}
\end{table}

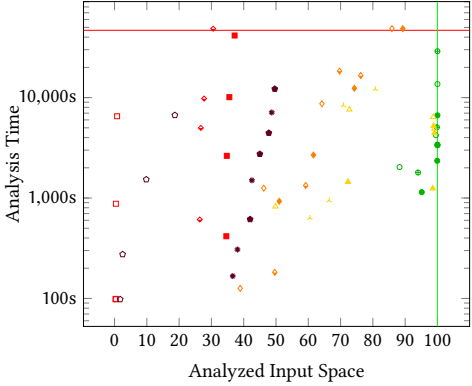
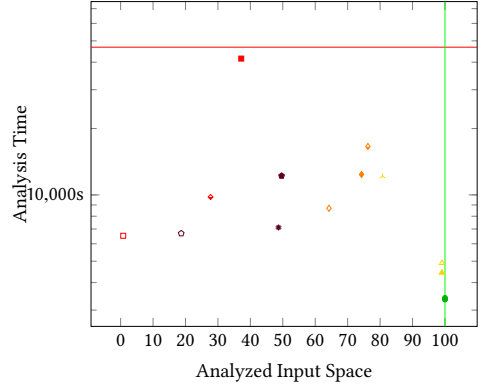
\begin{figure}[t]
	\begin{subfigure}[b]{.5\linewidth}
\begin{tikzpicture}[scale=0.75]
\begin{semilogyaxis}[xlabel={Analyzed Input 
Space}, 
xtick={0,10,20,30,40,50,60,70,80,90,100}, 
ylabel={Analysis Time}, 
log number format basis/.code 2 
args={$\pgfmathparse{#1^(#2)}\mathsf{\pgfmathprintnumber{\pgfmathresult}}$s},
]
\addplot[scatter,only marks,  
point 
meta=explicit 
symbolic,scatter/classes={
nosym45={mark=pentagon,purple!50!black,mark 
	size=1.5pt},
nosym40={mark=square,red,mark 
	size=1.25pt},
nosym20={mark=diamond,orange,mark 
	size=1.5pt},
nosym12={mark=triangle,yellow!70!orange,mark 
	size=1.5pt},
nosym10={mark=o,green!70!black,mark 
	size=1.25pt},
sym45={mark=pentagon*,purple!50!black,mark 
	size=1.5pt},
sym40={mark=square*,red,mark 
	size=1.25pt},
sym20={mark=diamond*,orange,mark 
	size=1.5pt},
sym12={mark=triangle*,yellow!70!orange,mark 
	size=1.5pt},
sym10={mark=*,green!70!black,mark 
	size=1.25pt},
deep45={mark=10-pointed 
star,purple!50!black,mark 
	size=1.5pt},
deep40={mark=halfsquare*,red,mark 
	size=1.25pt},
deep20={mark=halfdiamond*,orange,mark 
	size=1.5pt},
deep12={mark=Mercedes 
star,yellow!70!orange,mark 
	size=1.5pt},
deep10={mark=oplus,green!70!black,mark 
	size=1.25pt}
%
}]
table[meta=label] {
x			y			label

%
%

88.26	2035	nosym10
99.51	4225	nosym10
100.00	13643	nosym10
100.00	3358	nosym10

95.14	1145	sym10
99.93	2348	sym10
100.00	6662	sym10
100.00	3368	sym10

93.99	1795	deep10
99.83	5064	deep10
100.00	28947	deep10
100.00	3403	deep10

%
%

49.83	823		nosym12
72.74	7609	nosym12
98.68	6403	nosym12
99.06	4907	nosym12

72.29	1449	sym12
98.54	1246	sym12
98.79	5238	sym12
99.06	4435	sym12

60.52	632		deep12
66.46	944		deep12
70.87	8351	deep12
80.76	12131	deep12

%
%

38.92	126		nosym20
46.22	1251	nosym20
64.24	8691	nosym20
85.90	 48600	nosym20

51.01	928		sym20
61.60	2680	sym20
74.27	12380	sym20
89.27	 48600	sym20

49.62	182		deep20
59.20	1333	deep20
69.69	18367	deep20
76.25	16583	deep20

%
%

0.35	99		nosym40
0.35	98		nosym40
0.42	877		nosym40
0.80	6523	nosym40

34.62	416		sym40
34.76	2633	sym40
35.56	10095	sym40
37.19	41541	sym40

26.39	611		deep40
26.74	4991	deep40
27.74	9782	deep40
30.56	 48600	deep40

%
%

1.74	98		nosym45
2.50	275		nosym45
9.83	1530	nosym45
18.68	6684	nosym45

41.98	614		sym45
45.00	2742	sym45
47.78	4445	sym45
49.62	12200	sym45

36.60	167		deep45
38.06	307		deep45
42.54	1501	deep45
48.68	7114	deep45
};
\draw[red] ({rel axis cs:0,0}|-{axis cs:0,46800}) -- 
({rel 
axis  
cs:1,0}|-{axis cs:1,46800});
\draw[green] ({rel axis cs:0,0}-|{axis cs:100,0}) -- ({rel 
axis cs:0,1}-|{axis cs:100,0});
\end{semilogyaxis}
\end{tikzpicture}
		\caption{}
		\label{fig:censusF1}
	\end{subfigure}%
	\begin{subfigure}[b]{.5\linewidth}
\begin{tikzpicture}[scale=0.75]
\begin{semilogyaxis}[xlabel={Analyzed Input 
Space}, 
xtick={0,10,20,30,40,50,60,70,80,90,100}, 
ylabel={Analysis Time}, 
ytick={1000,10000}, 
log number format basis/.code 2 
args={$\pgfmathparse{#1^(#2)}\mathsf{\pgfmathprintnumber{\pgfmathresult}}$s},
]
\addplot[scatter,only marks,  
point 
meta=explicit 
symbolic,scatter/classes={
	nomark={mark=none,mark size=1.25pt},
nosym45={mark=pentagon,purple!50!black,mark 
	size=1.5pt},
nosym40={mark=square,red,mark 
	size=1.25pt},
nosym20={mark=diamond,orange,mark 
	size=1.5pt},
nosym12={mark=triangle,yellow!70!orange,mark 
	size=1.5pt},
nosym10={mark=o,green!70!black,mark 
	size=1.25pt},
sym45={mark=pentagon*,purple!50!black,mark 
	size=1.5pt},
sym40={mark=square*,red,mark 
	size=1.25pt},
sym20={mark=diamond*,orange,mark 
	size=1.5pt},
sym12={mark=triangle*,yellow!70!orange,mark 
	size=1.5pt},
sym10={mark=*,green!70!black,mark 
	size=1.25pt},
deep45={mark=10-pointed 
star,purple!50!black,mark 
	size=1.5pt},
deep40={mark=halfsquare*,red,mark 
	size=1.25pt},
deep20={mark=halfdiamond*,orange,mark 
	size=1.5pt},
deep12={mark=Mercedes 
star,yellow!70!orange,mark 
	size=1.5pt},
deep10={mark=oplus,green!70!black,mark 
	size=1.25pt}
%
}]
table[meta=label] {
x			y			label

%
%

100.00	3358	nosym10

100.00	3368	sym10

100.00	3403	deep10

%
%

99.06	4907	nosym12

99.06	4435	sym12

80.76	12131	deep12

%
%

64.24	8691	nosym20

74.27	12380	sym20

76.25	16583	deep20

%
%

0.80	6523	nosym40

37.19	41541	sym40

27.74	9782	deep40

%
%

18.68	6684	nosym45

49.62	12200	sym45

48.68	7114	deep45

50.00 56800 nomark
};
\draw[red] ({rel axis cs:0,0}|-{axis cs:0,46800}) -- 
({rel 
axis  
cs:1,0}|-{axis cs:1,46800});

\draw[green] ({rel axis cs:0,0}-|{axis cs:100,0}) -- 
({rel 
axis cs:0,1}-|{axis cs:100,0});
\end{semilogyaxis}
\end{tikzpicture}
		\caption{Zoom on Best 
		$\upperbound$-Configurations}
		\label{fig:censusF2}
	\end{subfigure}
	\caption{Comparison of
		Different Model Structures (Adult
		Census Data)}\label{fig:censusF}
	\vspace{-1em}
\end{figure}

Table~\ref{tbl:censusT} shows the results. The first 
column ($\size{\model}$) shows
the total number of hidden nodes and introduces the
marker symbols used in the scatter plot of 
Figure~\ref{fig:censusF} (to identify the domain 
used for the forward pre-analysis: left, center, and 
right symbols respectively refer to the 
\textsc{boxes}, \textsc{symbolic}, and 
\textsc{deeppoly} domains).
The models have the following number of hidden 
layers
and nodes per layer (from top to bottom): 2 and 5; 4 
and 3; 4 and 5;
4 and 10; 9 and 5.

Column $\upperbound$ shows the chosen 
upper 
bound for the analysis. For each model, we tried 
four different choices of $\upperbound$.
Column \textsc{input} shows the input-space 
coverage, i.e., the
percentage of the input space that was completed 
by the analysis. 
%
	Column $\size{\text{C}}$ shows the 
	total number of analyzed (i.e., completed) input 
	space partitions. Column $\size{\text{F}}$ shows 
	the total number of abstract activation patterns 
	(left) and feasible input partitions (right) that the 
	backward analysis had to explore. The difference 
	between $\size{\text{C}}$ and the number of 
	partitions shown in $\size{\text{F}}$ are the input 
	partitions that the pre-analysis found to be 
	already 
	fair (i.e., uniquely classified). 
	Finally, column \textsc{time} shows the analysis 
	running time.
We used
a lower bound $\lowerbound$ of $0.5$ and a 
time limit of $13$h. 
For each model in Table~\ref{tbl:censusT}, we 
highlighted the configuration (i.e., domain used for 
the pre-analysis and chosen $\upperbound$)
that achieved the highest input-space coverage 
(the analysis running time being decisive in 
case of equality or timeout).

The scatter plot of Figure~\ref{fig:censusF1} 
visualizes the input coverage
and analysis running time. We zoom in on the best 
$\upperbound$-configurations for each 
pre-analysis domain (i.e., the chosen 
$\upperbound$) in 
Figure~\ref{fig:censusF2}.

Overall, we observe that 
\emph{coverage decreases
  for larger model structures}, and the more 
  precise \textsc{symbolic} and \textsc{deeppoly} 
  domains result in a 
  significant
  coverage boost, especially for larger 
  structures.
We also note that, as in this case we are analyzing 
the entire 
input space, \textsc{deeppoly} generally performs 
worse than
the \textsc{symbolic} domain. In 
particular, for larger structures, \emph{the 
\textsc{symbolic} domain 
often yields a higher input coverage in a shorter 
analysis 
running time}.
Finally, we observe that \emph{increasing the upper 
bound $\upperbound$ tends to increase
  coverage independently of the specific model structure.} However,
interestingly, this does not always come at the 
expense of an
  increased running time. In fact, such a change 
  often results in decreasing the number of 
  partitions that the expensive backward analysis 
  needs to analyze (cf. columns 
  $\size{\text{F}}$) and, 
  in turn, this reduces the 
  overall running time.

\begin{table}[t]
	\caption{Comparison of Different Input Space 
		Sizes and Model Structures (Adult
		Census Data)}\label{tbl:censusNEW}
\centering
\resizebox{\textwidth}{!}{%
	\begin{tabular}{|c|c|cc|cc|c|cc|cc|c|cc|cc|c|}
		\hline
		\multirow{2}{*}{$\size{\model}$} & 
		\multirow{2}{*}{\textsc{query}} &     
		\multicolumn{5}{c|}{\textsc{boxes}} &     
		\multicolumn{5}{c|}{\textsc{symbolic}} 
		&     \multicolumn{5}{c|}{\textsc{deeppoly}} \\
		& & \textsc{input} & $\size{\text{C}}$ & 		
		\multicolumn{2}{c|}{$\size{\text{F}}$} & 		
		\textsc{time} &         \textsc{input} & 
		$\size{\text{C}}$ &         		
		\multicolumn{2}{c|}{$\size{\text{F}}$} 
		&         		\textsc{time} & 		\textsc{input} 
		& $\size{\text{C}}$ & 		
		\multicolumn{2}{c|}{$\size{\text{F}}$} & 		
		\textsc{time} \\
		\hline
		\rowcolor{gray!10} \cellcolor{myred!20} & 
		\textsc{F} & {\small $100.000\%$} &  &  &  &  
		& {\small $100.000\%$} &  &  &  &  & 
		\cellcolor{myred!20} {\small 
		$\mathbf{100.000\%}$} & \cellcolor{myred!20} 
		& 
		\cellcolor{myred!20} & \cellcolor{myred!20} & 
		\cellcolor{myred!20} \\
		\rowcolor{gray!10} \cellcolor{myred!20} & 
		\textcolor{gray}{\scriptsize $0.009\%$} & 
		\textcolor{gray}{\scriptsize $0.009\%$} & 
		\multirow{-2}{*}{\small $9$} & 
		\multirow{-2}{*}{\small $2$} & 
		\multirow{-2}{*}{\small $3$} & 
		\multirow{-2}{*}{\small 3m 3s} & 
		\textcolor{gray}{\scriptsize $0.009\%$} & 
		\multirow{-2}{*}{\small $5$} & 
		\multirow{-2}{*}{\small $1$} & 
		\multirow{-2}{*}{\small $2$} & 
		\multirow{-2}{*}{\small 3m 5s} & 
		\cellcolor{myred!20}
		\textcolor{gray}{\scriptsize 
		$\mathbf{0.009\%}$} & 
		\cellcolor{myred!20}
		\multirow{-2}{*}{\small $\mathbf{3}$} & 
		\cellcolor{myred!20}
		\multirow{-2}{*}{\small $\mathbf{1}$} & 
		\cellcolor{myred!20}
		\multirow{-2}{*}{\small $\mathbf{1}$} & 
		\cellcolor{myred!20}
		\multirow{-2}{*}{\small \textbf{2m 33s}} \\
		\cellcolor{myred!20} & \textsc{E} & {\small 
		$99.996\%$} &  &  &  &  & {\small 
		$100.000\%$} &  &  &  &  & 
		\cellcolor{myred!20} {\small 
		$\mathbf{100.000\%}$} & \cellcolor{myred!20} 
		& 
		\cellcolor{myred!20} & \cellcolor{myred!20} & 
		\cellcolor{myred!20} \\
		\cellcolor{myred!20} & 
		\textcolor{gray}{\scriptsize $0.104\%$} & 
		\textcolor{gray}{\scriptsize $0.104\%$} & 
		\multirow{-2}{*}{\small $83$} & 
		\multirow{-2}{*}{\small $9$} & 
		\multirow{-2}{*}{\small $39$} & 
		\multirow{-2}{*}{\small 3m 13s} & 
		\textcolor{gray}{\scriptsize $0.104\%$} & 
		\multirow{-2}{*}{\small $26$} & 
		\multirow{-2}{*}{\small $3$} & 
		\multirow{-2}{*}{\small $9$} & 
		\multirow{-2}{*}{\small 3m 8s} & 
		\cellcolor{myred!20}
		\textcolor{gray}{\scriptsize 
		$\mathbf{0.104\%}$} & 
		\cellcolor{myred!20}
		\multirow{-2}{*}{\small $\mathbf{22}$} & 
		\cellcolor{myred!20}
		\multirow{-2}{*}{\small $\mathbf{3}$} & 
		\cellcolor{myred!20}
		\multirow{-2}{*}{\small $\mathbf{9}$} & 
		\cellcolor{myred!20}
		\multirow{-2}{*}{\small \textbf{2m 38s}} \\
		\rowcolor{gray!10} \cellcolor{myred!20} & 
		\textsc{D} & {\small $99.978\%$} &  &  &  &  & 
		\cellcolor{myred!20} {\small 
		$\mathbf{100.000\%}$} & 
		\cellcolor{myred!20} & \cellcolor{myred!20} & 
		\cellcolor{myred!20} & \cellcolor{myred!20} & 
		{\small 
		$100.000\%$} &  &  &  &  \\
		\rowcolor{gray!10} \cellcolor{myred!20} & 
		\textcolor{gray}{\scriptsize $1.042\%$} & 
		\textcolor{gray}{\scriptsize $1.042\%$} & 
		\multirow{-2}{*}{\small $457$} & 
		\multirow{-2}{*}{\small $13$} & 
		\multirow{-2}{*}{\small $176$} & 
		\multirow{-2}{*}{\small 5m } & 
		\cellcolor{myred!20}
		\textcolor{gray}{\scriptsize 
		$\mathbf{1.042\%}$} & 
		\cellcolor{myred!20}
		\multirow{-2}{*}{\small $\mathbf{292}$} & 
		\cellcolor{myred!20}
		\multirow{-2}{*}{\small $\mathbf{9}$} & 
		\cellcolor{myred!20}
		\multirow{-2}{*}{\small $\mathbf{63}$} & 
		\cellcolor{myred!20}
		\multirow{-2}{*}{\small \textbf{4m 50s}} & 
		\textcolor{gray}{\scriptsize $1.042\%$} & 
		\multirow{-2}{*}{\small $287$} & 
		\multirow{-2}{*}{\small $6$} & 
		\multirow{-2}{*}{\small $65$} & 
		\multirow{-2}{*}{\small 5m 14s} \\
		\cellcolor{myred!20} & \textsc{C} & {\small 
		$99.696\%$} &  &  &  &  & 
		\cellcolor{myred!20} {\small 
		$\mathbf{100.000\%}$} & \cellcolor{myred!20} 
		& 
		\cellcolor{myred!20} & \cellcolor{myred!20} & 
		\cellcolor{myred!20} & {\small 
		$100.000\%$} &  &  &  &  \\
		\cellcolor{myred!20} & 
		\textcolor{gray}{\scriptsize $8.333\%$} & 
		\textcolor{gray}{\scriptsize $8.308\%$} & 
		\multirow{-2}{*}{\small $3173$} & 
		\multirow{-2}{*}{\small $20$} & 
		\multirow{-2}{*}{\small $1211$} & 
		\multirow{-2}{*}{\small 36m 12s} & 
		\cellcolor{myred!20}
		\textcolor{gray}{\scriptsize 
		$\mathbf{8.333\%}$} & 
		\cellcolor{myred!20}
		\multirow{-2}{*}{\small $\mathbf{2668}$} & 
		\cellcolor{myred!20}
		\multirow{-2}{*}{\small $\mathbf{13}$} & 
		\cellcolor{myred!20}
		\multirow{-2}{*}{\small $\mathbf{417}$} & 
		\cellcolor{myred!20}
		\multirow{-2}{*}{\small \textbf{17m 40s}} & 
		\textcolor{gray}{\scriptsize $8.333\%$} & 
		\multirow{-2}{*}{\small $2887$} & 
		\multirow{-2}{*}{\small $10$} & 
		\multirow{-2}{*}{\small $519$} & 
		\multirow{-2}{*}{\small 29m 52s} \\
		\rowcolor{gray!10} \cellcolor{myred!20} & 
		\textsc{B} & {\small $97.318\%$} &  &  &  &  & 
		\cellcolor{myred!20} {\small 
		$\mathbf{99.991\%}$} 
		&\cellcolor{myred!20}  & \cellcolor{myred!20} 
		& \cellcolor{myred!20} & \cellcolor{myred!20} 
		& {\small 
		$99.978\%$} &  &  &  &  \\
		\rowcolor{gray!10} \cellcolor{myred!20} & 
		\textcolor{gray}{\scriptsize $50\%$} & 
		\textcolor{gray}{\scriptsize $48.659\%$} & 
		\multirow{-2}{*}{\small $15415$} & 
		\multirow{-2}{*}{\small $61$} & 
		\multirow{-2}{*}{\small $5646$} & 
		\multirow{-2}{*}{\small 1h 39m 36s} & 
		\cellcolor{myred!20}
		\textcolor{gray}{\scriptsize 
		$\mathbf{49.996\%}$} & 
		\cellcolor{myred!20}
		\multirow{-2}{*}{\small $\mathbf{12617}$} & 
		\cellcolor{myred!20}
		\multirow{-2}{*}{\small $\mathbf{34}$} & 
		\cellcolor{myred!20}
		\multirow{-2}{*}{\small $\mathbf{2112}$} & 
		\cellcolor{myred!20}
		\multirow{-2}{*}{\small \textbf{1h 1m 19s}} & 
		\textcolor{gray}{\scriptsize $49.989\%$} & 
		\multirow{-2}{*}{\small $13973$} & 
		\multirow{-2}{*}{\small $24$} & 
		\multirow{-2}{*}{\small $2405$} & 
		\multirow{-2}{*}{\small 1h 14m 19s} \\
		\cellcolor{myred!20} & \textsc{A} & {\small 
		$94.032\%$} &  &  &  &  & 
		\cellcolor{myred!20} {\small 
		$\mathbf{99.935\%}$} & \cellcolor{myred!20} 
		& 
		\cellcolor{myred!20} & \cellcolor{myred!20} & 
		\cellcolor{myred!20} & {\small 
		$99.896\%$} &  &  &  &  \\
		\cellcolor{myred!20} \multirow{-12}{*}{$20$} & 
		\textcolor{gray}{\scriptsize $100\%$} & 
		\textcolor{gray}{\scriptsize $94.032\%$} & 
		\multirow{-2}{*}{\small $18642$} & 
		\multirow{-2}{*}{\small $70$} & 
		\multirow{-2}{*}{\small $8700$} & 
		\multirow{-2}{*}{\small 2h 30m 46s} & 
		\cellcolor{myred!20}
		\textcolor{gray}{\scriptsize 
		$\mathbf{99.935\%}$} & 
		\cellcolor{myred!20}
		\multirow{-2}{*}{\small $\mathbf{15445}$} & 
		\cellcolor{myred!20}
		\multirow{-2}{*}{\small $\mathbf{40}$} & 
		\cellcolor{myred!20}
		\multirow{-2}{*}{\small $\mathbf{3481}$} & 
		\cellcolor{myred!20}
		\multirow{-2}{*}{\small \textbf{1h 29m} } & 
		\textcolor{gray}{\scriptsize $99.896\%$} & 
		\multirow{-2}{*}{\small $17784$} & 
		\multirow{-2}{*}{\small $39$} & 
		\multirow{-2}{*}{\small $4076$} & 
		\multirow{-2}{*}{\small 1h 47m 7s} \\
		\hline
		\rowcolor{gray!10} \cellcolor{myorange!20} & 
		\textsc{F} & {\small $99.931\%$} &  &  &  &  & 
		\cellcolor{myorange!20} {\small 
		$\mathbf{99.961\%}$} 
		& \cellcolor{myorange!20} & 
		\cellcolor{myorange!20} & 
		\cellcolor{myorange!20} & 
		\cellcolor{myorange!20} & {\small 
		$99.957\%$} &  &  &  &  \\
		\rowcolor{gray!10} \cellcolor{myorange!20} & 
		\textcolor{gray}{\scriptsize $0.009\%$} & 
		\textcolor{gray}{\scriptsize $0.009\%$} & 
		\multirow{-2}{*}{\small $11$} & 
		\multirow{-2}{*}{\small $0$} & 
		\multirow{-2}{*}{\small $0$} & 
		\multirow{-2}{*}{\small 3m 5s} & 
		\cellcolor{myorange!20}
		\textcolor{gray}{\scriptsize 
		$\mathbf{0.009\%}$} & 
		\cellcolor{myorange!20}
		\multirow{-2}{*}{\small $\mathbf{17}$} & 
		\cellcolor{myorange!20}
		\multirow{-2}{*}{\small $\mathbf{0}$} & 
		\cellcolor{myorange!20}
		\multirow{-2}{*}{\small $\mathbf{0}$} & 
		\cellcolor{myorange!20}
		\multirow{-2}{*}{\small \textbf{3m 2s}} & 
		\textcolor{gray}{\scriptsize $0.009\%$} & 
		\multirow{-2}{*}{\small $10$} & 
		\multirow{-2}{*}{\small $0$} & 
		\multirow{-2}{*}{\small $0$} & 
		\multirow{-2}{*}{\small 2m 36s} \\
		\cellcolor{myorange!20} & \textsc{E} & {\small 
		$99.583\%$} &  &  &  &  & 
		\cellcolor{myorange!20} {\small 
		$\mathbf{99.783\%}$} & 
		\cellcolor{myorange!20} & 
		\cellcolor{myorange!20} & 
		\cellcolor{myorange!20} & 
		\cellcolor{myorange!20} & {\small 
		$99.753\%$} &  &  &  &  \\
		\cellcolor{myorange!20} & 
		\textcolor{gray}{\scriptsize $0.104\%$} & 
		\textcolor{gray}{\scriptsize $0.104\%$} & 
		\multirow{-2}{*}{\small $61$} & 
		\multirow{-2}{*}{\small $0$} & 
		\multirow{-2}{*}{\small $0$} & 
		\multirow{-2}{*}{\small 3m 6s} & 
		\cellcolor{myorange!20}
		\textcolor{gray}{\scriptsize 
		$\mathbf{0.104\%}$} & 
		\cellcolor{myorange!20}
		\multirow{-2}{*}{\small $\mathbf{89}$} & 
		\cellcolor{myorange!20}
		\multirow{-2}{*}{\small $\mathbf{0}$} & 
		\cellcolor{myorange!20}
		\multirow{-2}{*}{\small $\mathbf{0}$} & 
		\cellcolor{myorange!20}
		\multirow{-2}{*}{\small \textbf{3m 10s}} & 
		\textcolor{gray}{\scriptsize $0.104\%$} & 
		\multirow{-2}{*}{\small $74$} & 
		\multirow{-2}{*}{\small $0$} & 
		\multirow{-2}{*}{\small $0$} & 
		\multirow{-2}{*}{\small 2m 44s} \\
		\rowcolor{gray!10} \cellcolor{myorange!20} & 
		\textsc{D} & {\small $97.917\%$} &  &  &  &  & 
		\cellcolor{myorange!20} {\small 
		$\mathbf{99.258\%}$} 
		& \cellcolor{myorange!20} & 
		\cellcolor{myorange!20} & 
		\cellcolor{myorange!20} & 
		\cellcolor{myorange!20} & {\small 
		$98.984\%$} &  &  &  &  \\
		\rowcolor{gray!10} \cellcolor{myorange!20} & 
		\textcolor{gray}{\scriptsize $1.042\%$} & 
		\textcolor{gray}{\scriptsize $1.020\%$} & 
		\multirow{-2}{*}{\small $151$} & 
		\multirow{-2}{*}{\small $0$} & 
		\multirow{-2}{*}{\small $0$} & 
		\multirow{-2}{*}{\small 2m 56s} & 
		\cellcolor{myorange!20}
		\textcolor{gray}{\scriptsize 
		$\mathbf{1.034\%}$} & 
		\cellcolor{myorange!20}
		\multirow{-2}{*}{\small $\mathbf{297}$} & 
		\cellcolor{myorange!20}
		\multirow{-2}{*}{\small $\mathbf{0}$} & 
		\cellcolor{myorange!20}
		\multirow{-2}{*}{\small $\mathbf{0}$} & 
		\cellcolor{myorange!20}
		\multirow{-2}{*}{\small \textbf{3m 41s}} & 
		\textcolor{gray}{\scriptsize $1.031\%$} & 
		\multirow{-2}{*}{\small $477$} & 
		\multirow{-2}{*}{\small $0$} & 
		\multirow{-2}{*}{\small $0$} & 
		\multirow{-2}{*}{\small 2m 58s} \\
		\cellcolor{myorange!20} & \textsc{C} & {\small 
		$83.503\%$} &  &  &  &  & {\small 
		$95.482\%$} &  &  &  &  & 
		\cellcolor{myorange!20} {\small 
		$\mathbf{93.225\%}$} & 
		\cellcolor{myorange!20} & 
		\cellcolor{myorange!20} & 
		\cellcolor{myorange!20} & 
		\cellcolor{myorange!20} \\
		\cellcolor{myorange!20} & 
		\textcolor{gray}{\scriptsize $8.333\%$} & 
		\textcolor{gray}{\scriptsize $6.958\%$} & 
		\multirow{-2}{*}{\small $506$} & 
		\multirow{-2}{*}{\small $2$} & 
		\multirow{-2}{*}{\small $3$} & 
		\multirow{-2}{*}{\small 2h 1m } & 
		\textcolor{gray}{\scriptsize $7.956\%$} & 
		\multirow{-2}{*}{\small $885$} & 
		\multirow{-2}{*}{\small $25$} & 
		\multirow{-2}{*}{\small $34$} & 
		\multirow{-2}{*}{\small >13h} & 
		\cellcolor{myorange!20}
		\textcolor{gray}{\scriptsize 
		$\mathbf{7.768\%}$} & 
		\cellcolor{myorange!20}
		\multirow{-2}{*}{\small $\mathbf{1145}$} & 
		\cellcolor{myorange!20}
		\multirow{-2}{*}{\small $\mathbf{23}$} & 
		\cellcolor{myorange!20}
		\multirow{-2}{*}{\small $\mathbf{33}$} & 
		\cellcolor{myorange!20}
		\multirow{-2}{*}{\small \textbf{12h 57m 37s}} 
		\\
		\rowcolor{gray!10} \cellcolor{myorange!20} & 
		\textsc{B} & \cellcolor{myorange!20} {\small 
		$\mathbf{25.634\%}$} & 
		\cellcolor{myorange!20} & 
		\cellcolor{myorange!20} & 
		\cellcolor{myorange!20} & 
		\cellcolor{myorange!20} & 
		{\small $76.563\%$} &  &  &  &  & {\small 
		$63.906\%$} &  &  &  &  \\
		\rowcolor{gray!10} \cellcolor{myorange!20} & 
		\textcolor{gray}{\scriptsize $50\%$} & 
		\cellcolor{myorange!20}
		\textcolor{gray}{\scriptsize 
		$\mathbf{12.817\%}$} & 
		\cellcolor{myorange!20}
		\multirow{-2}{*}{\small $\mathbf{5516}$} & 
		\cellcolor{myorange!20}
		\multirow{-2}{*}{\small $\mathbf{7}$} & 
		\cellcolor{myorange!20}
		\multirow{-2}{*}{\small $\mathbf{11}$} & 
		\cellcolor{myorange!20}
		\multirow{-2}{*}{\small \textbf{1h 28m 6s}} & 
		\textcolor{gray}{\scriptsize $38.281\%$} & 
		\multirow{-2}{*}{\small $4917$} & 
		\multirow{-2}{*}{\small $123$} & 
		\multirow{-2}{*}{\small $182$} & 
		\multirow{-2}{*}{\small >13h} & 
		\textcolor{gray}{\scriptsize $31.953\%$} & 
		\multirow{-2}{*}{\small $7139$} & 
		\multirow{-2}{*}{\small $117$} & 
		\multirow{-2}{*}{\small $152$} & 
		\multirow{-2}{*}{\small >13h} \\
		\cellcolor{myorange!20} & \textsc{A} & {\small 
		$0.052\%$} &  &  &  &  & 
		\cellcolor{myorange!20} {\small 
		$\mathbf{61.385\%}$} 
		& \cellcolor{myorange!20} & 
		\cellcolor{myorange!20} & 
		\cellcolor{myorange!20} & 
		\cellcolor{myorange!20} & {\small 
		$43.698\%$} &  &  &  &  
		\\
		\cellcolor{myorange!20} 
		\multirow{-12}{*}{$80$} & 
		\textcolor{gray}{\scriptsize $100\%$} & 
		\textcolor{gray}{\scriptsize $0.052\%$} & 
		\multirow{-2}{*}{\small $12$} & 
		\multirow{-2}{*}{\small $0$} & 
		\multirow{-2}{*}{\small $0$} & 
		\multirow{-2}{*}{\small 25m 51s} & 
		\cellcolor{myorange!20}
		\textcolor{gray}{\scriptsize 
		$\mathbf{61.385\%}$} & 
		\cellcolor{myorange!20}
		\multirow{-2}{*}{\small $\mathbf{5156}$} & 
		\cellcolor{myorange!20}
		\multirow{-2}{*}{\small $\mathbf{73}$} & 
		\cellcolor{myorange!20}
		\multirow{-2}{*}{\small $\mathbf{102}$} & 
		\cellcolor{myorange!20}
		\multirow{-2}{*}{\small \textbf{10h 25m 2s}} & 
		\textcolor{gray}{\scriptsize $43.698\%$} & 
		\multirow{-2}{*}{\small $4757$} & 
		\multirow{-2}{*}{\small $68$} & 
		\multirow{-2}{*}{\small $88$} & 
		\multirow{-2}{*}{\small >13h} \\
		\hline
		\rowcolor{gray!10} \cellcolor{myyellow!20} & 
		\textsc{F} & {\small $99.931\%$} &  &  &  &  & 
		\cellcolor{myyellow!20} {\small 
		$\mathbf{99.944\%}$} & 
		\cellcolor{myyellow!20} & 
		\cellcolor{myyellow!20} & 
		\cellcolor{myyellow!20} & 
		\cellcolor{myyellow!20} & {\small 
		$99.931\%$} &  &  &  &  \\
		\rowcolor{gray!10} \cellcolor{myyellow!20} & 
		\textcolor{gray}{\scriptsize $0.009\%$} & 
		\textcolor{gray}{\scriptsize $0.009\%$} & 
		\multirow{-2}{*}{\small $6$} & 
		\multirow{-2}{*}{\small $0$} & 
		\multirow{-2}{*}{\small $0$} & 
		\multirow{-2}{*}{\small 3m 15s} & 
		\cellcolor{myyellow!20}
		\textcolor{gray}{\scriptsize 
		$\mathbf{0.009\%}$} & 
		\cellcolor{myyellow!20}
		\multirow{-2}{*}{\small $\mathbf{9}$} & 
		\cellcolor{myyellow!20}
		\multirow{-2}{*}{\small $\mathbf{0}$} & 
		\cellcolor{myyellow!20}
		\multirow{-2}{*}{\small $\mathbf{0}$} & 
		\cellcolor{myyellow!20}
		\multirow{-2}{*}{\small \textbf{3m 35s}} & 
		\textcolor{gray}{\scriptsize $0.009\%$} & 
		\multirow{-2}{*}{\small $6$} & 
		\multirow{-2}{*}{\small $0$} & 
		\multirow{-2}{*}{\small $0$} & 
		\multirow{-2}{*}{\small 3m 30s} \\
		\cellcolor{myyellow!20} & \textsc{E} & {\small 
		$99.583\%$} &  &  &  &  & 
		\cellcolor{myyellow!20} {\small 
		$\mathbf{99.627\%}$} & 
		\cellcolor{myyellow!20} & 
		\cellcolor{myyellow!20} & 
		\cellcolor{myyellow!20} & 
		\cellcolor{myyellow!20} & {\small 
		$99.583\%$} &  &  &  &  \\
		\cellcolor{myyellow!20} & 
		\textcolor{gray}{\scriptsize $0.104\%$} & 
		\textcolor{gray}{\scriptsize $0.104\%$} & 
		\multirow{-2}{*}{\small $121$} & 
		\multirow{-2}{*}{\small $0$} & 
		\multirow{-2}{*}{\small $0$} & 
		\multirow{-2}{*}{\small 3m 39s} & 
		\cellcolor{myyellow!20}
		\textcolor{gray}{\scriptsize 
		$\mathbf{0.104\%}$} & 
		\cellcolor{myyellow!20}
		\multirow{-2}{*}{\small $\mathbf{120}$} & 
		\cellcolor{myyellow!20}
		\multirow{-2}{*}{\small $\mathbf{0}$} & 
		\cellcolor{myyellow!20}
		\multirow{-2}{*}{\small $\mathbf{0}$} & 
		\cellcolor{myyellow!20}
		\multirow{-2}{*}{\small \textbf{6m 34s}} & 
		\textcolor{gray}{\scriptsize $0.104\%$} & 
		\multirow{-2}{*}{\small $31$} & 
		\multirow{-2}{*}{\small $0$} & 
		\multirow{-2}{*}{\small $0$} & 
		\multirow{-2}{*}{\small 4m 22s} \\
		\rowcolor{gray!10} \cellcolor{myyellow!20} & 
		\textsc{D} & {\small $97.917\%$} &  &  &  &  & 
		\cellcolor{myyellow!20} {\small 
		$\mathbf{98.247\%}$} & 
		\cellcolor{myyellow!20} & 
		\cellcolor{myyellow!20} & 
		\cellcolor{myyellow!20} & 
		\cellcolor{myyellow!20} & {\small 
		$97.917\%$} &  &  &  &  \\
		\rowcolor{gray!10} \cellcolor{myyellow!20} & 
		\textcolor{gray}{\scriptsize $1.042\%$} & 
		\textcolor{gray}{\scriptsize $1.020\%$} & 
		\multirow{-2}{*}{\small $151$} & 
		\multirow{-2}{*}{\small $0$} & 
		\multirow{-2}{*}{\small $0$} & 
		\multirow{-2}{*}{\small 6m 18s} & 
		\cellcolor{myyellow!20}
		\textcolor{gray}{\scriptsize 
		$\mathbf{1.024\%}$} & 
		\cellcolor{myyellow!20}
		\multirow{-2}{*}{\small $\mathbf{597}$} & 
		\cellcolor{myyellow!20}
		\multirow{-2}{*}{\small $\mathbf{0}$} & 
		\cellcolor{myyellow!20}
		\multirow{-2}{*}{\small $\mathbf{0}$} & 
		\cellcolor{myyellow!20}
		\multirow{-2}{*}{\small \textbf{21m 9s}} & 
		\textcolor{gray}{\scriptsize $1.020\%$} & 
		\multirow{-2}{*}{\small $301$} & 
		\multirow{-2}{*}{\small $0$} & 
		\multirow{-2}{*}{\small $0$} & 
		\multirow{-2}{*}{\small 9m 35s} \\
		\cellcolor{myyellow!20} & \textsc{C} & {\small 
		$83.333\%$} &  &  &  &  & 
		\cellcolor{myyellow!20} {\small 
		$\mathbf{88.294\%}$} & 
		\cellcolor{myyellow!20} & 
		\cellcolor{myyellow!20} & 
		\cellcolor{myyellow!20} & 
		\cellcolor{myyellow!20} & {\small 
		$83.342\%$} &  &  &  &  \\
		\cellcolor{myyellow!20} & 
		\textcolor{gray}{\scriptsize $8.333\%$} & 
		\textcolor{gray}{\scriptsize $6.944\%$} & 
		\multirow{-2}{*}{\small $120$} & 
		\multirow{-2}{*}{\small $0$} & 
		\multirow{-2}{*}{\small $0$} & 
		\multirow{-2}{*}{\small 30m 37s} & 
		\cellcolor{myyellow!20}
		\textcolor{gray}{\scriptsize 
		$\mathbf{7.358\%}$} & 
		\cellcolor{myyellow!20}
		\multirow{-2}{*}{\small $\mathbf{755}$} & 
		\cellcolor{myyellow!20}
		\multirow{-2}{*}{\small $\mathbf{0}$} & 
		\cellcolor{myyellow!20}
		\multirow{-2}{*}{\small $\mathbf{0}$} & 
		\cellcolor{myyellow!20}
		\multirow{-2}{*}{\small \textbf{1h 36m 35s}} & 
		\textcolor{gray}{\scriptsize $6.945\%$} & 
		\multirow{-2}{*}{\small $483$} & 
		\multirow{-2}{*}{\small $0$} & 
		\multirow{-2}{*}{\small $0$} & 
		\multirow{-2}{*}{\small 52m 29s} \\
		\rowcolor{gray!10} \cellcolor{myyellow!20} & 
		\textsc{B} & {\small $25.000\%$} &  &  &  &  & 
		\cellcolor{myyellow!20} {\small 
		$\mathbf{46.063\%}$} & 
		\cellcolor{myyellow!20} & 
		\cellcolor{myyellow!20} & 
		\cellcolor{myyellow!20} & 
		\cellcolor{myyellow!20} & {\small 
		$25.074\%$} &  &  &  &  \\
		\rowcolor{gray!10} \cellcolor{myyellow!20} & 
		\textcolor{gray}{\scriptsize $50\%$} & 
		\textcolor{gray}{\scriptsize $12.500\%$} & 
		\multirow{-2}{*}{\small $5744$} & 
		\multirow{-2}{*}{\small $0$} & 
		\multirow{-2}{*}{\small $0$} & 
		\multirow{-2}{*}{\small 2h 24m 36s} & 
		\cellcolor{myyellow!20}
		\textcolor{gray}{\scriptsize 
		$\mathbf{23.032\%}$} & 
		\cellcolor{myyellow!20}
		\multirow{-2}{*}{\small $\mathbf{4676}$} & 
		\cellcolor{myyellow!20}
		\multirow{-2}{*}{\small $\mathbf{0}$} & 
		\cellcolor{myyellow!20}
		\multirow{-2}{*}{\small $\mathbf{0}$} & 
		\cellcolor{myyellow!20}
		\multirow{-2}{*}{\small \textbf{7h 25m 57s}} & 
		\textcolor{gray}{\scriptsize $12.537\%$} & 
		\multirow{-2}{*}{\small $5762$} & 
		\multirow{-2}{*}{\small $4$} & 
		\multirow{-2}{*}{\small $4$} & 
		\multirow{-2}{*}{\small >13h} \\
		\cellcolor{myyellow!20} & \textsc{A} & {\small 
		$0.000\%$} &  &  &  &  & 
		\cellcolor{myyellow!20} {\small 
		$\mathbf{24.258\%}$} 
		& \cellcolor{myyellow!20} & 
		\cellcolor{myyellow!20} & 
		\cellcolor{myyellow!20} & 
		\cellcolor{myyellow!20} & {\small $0.017\%$} 
		&  &  &  &  \\
		\cellcolor{myyellow!20} 
		\multirow{-12}{*}{$320$} & 
		\textcolor{gray}{\scriptsize $100\%$} & 
		\textcolor{gray}{\scriptsize $0.000\%$} & 
		\multirow{-2}{*}{\small $0$} & 
		\multirow{-2}{*}{\small $0$} & 
		\multirow{-2}{*}{\small $0$} & 
		\multirow{-2}{*}{\small 2h 54m 25s} & 
		\cellcolor{myyellow!20}
		\textcolor{gray}{\scriptsize 
		$\mathbf{24.258\%}$} & 
		\cellcolor{myyellow!20}
		\multirow{-2}{*}{\small $\mathbf{2436}$} & 
		\cellcolor{myyellow!20}
		\multirow{-2}{*}{\small $\mathbf{0}$} & 
		\cellcolor{myyellow!20}
		\multirow{-2}{*}{\small $\mathbf{0}$} & 
		\cellcolor{myyellow!20}
		\multirow{-2}{*}{\small \textbf{9h 41m 36s}} & 
		\textcolor{gray}{\scriptsize $0.017\%$} & 
		\multirow{-2}{*}{\small $4$} & 
		\multirow{-2}{*}{\small $0$} & 
		\multirow{-2}{*}{\small $0$} & 
		\multirow{-2}{*}{\small 5h 3m 33s} \\
		\hline
		\rowcolor{gray!10} \cellcolor{mygreen!20} & 
		\textsc{F} & {\small $99.931\%$} &  &  &  &  & 
		\cellcolor{mygreen!20} {\small 
		$\mathbf{99.948\%}$} & 
		\cellcolor{mygreen!20} & 
		\cellcolor{mygreen!20} 
		&\cellcolor{mygreen!20}  & 
		\cellcolor{mygreen!20} & {\small 
		$99.931\%$} &  &  &  &  \\
		\rowcolor{gray!10} \cellcolor{mygreen!20} & 
		\textcolor{gray}{\scriptsize $0.009\%$} & 
		\textcolor{gray}{\scriptsize $0.009\%$} & 
		\multirow{-2}{*}{\small $11$} & 
		\multirow{-2}{*}{\small $0$} & 
		\multirow{-2}{*}{\small $0$} & 
		\multirow{-2}{*}{\small 7m 35s} & 
		\cellcolor{mygreen!20}
		\textcolor{gray}{\scriptsize 
		$\mathbf{0.009\%}$} & 
		\cellcolor{mygreen!20}
		\multirow{-2}{*}{\small $\mathbf{10}$} & 
		\cellcolor{mygreen!20}
		\multirow{-2}{*}{\small $\mathbf{0}$} & 
		\cellcolor{mygreen!20}
		\multirow{-2}{*}{\small $\mathbf{0}$} & 
		\cellcolor{mygreen!20}
		\multirow{-2}{*}{\small \textbf{24m 42s}} & 
		\textcolor{gray}{\scriptsize $0.009\%$} & 
		\multirow{-2}{*}{\small $6$} & 
		\multirow{-2}{*}{\small $0$} & 
		\multirow{-2}{*}{\small $0$} & 
		\multirow{-2}{*}{\small 7m 6s} \\
		\cellcolor{mygreen!20} & \textsc{E} & {\small 
		$99.583\%$} &  &  &  &  & 
		\cellcolor{mygreen!20} {\small 
		$\mathbf{99.674\%}$} & 
		\cellcolor{mygreen!20} & 
		\cellcolor{mygreen!20} & 
		\cellcolor{mygreen!20} & 
		\cellcolor{mygreen!20} & {\small 
		$99.583\%$} &  &  &  &  \\
		\cellcolor{mygreen!20} & 
		\textcolor{gray}{\scriptsize $0.104\%$} & 
		\textcolor{gray}{\scriptsize $0.104\%$} & 
		\multirow{-2}{*}{\small $31$} & 
		\multirow{-2}{*}{\small $0$} & 
		\multirow{-2}{*}{\small $0$} & 
		\multirow{-2}{*}{\small 15m 49s} & 
		\cellcolor{mygreen!20}
		\textcolor{gray}{\scriptsize 
		$\mathbf{0.104\%}$} & 
		\cellcolor{mygreen!20}
		\multirow{-2}{*}{\small $\mathbf{71}$} & 
		\cellcolor{mygreen!20}
		\multirow{-2}{*}{\small $\mathbf{0}$} & 
		\cellcolor{mygreen!20}
		\multirow{-2}{*}{\small $\mathbf{0}$} & 
		\cellcolor{mygreen!20}
		\multirow{-2}{*}{\small \textbf{51m 52s}} & 
		\textcolor{gray}{\scriptsize $0.104\%$} & 
		\multirow{-2}{*}{\small $31$} & 
		\multirow{-2}{*}{\small $0$} & 
		\multirow{-2}{*}{\small $0$} & 
		\multirow{-2}{*}{\small 15m 14s} \\
		\rowcolor{gray!10} \cellcolor{mygreen!20} & 
		\textsc{D} & {\small $97.917\%$} &  &  &  &  & 
		\cellcolor{mygreen!20} {\small 
		$\mathbf{98.668\%}$} & 
		\cellcolor{mygreen!20} & 
		\cellcolor{mygreen!20} & 
		\cellcolor{mygreen!20} & 
		\cellcolor{mygreen!20} & {\small 
		$97.917\%$} &  &  &  &  \\
		\rowcolor{gray!10} \cellcolor{mygreen!20} & 
		\textcolor{gray}{\scriptsize $1.042\%$} & 
		\textcolor{gray}{\scriptsize $1.020\%$} & 
		\multirow{-2}{*}{\small $151$} & 
		\multirow{-2}{*}{\small $0$} & 
		\multirow{-2}{*}{\small $0$} & 
		\multirow{-2}{*}{\small 1h 49s} & 
		\cellcolor{mygreen!20}
		\textcolor{gray}{\scriptsize 
		$\mathbf{1.028\%}$} & 
		\cellcolor{mygreen!20}
		\multirow{-2}{*}{\small $\mathbf{557}$} & 
		\cellcolor{mygreen!20}
		\multirow{-2}{*}{\small $\mathbf{0}$} & 
		\cellcolor{mygreen!20}
		\multirow{-2}{*}{\small $\mathbf{0}$} & 
		\cellcolor{mygreen!20}
		\multirow{-2}{*}{\small \textbf{3h 31m 45s}} & 
		\textcolor{gray}{\scriptsize $1.020\%$} & 
		\multirow{-2}{*}{\small $301$} & 
		\multirow{-2}{*}{\small $0$} & 
		\multirow{-2}{*}{\small $0$} & 
		\multirow{-2}{*}{\small 1h 3m 33s} \\
		\cellcolor{mygreen!20} & \textsc{C} & 
		\cellcolor{mygreen!20} {\small 
		$\mathbf{83.333\%}$} & 
		\cellcolor{mygreen!20} & 
		\cellcolor{mygreen!20} & 
		\cellcolor{mygreen!20} & 
		\cellcolor{mygreen!20} &  &  &  &  &  & 
		{\small $83.333\%$} &  &  &  &  \\
		\cellcolor{mygreen!20} & 
		\textcolor{gray}{\scriptsize $8.333\%$} & 
		\cellcolor{mygreen!20}
		\textcolor{gray}{\scriptsize 
		$\mathbf{6.944\%}$} & 
		\cellcolor{mygreen!20}
		\multirow{-2}{*}{\small $\mathbf{481}$} & 
		\cellcolor{mygreen!20}
		\multirow{-2}{*}{\small $\mathbf{0}$} & 
		\cellcolor{mygreen!20}
		\multirow{-2}{*}{\small $\mathbf{0}$} & 
		\cellcolor{mygreen!20}
		\multirow{-2}{*}{\small \textbf{7h 11m 39s}} & 
		\multirow{-2}{*}{\small $-$} & 
		\multirow{-2}{*}{\small $-$} & 
		\multirow{-2}{*}{\small $-$} & 
		\multirow{-2}{*}{\small $-$} & 
		\multirow{-2}{*}{\small >13h} & 
		\textcolor{gray}{\scriptsize $6.944\%$} & 
		\multirow{-2}{*}{\small $481$} & 
		\multirow{-2}{*}{\small $0$} & 
		\multirow{-2}{*}{\small $0$} & 
		\multirow{-2}{*}{\small 7h 12m 57s} \\
		\rowcolor{gray!10} \cellcolor{mygreen!20} & 
		\textsc{B} &  &  &  &  &  &  &  &  &  &  &  &  &  
		&  &  \\
		\rowcolor{gray!10} \cellcolor{mygreen!20} & 
		\textcolor{gray}{\scriptsize $50\%$} & 
		\multirow{-2}{*}{\small $-$} & 
		\multirow{-2}{*}{\small $-$} & 
		\multirow{-2}{*}{\small $-$} & 
		\multirow{-2}{*}{\small $-$} & 
		\multirow{-2}{*}{\small >13h} & 
		\multirow{-2}{*}{\small $-$} & 
		\multirow{-2}{*}{\small $-$} & 
		\multirow{-2}{*}{\small $-$} & 
		\multirow{-2}{*}{\small $-$} & 
		\multirow{-2}{*}{\small >13h} & 
		\multirow{-2}{*}{\small $-$} & 
		\multirow{-2}{*}{\small $-$} & 
		\multirow{-2}{*}{\small $-$} & 
		\multirow{-2}{*}{\small $-$} & 
		\multirow{-2}{*}{\small >13h} \\
		\cellcolor{mygreen!20} & \textsc{A} &  &  &  &  
		&  &  &  &  &  &  &  &  &  &  &  \\
		\cellcolor{mygreen!20} 
		\multirow{-12}{*}{$1280$} & 
		\textcolor{gray}{\scriptsize $100\%$} & 
		\multirow{-2}{*}{\small $-$} & 
		\multirow{-2}{*}{\small $-$} & 
		\multirow{-2}{*}{\small $-$} & 
		\multirow{-2}{*}{\small $-$} & 
		\multirow{-2}{*}{\small >13h} & 
		\multirow{-2}{*}{\small $-$} & 
		\multirow{-2}{*}{\small $-$} & 
		\multirow{-2}{*}{\small $-$} & 
		\multirow{-2}{*}{\small $-$} & 
		\multirow{-2}{*}{\small >13h} & 
		\multirow{-2}{*}{\small $-$} & 
		\multirow{-2}{*}{\small $-$} & 
		\multirow{-2}{*}{\small $-$} & 
		\multirow{-2}{*}{\small $-$} & 
		\multirow{-2}{*}{\small >13h} \\
		\hline
	\end{tabular}
}

\end{table}

\paragraph{\textbf{RQ4: Effect of Analyzed 
Input 
Space on Scalability.}}
As said above, the analysis of the models 
considered in 
Table~\ref{tbl:censusT}  is conducted \emph{on the 
entire 
input space}. 
In practice, as already mentioned, one might be 
interested in just a portion 
of the input space, e.g., depending on the 
probability distribution. 
More generally, 
\textbf{we argue that the size of the 
analyzed input space (rather than the size of the 
analyzed neural
network) is the most 
important factor that affects the performance of the 
analysis}. 
To support this claim, we trained even larger 
models and analyzed them with respect to queries 
exercising different input space sizes. 
Table~\ref{tbl:censusNEW} shows the results. The 
first column again shows the total number of 
hidden nodes for each trained model. In particular, 
the models we analyzed have the following number 
of hidden layers and nodes per layer (from top to 
bottom): 4 and 5; 8 and 10; 16 and 20; 32 and 40. 
Column \textsc{query} shows the query used for the 
analysis and the corresponding exercised input 
space size. Specifically, the queries identify people 
with the following characteristics:
\begin{description}
	\item[$A$:] $\textsc{true}$  
	\hfill{exercised input space: $100.0\%$}
	
	\item[$B$:] $A \land \text{age}\footnote{This 
	corresponds to $age \leq 0.5$ with min-max 
	scaling between $0$ and $1$.} \leq 53.5$
	\hfill{exercised input space: $50.00\%$}
	
	\item[$C$:] $B \land \text{race} = \text{white}$
	\hfill{exercised input space: $8.333\%$ ($3$ race
	choices)}
	
	\item[$D$:] $C \land \text{work class} = 
	\text{private}$
	\hfill{exercised input space: $1.043\%$ ($4$ 
	work class choices)}

	\item[$E$:] $D \land \text{marital status} = 
	\text{single}$
	\hfill{exercised input space: $0.104\%$ ($5$ 
	marital status choices)}

	\item[$F$:] $E \land \text{occupation} = 
	\text{blue-collar}$
	\hfill{exercised input space: $0.009\%$  ($6$ 
	occupation choices)}
\end{description}
For the analysis budget, we used $\lowerbound = 
0.25$, 
$\upperbound = 0.1 * \size{\model}$, and a 
time limit of $13$h. Column \textsc{input} shows, 
for each domain used for the forward pre-analysis, 
the coverage of the exercised input space 
(i.e., the percentage of the input space that satisfies 
the query and was completed 
by the analysis) and the 
corresponding input-space coverage (i.e., the same 
percentage but this time scaled to the entire input 
space).
Columns $\upperbound$, 
$\size{\text{C}}$, $\size{\text{F}}$, and 
\textsc{time} are as before.
Where a timeout is indicated (i.e., \textsc{time} 
$>13$h) and the values for the \textsc{input}, 
$\size{\text{C}}$, and $\size{\text{F}}$ columns are 
missing, it means that the timeout occurred during 
the pre-analysis; otherwise, it happened during the 
backward analysis.
For each model and query, we 
highlighted the configuration (i.e., the abstract 
domain used for 
the pre-analysis)
that achieved the highest input-space coverage 
with 
the shortest analysis running time.
Note that, where the $\size{\text{F}}$ column only 
contains zeros, it means that the backward analysis 
had no activation patterns to explore; this implies 
that the entire covered input space (i.e., the 
percentage shown in the \textsc{input} column) was 
already 
certified to be fair by the forward analysis.

Overall, we observe that \emph{whenever the 
analyzed 
input space is small enough} (i.e., queries $D-F$), 
\emph{the size of the neural network has little 
influence on the input space coverage} and slightly 
impacts the analysis running time, 
independently of the domain used for the forward 
pre-analysis. Instead, for larger analyzed input 
spaces (i.e., queries $A-C$) performance degrades 
quickly for larger neural networks. These results 
thus support our claim. Again, as expected, we 
observe that the 
\textsc{symbolic} domain generally is the better 
choice for the forward pre-analysis, in particular for 
queries exercising a larger input space or larger 
neural networks.

\paragraph{\textbf{RQ5: 
Scalability-vs-Precision Tradeoff.}}
To evaluate the effect of the analysis budget (bounds $\lowerbound$
and $\upperbound$), we analyzed a model using different budget
configurations. For this experiment, we used the Japanese Credit
Screening\footnote{\url{https://archive.ics.uci.edu/ml/datasets/Japanese+Credit+Screening}}
dataset, which we made fair with respect to gender. Our 2-class model (17 inputs and 4 hidden layers with 5 nodes each) had a
classification accuracy of 86\%. Note that accuracy does not increase
by adding more layers or nodes per layer, in fact, it may
significantly decrease --- we tried up to 100 hidden layers with 100
nodes each.



\begin{table}[t]
	\caption{Comparison of Different
		Analysis Configurations (Japanese Credit
		Screening) --- $12$ 
		CPUs}\label{tbl:japanese12T}
\centering
\resizebox{\textwidth}{!}{%
	\begin{tabular}{|c|c|cc|cc|c|cc|cc|c|cc|cc|c|}
		\hline
		\multirow{2}{*}{$\lowerbound$} & 
		\multirow{2}{*}{\textsc{U}} &     
		\multicolumn{5}{c|}{$\usemark{pentagon*}$ 
		\textsc{boxes}} &     
		\multicolumn{5}{c|}{$\usemark{triangle*}$ 
		\textsc{symbolic}} &     
		\multicolumn{5}{c|}{$\usemark{star}$ 
		\textsc{deeppoly}} \\
		& & \textsc{input} & $\size{\text{C}}$ & 		
		\multicolumn{2}{c|}{$\size{\text{F}}$} & 		
		\textsc{time} &         \textsc{input} & 
		$\size{\text{C}}$ &         		
		\multicolumn{2}{c|}{$\size{\text{F}}$} 
		&         		\textsc{time} & 		\textsc{input} 
		& $\size{\text{C}}$ & 		
		\multicolumn{2}{c|}{$\size{\text{F}}$} & 		
		\textsc{time} \\
		\hline
		\rowcolor{gray!10} \cellcolor{myred!20} & $4$ 
		& {\small $15.28\%$} & {\small $37$} & 
		{\small $0$} & {\small $0$} & {\small 8s} & 
		{\small $58.33\%$} & {\small $79$} & {\small 
		$8$} & {\small $20$} & {\small 1m 26s} & 
		{\small $69.79\%$} & {\small $115$} & {\small 
		$10$} & {\small $39$} & {\small 3m 18s} \\
		\cellcolor{myred!20} & $6$ & {\small 
		$17.01\%$} & {\small $39$} & {\small $6$} & 
		{\small $6$} & {\small 51s} & {\small 
		$69.10\%$} & {\small $129$} & {\small $22$} 
		& {\small $61$} & {\small 5m 41s} & {\small 
		$80.56\%$} & {\small $104$} & {\small $23$} 
		& {\small $51$} & {\small 7m 53s} \\
		\rowcolor{gray!10} \cellcolor{myred!20} & $8$ 
		& {\small $51.39\%$} & {\small $90$} & 
		{\small $28$} & {\small $85$} & {\small 12m 
		2s} & {\small $82.64\%$} & {\small $88$} & 
		{\small $31$} & {\small $67$} & {\small 12m 
		35s} & {\small $91.32\%$} & {\small $84$} & 
		{\small $27$} & {\small $56$} & {\small 19m 
		33s} \\
		\multirow{-4}{*}{\cellcolor{myred!20}$0.5$} 
		&$10$ & {\small $79.86\%$} & {\small $89$} 
		& {\small $34$} & {\small $89$} & {\small 34m 
		15s} & {\small $93.06\%$} & {\small $98$} & 
		{\small $40$} & {\small $83$} & {\small 42m 
		32s} & {\small $96.88\%$} & {\small $83$} & 
		{\small $29$} & {\small $58$} & {\small 43m 
		39s} \\
		\hline
		\rowcolor{gray!10} \cellcolor{myorange!20} & 
		$4$ & {\small $59.09\%$} & {\small $1115$} & 
		{\small $20$} & {\small $415$} & {\small 54m 
		32s} & {\small $95.94\%$} & {\small $884$} & 
		{\small $39$} & {\small $484$} & {\small 54m 
		31s} & {\small $98.26\%$} & {\small $540$} & 
		{\small $65$} & {\small $293$} & {\small 14m 
		29s} \\
		\cellcolor{myorange!20} & $6$ & {\small 
		$83.77\%$} & {\small $1404$} & {\small $79$} 
		& {\small $944$} & {\small 37m 19s} & {\small 
		$98.68\%$} & {\small $634$} & {\small $66$} 
		& {\small $376$} & {\small 23m 31s} & {\small 
		$99.70\%$} & {\small $322$} & {\small $79$} 
		& {\small $205$} & {\small 13m 25s} \\
		\rowcolor{gray!10} \cellcolor{myorange!20} & 
		$8$ & {\small $96.07\%$} & {\small $869$} & 
		{\small $140$} & {\small $761$} & {\small 1h 
		7m 29s} & {\small $99.72\%$} & {\small 
		$310$} & {\small $67$} & {\small $247$} & 
		{\small 1h 3m 33s} & {\small $99.98\%$} & 
		{\small $247$} & {\small $69$} & {\small 
		$177$} & {\small 22m 52s} \\
		\multirow{-4}{*}{\cellcolor{myorange!20}$0.25$}
		 &$10$ & {\small $99.54\%$} & {\small $409$} 
		& {\small $93$} & {\small $403$} & {\small 1h 
		35m 20s} & {\small $99.98\%$} & {\small 
		$195$} & {\small $52$} & {\small $176$} & 
		{\small 1h 2m 13s} & {\small $100.00\%$} & 
		{\small $111$} & {\small $47$} & {\small 
		$87$} & {\small 34m 56s} \\
		\hline
		\rowcolor{gray!10} \cellcolor{myyellow!20} & 
		$4$ & {\small $97.13\%$} & {\small $12449$} 
		& {\small $200$} & {\small $9519$} & {\small 
		3h 33m 48s} & {\small $99.99\%$} & {\small 
		$1101$} & {\small $60$} & {\small $685$} & 
		{\small 47m 46s} & {\small $99.99\%$} & 
		{\small $768$} & {\small $81$} & {\small 
		$415$} & {\small 19m 1s} \\
		\cellcolor{myyellow!20} & $6$ & {\small 
		$99.83\%$} & {\small $5919$} & {\small 
		$276$} & {\small $4460$} & {\small 3h 23m } 
		& {\small $100.00\%$} & {\small $988$} & 
		{\small $77$} & {\small $606$} & {\small 26m 
		47s} & {\small $100.00\%$} & {\small $489$} 
		& {\small $80$} & {\small $298$} & {\small 
		16m 54s} \\
		\rowcolor{gray!10} \cellcolor{myyellow!20} & 
		$8$ & {\small $99.98\%$} & {\small $1926$} & 
		{\small $203$} & {\small $1568$} & {\small 2h 
		14m 25s} & {\small $100.00\%$} & {\small 
		$404$} & {\small $73$} & {\small $309$} & 
		{\small 46m 31s} & {\small $100.00\%$} & 
		{\small $175$} & {\small $57$} & {\small 
		$129$} & {\small 20m 11s} \\
		\multirow{-4}{*}{\cellcolor{myyellow!20}$0.125$}
		 &$10$ & {\small $100.00\%$} & {\small 
		$428$} & {\small $95$} & {\small $427$} & 
		{\small 1h 39m 31s} & {\small $100.00\%$} & 
		{\small $151$} & {\small $53$} & {\small 
		$141$} & {\small 57m 32s} & {\small 
		$100.00\%$} & {\small $80$} & {\small $39$} 
		& {\small $62$} & {\small 28m 33s} \\
		\hline
		\rowcolor{gray!10} \cellcolor{mygreen!20} & 
		$4$ & {\small $100.00\%$} & {\small 
		$19299$} & {\small $295$} & {\small 
		$15446$} & {\small 6h 13m 24s} & {\small 
		$100.00\%$} & {\small $1397$} & {\small 
		$60$} & {\small $885$} & {\small 40m 5s} & 
		\cellcolor{mygreen!20}
		{\small $\mathbf{100.00\%}$} & 
		\cellcolor{mygreen!20} {\small 
		$\mathbf{766}$} & \cellcolor{mygreen!20}
		{\small $\mathbf{87}$} & 
		\cellcolor{mygreen!20} {\small 
		$\mathbf{425}$} & \cellcolor{mygreen!20} 
		{\small \textbf{16m 
		41s}} \\
		\cellcolor{mygreen!20} & $6$ & {\small 
		$100.00\%$} & {\small $4843$} & {\small 
		$280$} & {\small $3679$} & {\small 2h 24m 
		7s} & {\small $100.00\%$} & {\small $763$} & 
		{\small $66$} & {\small $446$} & {\small 35m 
		24s} & {\small $100.00\%$} & {\small $401$} 
		& {\small $81$} & {\small $242$} & {\small 
		32m 29s} \\
		\rowcolor{gray!10} \cellcolor{mygreen!20} & 
		$8$ & {\small $100.00\%$} & {\small $1919$} 
		& {\small $208$} & {\small $1567$} & {\small 
		2h 9m 59s} & {\small $100.00\%$} & {\small 
		$404$} & {\small $73$} & {\small $309$} & 
		{\small 45m 48s} & {\small $100.00\%$} & 
		{\small $193$} & {\small $68$} & {\small 
		$144$} & {\small 24m 16s} \\
		\multirow{-4}{*}{\cellcolor{mygreen!20}$0$} 
		&$10$ & {\small $100.00\%$} & {\small 
		$486$} & {\small $102$} & {\small $475$} & 
		{\small 1h 41m 3s} & {\small $100.00\%$} & 
		{\small $217$} & {\small $55$} & {\small 
		$192$} & {\small 1h 2m 11s} & {\small 
		$100.00\%$} & {\small $121$} & {\small $50$} 
		& {\small $91$} & {\small 30m 53s} \\
		\hline
	\end{tabular}
}
\end{table}


\begin{figure}[t]
	\begin{subfigure}[b]{.5\linewidth}
\begin{tikzpicture}[scale=0.75]
\begin{semilogyaxis}[xlabel={Analyzed Input 
Space}, 
xtick={0,10,20,30,40,50,60,70,80,90,100}, 
ylabel={Analysis Time}, 
ymax=99000,
log number format basis/.code 2 
args={$\pgfmathparse{#1^(#2)}\mathsf{\pgfmathprintnumber{\pgfmathresult}}$s},
]
\addplot[scatter,only marks,  
point meta=explicit 
symbolic,scatter/classes={
nosym58={mark=pentagon,red,mark size=1.25pt},
nosym5={mark=pentagon*,red,mark size=1.25pt},
nosym258={mark=pentagon,orange,mark 
	size=1.25pt},
nosym25={mark=pentagon*,orange,mark 
size=1.25pt},
nosym1258={mark=pentagon,yellow!70!orange,mark
	size=1.25pt},
nosym125={mark=pentagon*,yellow!70!orange,mark
 size=1.25pt},
nosym08={mark=pentagon,green!70!black,mark 
	size=1.25pt},
nosym0={mark=pentagon*,green!70!black,mark 
size=1.25pt},
sym58={mark=triangle,red,mark size=1.5pt},
sym5={mark=triangle*,red,mark size=1.5pt},
sym258={mark=triangle,orange,mark size=1.5pt},
sym25={mark=triangle*,orange,mark size=1.5pt},
sym1258={mark=triangle,yellow!70!orange,mark size=1.5pt},
sym125={mark=triangle*,yellow!70!orange,mark size=1.5pt},
sym08={mark=triangle,green!70!black,mark size=1.5pt},
sym0={mark=triangle*,green!70!black,mark size=1.5pt},
deep58={mark=Mercedes star,red,mark size=1.5pt},
deep5={mark=star,red,mark size=1.5pt},
deep258={mark=Mercedes star,orange,mark size=1.5pt},
deep25={mark=star,orange,mark size=1.5pt},
deep1258={mark=Mercedes star,yellow!70!orange,mark size=1.5pt},
deep125={mark=star,yellow!70!orange,mark size=1.5pt},
deep08={mark=Mercedes star,green!70!black,mark size=1.5pt},
deep0={mark=star,green!70!black,mark size=1.5pt}
}]
table[meta=label] {
x			y			label


15.28	8 		nosym5
17.01	51		nosym5
51.39	722		nosym5
79.86	2055		nosym5


15.28	22		nosym58
17.01	63		nosym58
51.39	1367		nosym58
79.86	3714		nosym58


59.09	3272	nosym25
83.77	2239	nosym25
96.07	4049	nosym25
99.54	5720	nosym25


59.09	3291		nosym258
83.77	8390		nosym258
96.07	15229		nosym258
99.54	18214		nosym258


97.13	12828	nosym125
99.83	12180	nosym125
99.98	8065	nosym125
100.00	5971	nosym125


97.13	35160		nosym1258
99.83	31211		nosym1258
99.98	16798		nosym1258
100.00	17130		nosym1258


100.00	22404	nosym0
100.00	8647	nosym0
100.00	7799	nosym0
100.00	6063	nosym0


100.00	46800		nosym08
100.00	32927		nosym08
100.00	22529		nosym08
100.00	15543		nosym08


58.33	86		sym5
69.10	341		sym5
82.64	755		sym5
93.06	2552	sym5


58.33	151		sym58
69.10	412		sym58
82.64	1445		sym58
93.06	3368		sym58


95.94	3271	sym25
98.68	1411	sym25
99.72	3813	sym25
99.98	3733	sym25


95.94	1812		sym258
98.68	3057		sym258
99.72	3953		sym258
99.98	5954		sym258


99.99	2866	sym125
100.00	1607	sym125
100.00	2791	sym125
100.00	3452	sym125


99.99	6385		sym1258
100.00	3058		sym1258
100.00	4205		sym1258
100.00	7467		sym1258


100.00	2405	sym0
100.00	2124	sym0
100.00	2748	sym0
100.00	3731	sym0


100.00	3990		sym08
100.00	3037		sym08
100.00	4696		sym08
100.00	5756		sym08


69.79	198		deep5
80.56	473		deep5
91.32	1173		deep5
96.88	2619		deep5


98.26	869		deep25
99.70	805		deep25
99.98	1372		deep25
100.00	2096		deep25


99.99	1141		deep125
100.00	1014		deep125
100.00	1211		deep125
100.00	1713		deep125


100.00	1001	deep0
100.00	1949	deep0
100.00	1456	deep0
100.00	1853	deep0


69.79	237		deep58
80.56	1386		deep58
91.32	1648		deep58
96.88	2102		deep58


98.26	1235	deep258
99.70	2082	deep258
99.98	2545	deep258
100.00	4098	deep258


99.99	3079		deep1258
100.00	2873		deep1258
100.00	3372		deep1258
100.00	4589		deep1258


100.00	2278		deep08
100.00	2306	deep08
100.00	2900		deep08
100.00	3051		deep08

};
\draw[red] ({rel axis cs:0,0}|-{axis cs:0,46800}) -- 
({rel 
axis  
cs:1,0}|-{axis cs:0,46800});
\draw[green] ({rel axis cs:0,0}-|{axis cs:100,0}) -- ({rel 
	axis cs:0,1}-|{axis cs:100,0});
\end{semilogyaxis}
\end{tikzpicture}
		\caption{}
		\label{fig:japaneseFA}
	 \end{subfigure}%
  \begin{subfigure}[b]{.5\linewidth}
\begin{tikzpicture}[scale=0.75]
\begin{semilogyaxis}[xlabel={Analyzed Input 
Space}, 
xtick={90,100}, 
ylabel={Analysis Time}, 
log number format basis/.code 2 
args={$\pgfmathparse{#1^(#2)}\mathsf{\pgfmathprintnumber{\pgfmathresult}}$s},
]
\addplot[scatter,only marks,  
point meta=explicit 
symbolic,scatter/classes={
nomark={mark=none,mark size=1.25pt},
nosym58={mark=pentagon,red,mark size=1.25pt},
nosym5={mark=pentagon*,red,mark size=1.25pt},
nosym258={mark=pentagon,orange,mark 
	size=1.25pt},
nosym25={mark=pentagon*,orange,mark 
size=1.25pt},
nosym1258={mark=pentagon,yellow!70!orange,mark
	size=1.25pt},
nosym125={mark=pentagon*,yellow!70!orange,mark
 size=1.25pt},
nosym08={mark=pentagon,green!70!black,mark 
	size=1.25pt},
nosym0={mark=pentagon*,green!70!black,mark 
size=1.25pt},
sym58={mark=triangle,red,mark size=1.5pt},
sym5={mark=triangle*,red,mark size=1.5pt},
sym258={mark=triangle,orange,mark size=1.5pt},
sym25={mark=triangle*,orange,mark size=1.5pt},
sym1258={mark=triangle,yellow!70!orange,mark size=1.5pt},
sym125={mark=triangle*,yellow!70!orange,mark size=1.5pt},
sym08={mark=triangle,green!70!black,mark size=1.5pt},
sym0={mark=triangle*,green!70!black,mark size=1.5pt},
deep58={mark=Mercedes star,red,mark size=1.5pt},
deep5={mark=star,red,mark size=1.5pt},
deep258={mark=Mercedes star,orange,mark size=1.5pt},
deep25={mark=star,orange,mark size=1.5pt},
deep1258={mark=Mercedes star,yellow!70!orange,mark size=1.5pt},
deep125={mark=star,yellow!70!orange,mark size=1.5pt},
deep08={mark=Mercedes star,green!70!black,mark size=1.5pt},
deep0={mark=star,green!70!black,mark size=1.5pt}
}]
table[meta=label] {
x			y			label






96.07	4049	nosym25
99.54	5720	nosym25




99.98	8065	nosym125
100.00	5971	nosym125




100.00	8647	nosym0
100.00	7799	nosym0
100.00	6063	nosym0




93.06	2552	sym5


93.06	3368	sym58


95.94	3271		sym25
98.68	1411		sym25
99.72	3813		sym25
99.98	3733		sym25


95.94	1812	sym258
98.68	3057	sym258
99.72	3953	sym258
99.98	5954	sym258


99.99	2866		sym125
100.00	1607		sym125
100.00	2791		sym125
100.00	3452		sym125


99.99	6385	sym1258
100.00	3058	sym1258
100.00	4205	sym1258
100.00	7467	sym1258


100.00	2405	sym0
100.00	2124	sym0
100.00	2748	sym0
100.00	3731	sym0


100.00	3990	sym08
100.00	3037	sym08
100.00	4696	sym08
100.00	5756	sym08


91.32	1173	deep5
96.88	2619	deep5


99.98	1372	deep25
100.00	2096	deep25


99.99	1141	deep125
100.00	1014	deep125
100.00	1211	deep125
100.00	1713	deep125


100.00	1001	deep0
100.00	1949	deep0
100.00	1456	deep0
100.00	1853	deep0


91.32	1648	deep58
96.88	2102	deep58


98.26	1235	deep258
99.70	2082	deep258
99.98	2545	deep258
100.00	4098	deep258


99.99	3079	deep1258
100.00	2873	deep1258
100.00	3372	deep1258
100.00	4589	deep1258


100.00	2278	deep08
100.00	2306	deep08
100.00	2900	deep08
100.00	3051	deep08

90 1000 nomark

};
\draw[red] ({rel axis cs:0,0}|-{axis cs:0,46800}) -- 
({rel 
axis  
cs:1,0}|-{axis cs:0,46800});
\draw[green] ({rel axis cs:0,0}-|{axis cs:100,0}) -- ({rel 
	axis cs:0,1}-|{axis cs:100,0});
\end{semilogyaxis}
\end{tikzpicture}
	\caption{Zoom on  $90.00\% \leq \textsc{input}$ 
	and $1000s \leq 
	\textsc{time} \leq 1000s$}
	\label{fig:japaneseFB}
	\end{subfigure}
\caption{Comparison of Different
	Analysis Configurations (Japanese Credit
	Screening)}\label{fig:japaneseF}
\end{figure}
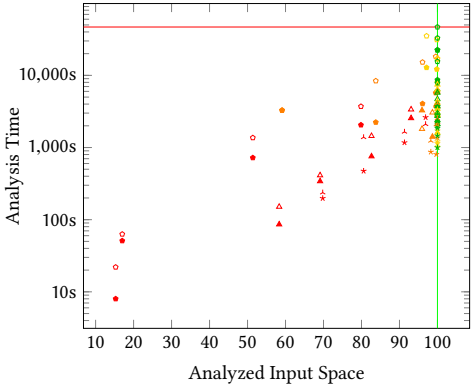
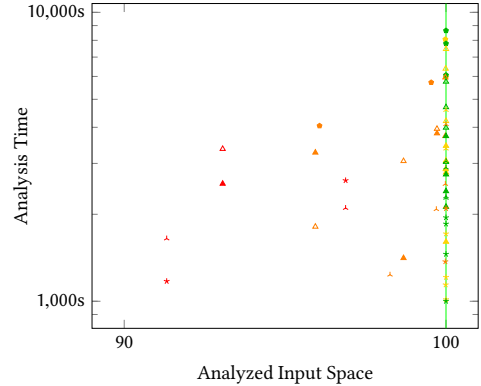

Table~\ref{tbl:japanese12T} shows the results of 
the analysis for different budget configurations and 
choices for the domain used for the forward 
pre-analysis.
The best configuration in terms of input-space 
coverage and analysis running time
is 
highlighted.
The symbol next to each domain name introduces 
the marker used in the
scatter plot of Figure~\ref{fig:japaneseFA}, which 
visualizes the coverage
and running time. Figure~\ref{fig:japaneseFB} 
zooms on $90.00\% \leq \textsc{input}$ 
and $1000s \leq 
\textsc{time} \leq 1000s$.

Overall, we observe that \emph{the more precise
  \textsc{symbolic} and \textsc{deeppoly} domains 
  boost input coverage}, most noticeably
  for configurations with a larger $\lowerbound$. 
  This
additional precision does not always result in longer running
times. In fact, a more 
precise pre-analysis 
often reduces the overall running time. This is 
because 
the pre-analysis is able to prove that more 
partitions are already fair without requiring them to 
go 
through the backward analysis (cf. columns 
$\size{\text{F}}$).

Independently of the chosen domain for the forward 
pre-analysis, as expected, \emph{a larger 
$\upperbound$ or a smaller $\lowerbound$
  increase precision.} Increasing $\upperbound$ or
$\lowerbound$ typically reduces the number of 
completed
partitions (cf. columns $\size{\text{C}}$). 
Consequently, partitions tend to be more complex,
requiring both forward and backward analyses. Since the backward
analysis tends to dominate the running time, more partitions generally
increase the running time (when comparing 
configurations with
similar coverage). Based on our experience, the optimal
budget largely depends on the analyzed model.

\begin{table}[t]
	\caption{Comparison of Different
		Analysis Configurations (Japanese Credit
		Screening) --- $4$ CPUs}\label{tbl:japanese4T}
\centering
\resizebox{\textwidth}{!}{%
	\begin{tabular}{|c|c|cc|cc|c|cc|cc|c|cc|cc|c|}
		\hline
		\multirow{2}{*}{$\lowerbound$} & 
		\multirow{2}{*}{\textsc{U}} &     
		\multicolumn{5}{c|}{$\usemark{pentagon}$ 
		\textsc{boxes}} &     
		\multicolumn{5}{c|}{$\usemark{triangle}$ 
		\textsc{symbolic}} &     
		\multicolumn{5}{c|}{$\usemark{Mercedes 
		star}$ \textsc{deeppoly}} \\
		& & \textsc{input} & $\size{\text{C}}$ & 		
		\multicolumn{2}{c|}{$\size{\text{F}}$} & 		
		\textsc{time} &         \textsc{input} & 
		$\size{\text{C}}$ &         		
		\multicolumn{2}{c|}{$\size{\text{F}}$} 
		&         		\textsc{time} & 		\textsc{input} 
		& $\size{\text{C}}$ & 		
		\multicolumn{2}{c|}{$\size{\text{F}}$} & 		
		\textsc{time} \\
		\hline
		\rowcolor{gray!10} \cellcolor{myred!20} & $4$ 
		& {\small $15.28\%$} & {\small $44$} & 
		{\small $0$} & {\small $0$} & {\small 22s} & 
		{\small $58.33\%$} & {\small $96$} & {\small 
		$8$} & {\small $26$} & {\small 2m 31s} & 
		{\small $69.79\%$} & {\small $85$} & {\small 
		$9$} & {\small $30$} & {\small 3m 57s} \\
		\cellcolor{myred!20} & $6$ & {\small 
		$17.01\%$} & {\small $40$} & {\small $5$} & 
		{\small $5$} & {\small 1m 3s} & {\small 
		$69.10\%$} & {\small $97$} & {\small $18$} & 
		{\small $45$} & {\small 6m 52s} & {\small 
		$80.56\%$} & {\small $131$} & {\small $26$} 
		& {\small $63$} & {\small 23m 6s} \\
		\rowcolor{gray!10} \cellcolor{myred!20} & $8$ 
		& {\small $51.39\%$} & {\small $96$} & 
		{\small $29$} & {\small $88$} & {\small 22m 
		47s} & {\small $82.64\%$} & {\small $128$} & 
		{\small $34$} & {\small $87$} & {\small 24m 
		5s} & {\small $91.32\%$} & {\small $117$} & 
		{\small $34$} & {\small $78$} & {\small 27m 
		28s} \\
		\multirow{-4}{*}{\cellcolor{myred!20}$0.5$} 
		&$10$ & {\small $79.86\%$} & {\small $109$} 
		& {\small $36$} & {\small $107$} & {\small 1h 
		1m 54s} & {\small $93.06\%$} & {\small 
		$104$} & {\small $36$} & {\small $92$} & 
		{\small 56m 8s} & {\small $96.88\%$} & 
		{\small $69$} & {\small $31$} & {\small $50$} 
		& {\small 35m 2s} \\
		\hline
		\rowcolor{gray!10} \cellcolor{myorange!20} & 
		$4$ & {\small $59.09\%$} & {\small $1147$} & 
		{\small $22$} & {\small $405$} & {\small 54m 
		51s} & {\small $95.94\%$} & {\small $715$} & 
		{\small $43$} & {\small $407$} & {\small 30m 
		12s} & {\small $98.26\%$} & {\small $488$} & 
		{\small $65$} & {\small $272$} & {\small 20m 
		35s} \\
		\cellcolor{myorange!20} & $6$ & {\small 
		$83.77\%$} & {\small $1757$} & {\small $80$} 
		& {\small $1149$} & {\small 2h 19m 50s} & 
		{\small $98.68\%$} & {\small $693$} & {\small 
		$73$} & {\small $400$} & {\small 50m 57s} & 
		{\small $99.70\%$} & {\small $322$} & {\small 
		$79$} & {\small $205$} & {\small 34m 42s} \\
		\rowcolor{gray!10} \cellcolor{myorange!20} & 
		$8$ & {\small $96.07\%$} & {\small $1129$} & 
		{\small $136$} & {\small $950$} & {\small 4h 
		13m 49s} & {\small $99.72\%$} & {\small 
		$289$} & {\small $62$} & {\small $232$} & 
		{\small 1h 5m 53s} & {\small $99.98\%$} & 
		{\small $153$} & {\small $56$} & {\small 
		$113$} & {\small 42m 25s} \\
		\multirow{-4}{*}{\cellcolor{myorange!20}$0.25$}
		 &$10$ & {\small $99.54\%$} & {\small $510$} 
		& {\small $92$} & {\small $497$} & {\small 5h 
		3m 34s} & {\small $99.98\%$} & {\small 
		$158$} & {\small $57$} & {\small $150$} & 
		{\small 1h 39m 14s} & {\small $100.00\%$} & 
		{\small $109$} & {\small $46$} & {\small 
		$85$} & {\small 1h 8m 18s} \\
		\hline
		\rowcolor{gray!10} \cellcolor{myyellow!20} & 
		$4$ & {\small $97.13\%$} & {\small $12398$} 
		& {\small $200$} & {\small $9491$} & {\small 
		9h 46m } & {\small $99.99\%$} & {\small 
		$1864$} & {\small $58$} & {\small $1188$} & 
		{\small 1h 46m 25s} & {\small $99.99\%$} & 
		{\small $1257$} & {\small $92$} & {\small 
		$670$} & {\small 51m 19s} \\
		\cellcolor{myyellow!20} & $6$ & {\small 
		$99.83\%$} & {\small $5919$} & {\small 
		$273$} & {\small $4460$} & {\small 8h 40m 
		11s} & {\small $100.00\%$} & {\small $697$} 
		& {\small $71$} & {\small $404$} & {\small 
		50m 58s} & {\small $100.00\%$} & {\small 
		$465$} & {\small $95$} & {\small $287$} & 
		{\small 47m 53s} \\
		\rowcolor{gray!10} \cellcolor{myyellow!20} & 
		$8$ & {\small $99.98\%$} & {\small $1331$} & 
		{\small $212$} & {\small $1158$} & {\small 4h 
		39m 58s} & {\small $100.00\%$} & {\small 
		$293$} & {\small $71$} & {\small $233$} & 
		{\small 1h 10m 5s} & {\small $100.00\%$} & 
		{\small $201$} & {\small $67$} & {\small 
		$151$} & {\small 56m 12s} \\
		\multirow{-4}{*}{\cellcolor{myyellow!20}$0.125$}
		 &$10$ & {\small $100.00\%$} & {\small 
		$428$} & {\small $94$} & {\small $427$} & 
		{\small 4h 45m 30s} & {\small $100.00\%$} & 
		{\small $211$} & {\small $55$} & {\small 
		$188$} & {\small 2h 4m 27s} & {\small 
		$100.00\%$} & {\small $121$} & {\small $50$} 
		& {\small $91$} & {\small 1h 16m 29s} \\
		\hline
		\rowcolor{gray!10} \cellcolor{mygreen!20} & 
		$4$ & {\small $100.00\%$} & {\small 
		$20631$} & {\small $296$} & {\small 
		$16611$} & {\small >13h} & {\small 
		$100.00\%$} & {\small $1424$} & {\small 
		$58$} & {\small $885$} & {\small 1h 6m 30s} 
		& \cellcolor{mygreen!20} {\small 
		$\mathbf{100.00\%}$} & 
		\cellcolor{mygreen!20} {\small 
		$\mathbf{911}$} & \cellcolor{mygreen!20}
		{\small $\mathbf{92}$} & 
		\cellcolor{mygreen!20} {\small 
		$\mathbf{502}$} & \cellcolor{mygreen!20} 
		{\small \textbf{37m 
		58s}} \\
		\cellcolor{mygreen!20} & $6$ & {\small 
		$100.00\%$} & {\small $6093$} & {\small 
		$296$} & {\small $4563$} & {\small 9h 8m 
		47s} & {\small $100.00\%$} & {\small $632$} 
		& {\small $72$} & {\small $371$} & {\small 
		50m 37s} & {\small $100.00\%$} & {\small 
		$403$} & {\small $85$} & {\small $247$} & 
		{\small 38m 26s} \\
		\rowcolor{gray!10} \cellcolor{mygreen!20} & 
		$8$ & {\small $100.00\%$} & {\small $1919$} 
		& {\small $211$} & {\small $1567$} & {\small 
		6h 15m 29s} & {\small $100.00\%$} & {\small 
		$378$} & {\small $79$} & {\small $287$} & 
		{\small 1h 18m 16s} & {\small $100.00\%$} & 
		{\small $174$} & {\small $65$} & {\small 
		$128$} & {\small 48m 20s} \\
		\multirow{-4}{*}{\cellcolor{mygreen!20}$0$} 
		&$10$ & {\small $100.00\%$} & {\small 
		$402$} & {\small $93$} & {\small $401$} & 
		{\small 4h 19m 3s} & {\small $100.00\%$} & 
		{\small $180$} & {\small $56$} & {\small 
		$154$} & {\small 1h 35m 56s} & {\small 
		$100.00\%$} & {\small $82$} & {\small $38$} 
		& {\small $63$} & {\small 50m 51s} \\
		\hline
	\end{tabular}
}
\end{table}

\paragraph{\textbf{RQ6: Leveraging Multiple 
CPUs.}}

To evaluate the effect of parallelizing the analysis using multiple
cores, we re-ran the analyses of RQ5 on 4~CPU cores instead of
12. Table~\ref{tbl:japanese4T} shows these results.
We observe the most significant increase in running 
time for $4$ cores for the \textsc{boxes} domain. 
On average, the running time increases by a factor 
of $2.6$. On the other hand, for the 
\textsc{symbolic} and 
\textsc{deeppoly} domains, the 
running time with 4
cores increases less drastically, on average by a 
factor of $1.6$ and $2$, respectively.
This is again explained
by the increased precision of the forward analysis; fewer partitions
require a backward pass, where parallelization is 
most effective.

\begin{toappendix}	
	
	\begin{table*}[t]
		\caption{Comparison of Different
			Analysis Configurations (Japanese Credit
			Screening) --- $24$ 
			vCPUs}\label{tbl:japanese24Tfull}
\centering
\resizebox{\textwidth}{!}{%
	\begin{tabular}{|c|c|cc|cc|c|cc|cc|c|cc|cc|c|}
		\hline
		\multirow{2}{*}{$\lowerbound$} & 
		\multirow{2}{*}{\textsc{U}} &     
		\multicolumn{5}{c|}{\textsc{boxes}} &     
		\multicolumn{5}{c|}{\textsc{symbolic}} &     
		\multicolumn{5}{c|}{\textsc{deeppoly}} \\
		& & \textsc{input} & $\size{\text{C}}$ & 		
		\multicolumn{2}{c|}{$\size{\text{F}}$} & 		
		\textsc{time} &         \textsc{input} & 
		$\size{\text{C}}$ &         		
		\multicolumn{2}{c|}{$\size{\text{F}}$} 
		&         		\textsc{time} & 		\textsc{input} 
		& $\size{\text{C}}$ & 		
		\multicolumn{2}{c|}{$\size{\text{F}}$} & 		
		\textsc{time} \\
		\hline
		\rowcolor{gray!10} \cellcolor{myred!20} & $4$ 
		& {\small $15.28\%$} & {\small $36$} & 
		{\small $0$} & {\small $0$} & {\small 7s} & 
		{\small $58.33\%$} & {\small $120$} & {\small 
		$7$} & {\small $34$} & {\small 3m 32s} & 
		{\small $69.79\%$} & {\small $75$} & {\small 
		$10$} & {\small $27$} & {\small 2m 43s} \\
		\cellcolor{myred!20} & $6$ & {\small 
		$17.01\%$} & {\small $39$} & {\small $6$} & 
		{\small $7$} & {\small 49s} & {\small 
		$69.10\%$} & {\small $80$} & {\small $21$} & 
		{\small $40$} & {\small 4m 19s} & {\small 
		$80.56\%$} & {\small $138$} & {\small $26$} 
		& {\small $65$} & {\small 12m 27s} \\
		\rowcolor{gray!10} \cellcolor{myred!20} & $8$ 
		& {\small $51.39\%$} & {\small $92$} & 
		{\small $30$} & {\small $86$} & {\small 12m 
		27s} & {\small $82.64\%$} & {\small $96$} & 
		{\small $32$} & {\small $76$} & {\small 14m 
		13s} & {\small $91.32\%$} & {\small $89$} & 
		{\small $36$} & {\small $61$} & {\small 13m 
		33s} \\
		\multirow{-4}{*}{\cellcolor{myred!20}$0.5$} 
		&$10$ & {\small $79.86\%$} & {\small $89$} 
		& {\small $34$} & {\small $89$} & {\small 29m 
		41s} & {\small $93.06\%$} & {\small $91$} & 
		{\small $37$} & {\small $83$} & {\small 47m 
		1s} & {\small $96.88\%$} & {\small $73$} & 
		{\small $33$} & {\small $52$} & {\small 30m } 
		\\
		\hline
		\rowcolor{gray!10} \cellcolor{myorange!20} & 
		$4$ & {\small $59.09\%$} & {\small $1320$} & 
		{\small $21$} & {\small $433$} & {\small 57m 
		33s} & {\small $95.94\%$} & {\small $656$} & 
		{\small $42$} & {\small $340$} & {\small 32m 
		38s} & {\small $98.26\%$} & {\small $488$} & 
		{\small $65$} & {\small $272$} & {\small 14m 
		11s} \\
		\cellcolor{myorange!20} & $6$ & {\small 
		$83.77\%$} & {\small $1600$} & {\small $80$} 
		& {\small $1070$} & {\small 1h 6m 58s} & 
		{\small $98.68\%$} & {\small $516$} & {\small 
		$61$} & {\small $287$} & {\small 18m 6s} & 
		{\small $99.70\%$} & {\small $286$} & {\small 
		$77$} & {\small $182$} & {\small 13m 14s} \\
		\rowcolor{gray!10} \cellcolor{myorange!20} & 
		$8$ & {\small $96.07\%$} & {\small $1148$} & 
		{\small $141$} & {\small $969$} & {\small 2h 
		41m 1s} & {\small $99.72\%$} & {\small 
		$260$} & {\small $58$} & {\small $207$} & 
		{\small 28m 57s} & {\small $99.98\%$} & 
		{\small $241$} & {\small $70$} & {\small 
		$175$} & {\small 29m 27s} \\
		\multirow{-4}{*}{\cellcolor{myorange!20}$0.25$}
		 &$10$ & {\small $99.54\%$} & {\small $409$} 
		& {\small $93$} & {\small $403$} & {\small 1h 
		38m 38s} & {\small $99.98\%$} & {\small 
		$213$} & {\small $50$} & {\small $189$} & 
		{\small 1h 16m 11s} & {\small $100.00\%$} & 
		{\small $88$} & {\small $42$} & {\small $68$} 
		& {\small 20m 25s} \\
		\hline
		\rowcolor{gray!10} \cellcolor{myyellow!20} & 
		$4$ & {\small $97.13\%$} & {\small $12449$} 
		& {\small $203$} & {\small $9519$} & {\small 
		3h 59m 27s} & {\small $99.99\%$} & {\small 
		$1101$} & {\small $59$} & {\small $685$} & 
		{\small 1h 2m 58s} & {\small $99.99\%$} & 
		{\small $892$} & {\small $86$} & {\small 
		$493$} & {\small 18m 4s} \\
		\cellcolor{myyellow!20} & $6$ & {\small 
		$99.83\%$} & {\small $4198$} & {\small 
		$266$} & {\small $3234$} & {\small 2h 31m 
		54s} & {\small $100.00\%$} & {\small $759$} 
		& {\small $73$} & {\small $461$} & {\small 
		51m 28s} & {\small $100.00\%$} & {\small 
		$563$} & {\small $108$} & {\small $344$} & 
		{\small 40m 35s} \\
		\rowcolor{gray!10} \cellcolor{myyellow!20} & 
		$8$ & {\small $99.98\%$} & {\small $1741$} & 
		{\small $217$} & {\small $1488$} & {\small 2h 
		16m 27s} & {\small $100.00\%$} & {\small 
		$308$} & {\small $67$} & {\small $242$} & 
		{\small 33m 14s} & {\small $100.00\%$} & 
		{\small $230$} & {\small $67$} & {\small 
		$167$} & {\small 22m 36s} \\
		\multirow{-4}{*}{\cellcolor{myyellow!20}$0.125$}
		 &$10$ & {\small $100.00\%$} & {\small 
		$582$} & {\small $97$} & {\small $564$} & 
		{\small 2h 16m 13s} & {\small $100.00\%$} & 
		{\small $180$} & {\small $56$} & {\small 
		$154$} & {\small 1h 5m 59s} & {\small 
		$100.00\%$} & {\small $80$} & {\small $39$} 
		& {\small $62$} & {\small 30m 18s} \\
		\hline
		\rowcolor{gray!10} \cellcolor{mygreen!20} & 
		$4$ & {\small $100.00\%$} & {\small 
		$16018$} & {\small $288$} & {\small 
		$12964$} & {\small 5h 3m 18s} & {\small 
		$100.00\%$} & {\small $1883$} & {\small 
		$63$} & {\small $1196$} & {\small 1h 52m 
		25s} & \cellcolor{mygreen!20} {\small 
		$\mathbf{100.00\%}$} & 
		\cellcolor{mygreen!20} {\small 
		$\mathbf{804}$} 
		& \cellcolor{mygreen!20} {\small 
		$\mathbf{90}$} & \cellcolor{mygreen!20} 
		{\small $\mathbf{442}$} & 
		\cellcolor{mygreen!20} {\small 
		\textbf{19m 47s}} \\
		\cellcolor{mygreen!20} & $6$ & {\small 
		$100.00\%$} & {\small $4675$} & {\small 
		$279$} & {\small $3503$} & {\small 3h 2m 
		30s} & {\small $100.00\%$} & {\small $632$} 
		& {\small $71$} & {\small $371$} & {\small 
		38m 3s} & {\small $100.00\%$} & {\small 
		$302$} & {\small $75$} & {\small $189$} & 
		{\small 19m 51s} \\
		\rowcolor{gray!10} \cellcolor{mygreen!20} & 
		$8$ & {\small $100.00\%$} & {\small $1609$} 
		& {\small $217$} & {\small $1382$} & {\small 
		2h 7m 9s} & {\small $100.00\%$} & {\small 
		$326$} & {\small $67$} & {\small $252$} & 
		{\small 1h 12s} & {\small $100.00\%$} & 
		{\small $194$} & {\small $68$} & {\small 
		$148$} & {\small 26m 9s} \\
		\multirow{-4}{*}{\cellcolor{mygreen!20}$0$} 
		&$10$ & {\small $100.00\%$} & {\small 
		$463$} & {\small $99$} & {\small $460$} & 
		{\small 2h 12m 12s} & {\small $100.00\%$} & 
		{\small $217$} & {\small $55$} & {\small 
		$192$} & {\small 1h 13m 55s} & {\small 
		$100.00\%$} & {\small $130$} & {\small $48$} 
		& {\small $98$} & {\small 50m 10s} \\
		\hline
	\end{tabular}
}
	\end{table*}


\subsection{RQ6: Leveraging Multiple 
	CPUs.}

Table~\ref{tbl:japanese24Tfull} shows the results of 
the experiment with the Japanese Credit Screening 
dataset on 24~vCPU.

\end{toappendix}

The appendix includes 
the same experiment on 24~vCPUs
(see Table
~\ref{tbl:japanese24Tfull}).

\section{Related Work}
\label{sect:relatedWork}



Significant progress has been made on testing and verifying
machine-learning models. We focus on fairness, 
safety, and robustness
properties in the following, especially of deep 
neural networks.

\paragraph{\textbf{Fairness Criteria.}}
There are countless fairness definitions in the 
literature. In this paper, we focus on causal fairness 
(specifically the fairness notion considered by 
Galhotra et al.
\cite{Galhotra17}) and compare here with the most 
popular and related notions.

\looseness=-1
\emph{Demographic parity} or \emph{group 
fairness} \cite{Feldman15} is the 
most 
common non-causal notion of fairness. It states 
that 
individuals with different values of
sensitive features, hence belonging to different groups, should have
the same probability of being predicted to the positive class. For
example, a loan system satisfies group fairness with respect to gender
if male and female applicants have equal probability of getting
loans. If unsatisfied, this notion is also referred to 
as \emph{disparate impact}. Our notion of fairness 
is stronger, as it imposes fairness on every pair of 
individuals
that differ only in sensitive features. 
A classifier that satisfies group fairness does not
necessarily satisfy causal fairness, because there 
may still exist 
pairs of individuals on which the classifier exhibits 
bias.

Another group-based notion of fairness is 
\emph{equality of opportunity} \cite{Hardt16}. It 
states that 
\emph{qualified} individuals with
different values of sensitive features should have equal probability
of being predicted to the positive class. For a loan system, this
means that male and female applicants who are qualified to receive loans
should have an equal chance of being approved.
By imposing fairness on every qualified pair of individuals that
differ only in sensitive features, we can generalize causal fairness
to also concern both prediction and actual results. We can then adapt
our technique to consider only the part of the input space that
includes qualified individuals.


Other causal notions of fairness 
\cite[etc.]{Kilbertus17,Kusner17,Nabi18,Chiappa19} 
require additional knowledge in the form of a 
\emph{causal model}. A causal model can drive the 
choice of the sensitive input(s) for our analysis.

\paragraph{\textbf{Testing and Verifying 
Fairness.}}
Galhotra et al. \cite{Galhotra17} proposed an approach, Themis, that
allows efficient fairness testing of software. Udeshi et
al. \cite{UdeshiAC18} designed an automated and directed
testing technique to generate discriminatory inputs for
machine-learning models. Tramer et al. \cite{TramerAGHHHJL17}
introduced the unwarranted-associations framework and instantiated it
in FairTest.
%
%
In contrast, our technique provides formal fairness guarantees.

Bastani et al.~\cite{Bastani18} used adaptive concentration inequalities to
design a scalable sampling technique for providing probabilistic
fairness guarantees for machine-learning models.  As mentioned in the
Introduction, our approach differs in that it gives definite (instead
of probabilistic) guarantees.  However, it might exclude partitions
for which the analysis is not exact.

Albarghouthi et al.~\cite{Albarghouthi17b} encoded fairness problems
as probabilistic program properties and developed an SMT-based
technique for verifying fairness of decision-making programs. As
discussed in the Introduction, this technique has been shown to scale
only up to neural networks with at most 3 inputs and a single hidden
layer with at most 2 nodes. In contrast, our approach is designed to be
perfectly parallel, and thus, is significantly more scalable.

A recent technique~\cite{RuossBalunovic2020} certifies individual
fairness of neural networks, which is a local 
property that coincides with
robustness within a particular distance metric. In particular,
individual fairness dictates that similar individuals should be treated
similarly. Our approach, however, targets certification of
neural networks for the global property of causal fairness.

For certain biased decision-making programs,
the program repair technique proposed by Albarghouthi et
al. \cite{Albarghouthi17a} can be used to repair their
bias. Albarghouthi and Vinitsky \cite{Albarghouthi19} further introduced
fairness-aware programming, where programmers can specify fairness
properties in their code for runtime checking.

\paragraph{\textbf{Robustness of Deep Neural 
Networks.}} Robustness is a
desirable property for traditional software
\cite{Chaudhuri12,Goubault13,Majumdar09}, especially control
systems. Deep neural networks are also expected to be robust. However,
research has shown that deep neural networks are not robust to small
perturbations of their inputs \cite{Szegedy13} and can even be easily
fooled \cite{Nguyen15}. Subtle imperceptible perturbations of inputs,
known as adversarial examples, can change their prediction
results. Various algorithms
\cite{Goodfellow14,Tabacof16,MadryMakelov2018,CarliniWagner2017-Robustness,ZhangCC19}
have been proposed that can effectively find adversarial
examples. Research on developing defense mechanisms against
adversarial examples
\cite{Goodfellow14,Huang15,Mirman18,CarliniWagner2017-Robustness,Cornelius2019,EngstromIlyas2018,AthalyeCarlini2018-Gradients,CarliniWagner2017-Bypassing,CarliniWagner2016,MirmanSV2019}
is also active.  Causal fairness is a special form
of robustness in the sense that neural networks are expected to be
\emph{globally} robust with respect to their sensitive features.

\paragraph{\textbf{Testing Deep Learning 
Systems.}} Multiple frameworks
have been proposed to test the robustness of deep learning
systems. Pei et al. \cite{PeiCYJ17} proposed the first whitebox
framework for testing such systems. They used neuron coverage to
measure the adequacy of test inputs. Sun et al. \cite{SunWu2018}
presented the first concolic-testing
\cite{SenMarinov2005,GodefroidKlarlund2005} approach for neural
networks. Tian et al. \cite{TianPei2018} and Zhang et
al. \cite{ZhangZhang2018} proposed frameworks for testing autonomous
driving systems. Gopinath et al. \cite{GopinathWang2018} used symbolic
execution \cite{Clarke1976,King1976}. Odena et
al. \cite{OdenaOlsson2019} were the first to develop coverage-guided
fuzzing for neural networks. Zhang et al. \cite{ZhangCC19} proposed a
blackbox-fuzzing technique to test their robustness.

\paragraph{\textbf{Formal Verification of Deep 
Neural Networks.}} Formal
verification of deep neural networks has mainly focused on safety
properties. However, the scalability of such techniques for verifying
large real-world neural networks is limited. Early work
\cite{Pulina10} applied abstract interpretation to verify a neural
network with six neurons. Recent work
\cite{Katz17,Gehr18,Singh19,WangPei2018-SecurityAnalysis,Huang17}
significantly improves scalability. Huang et al. 
\cite{Huang17}
proposed a framework that can verify local robustness of neural
networks based on SMT techniques \cite{BarrettTinelli2018}.  Katz et
al. \cite{Katz17} developed an efficient SMT solver for neural
networks with \textsc{ReLU} activation functions. Gehr et
al. \cite{Gehr18} traded precision for scalability and proposed a
sound abstract interpreter that can prove local robustness of
realistic deep neural networks. Singh et al. \cite{Singh19} proposed
the \textsc{deeppoly} domain for certifying robustness of neural
networks. Wang et al. \cite{WangPei2018-SecurityAnalysis} are the
first to use symbolic interval arithmetic to prove security properties
of neural networks.












%

\section{Conclusion and Future Work}

We have presented an automated, perfectly parallel analysis for certifying fairness of
neural networks. The analysis is configurable to support a wide range of use cases
throughout the development lifecycle of neural networks: ranging from short sanity checks
during development to formal fairness audits before deployments.

In future work, we plan to extend our technique in various ways, for instance, by
automatically tuning parameters (such as the upper bound $\upperbound$) during the
analysis or by feeding analysis results to other tools. Such tools may be used to provide
probabilistic fairness guarantees for partitions that could not be certified or repair
networks by eliminating bias.


%

\bibliography{bibliography}


\end{document}